\documentclass[twoside,11pt]{article}

\usepackage{jmlr2e}
\usepackage{amsmath,amssymb,amsfonts,mathtools}
\usepackage{mathtools}
\usepackage{graphicx}
\usepackage{booktabs}
\usepackage{hyperref}
\usepackage{array}
\usepackage{pifont}
\usepackage{color}
\usepackage{subcaption}
\usepackage{float}
\usepackage{multirow}
\usepackage{placeins}
\usepackage{booktabs}

\newcommand{\R}{\mathbb{R}}

\newcommand{\cmark}{\ding{51}}
\newcommand{\xmark}{\ding{55}}

\ShortHeadings{Hierarchical Partial-Order Models}{Li, Nicholls, Lee, Jiang}
\firstpageno{1}

\def\L{\mathcal{L}}
\def\P{\mathcal{P}}

\def\B{\mathcal{B}}

\def\M{\mathcal{M}}
\def\H{\mathcal{H}}
\def\A{\mathcal{A}}
\def\C{\mathcal{C}}

\def\U{U}
\def\u{u}

\usepackage{bm}
\usepackage{enumitem}

\usepackage{tikz}
\usepackage{tikz-network}
\usetikzlibrary{cd} 
\usetikzlibrary{backgrounds}
\pgfdeclarelayer{bg}    
\pgfsetlayers{bg,main}
\usetikzlibrary{fit,positioning,arrows.meta,bayesnet,shapes.geometric,chains,matrix,scopes,calc,decorations.markings}
\tikzstyle{param}=[circle, minimum size = 0.7cm, thick, draw=black!100, fill = gray!10, node distance = 0.5cm]
\tikzstyle{data}=[rectangle, minimum size = 0.7cm, thick, draw =black!100, node distance = 0.5cm]
\tikzstyle{model}=[rectangle, minimum size = 1cm, thick, draw=black!100, node distance = 0.5cm]

\definecolor{orange}{rgb}{0.8,0.6,0.0}
\definecolor{purple}{rgb}{0.4,0,0.8}
\definecolor{brightgreen}{rgb}{1.0,0.8,1.0}
\definecolor{indigo}{rgb}{0.1,0,1.0}
\definecolor{Orange}{rgb}{1.0,1.0,0.5}

\begin{document}

\title{Hierarchical Partial-Order Models for Ranking}

\author{Dongqing Li  \email{dongqing.li@kellogg.ox.ac.uk} \\
\addr{
    Department of Statistics, University of Oxford, 
    \\ 24-29 St Giles, Oxford, OX1 3LB, UK}
  \AND Geoff K. Nicholls \email{nicholls@stats.ox.ac.uk}\\
  \addr{
    Department of Statistics, University of Oxford, 
    \\ 24-29 St Giles, Oxford, OX1 3LB, UK}
  \AND
  Jeong Eun Lee \email{kate.lee@auckland.ac.nz}\\
  \addr{
    Department of Statistics, University of Auckland, \\
    38 Princes St, Auckland 1010, New Zealand}
  \AND
  Chuxuan (Jessie) Jiang  \email{chuxuanj@gmail.com}\\ 
  \addr{
    Department of Statistics, University of Oxford, 
    \\ 24-29 St Giles, Oxford, OX1 3LB, UK}
}

\editor{TBC}

\maketitle

\begin{abstract}

Rank aggregation combines information from ordered lists ranking items by preference.  Classical parametric models for such data, including the Mallows and Plackett–Luce models, assume the orders concentrate around one or more complete consensus rankings. Recent work relaxes the total-order assumption by allowing the consensus structure to be a partial order (poset), allowing for incomparabilities in preferences.
However, in many applications preference data exhibit group structure. We introduce hierarchical partial order (HPO) models, which extend poset-based models to accommodate grouped data through a hierarchy of latent posets. This framework, which parallels mixture model extensions of the Mallows and Plackett-Luce models, enables principled sharing of information across groups while preserving partial-order structure. We show that the Plackett–Luce model and its hierarchical variants are special cases of HPO-models.
We develop a hierarchical clustering extension (HCPO) for unsupervised clustering in settings where group labels are unknown. Bayesian inference for the latent poset hierarchy is performed using Markov chain Monte Carlo methods.
Experiments on synthetic and real-world datasets, including pairwise acoustic preference data and LLM agent traces, demonstrate that the proposed HPO and HCPO models outperform existing approaches in both predictive performance and structural interpretability.
\end{abstract}


 
\newpage

\section{Introduction}

Rank data arise in many applied settings. They may record the preferences of human assessors over choice sets of comparable items \citep{GormleyMurphy2009,raman14crowdGrade,crispino19}, the outcomes of competition between animals \citep{Johnson02MonkeyJasa,nagy13pigeonPO,foerster16} or ranking of people \citep{mogapi09,nicholls11} by their perceived expertise or status. Action-traces of LLMs  \citep{li2026delinearizing}, expressing a temporal logic, can also be treated as rank data.
In meta-analysis, studies rank different sets of genes by effect size \citep{Li19CompareAggregation} and continuous multivariate expression-levels are sometimes projected to rank-order data before analysis \citep{eliiseussen22RankSelect,vitelli23PanCancerSelect} as the resulting models are more robust and parsimonious than their continuous counterparts. Although these data are generated in quite different ways, various forms of rank-aggregation are used to identify favored rankings and recover underlying structural constraints informing the rank-outcome.

Classical rank aggregation combines multiple ranked lists from diverse sources into a single consensus ranking. \cite{lin10RankAgReview, Li19CompareAggregation} review heuristic methods \citep{borda1781}, optimisation based methods \citep{kemeny59principles} and parametric distribution-based methods \citep{mallows1957,thurstone1927law}. There is some overlap. For example, an estimator based on Kemeny's approach with a given distance measure is the MLE of a Mallows model with that distance.
Distribution-based models for rank aggregation generally assume that the lists in the data are random permutations which vary around a central consensus order. The Mallows \citep{fligner86} and Plackett-Luce models
\citep{luce1959possible,Plackett1975Analysis,hunter04,guiver2009bayesian} come under this broad heading.  Hierarchical and clustering extensions \citep{meila16,vitelli2018,caron2014bayesian,MollicaTardella2017} do allow ranking lists to come in groups with different consensus orders and the time-label in timeseries models \citep{bradlow01PLtime,asfaw17,caron12time,glickman24} also enforces a kind of group structure. However, these models all fundamentally reconstruct complete orders.

The assumption that the unknown true preference structure is given by one or more total orders ranking all items in the universe of choices from best to worst is often unjustified: an assessor may regard a pair of items as incomparable; not equal but actually having no order. Even when it is justified, a model that allows for incomparable pairs may have a higher marginal likelihood if it is more parsimonious \citep{jiang23}. Recent work \citep{nicholls25AOAS} has relaxed the requirement that the consensus orders be total orders and instead parameterises the underlying preference structure using a partial order or \emph{poset}: preferences express a set of order relations which may be incomplete. 


We give a new Hierarchical model using Posets (HPO) to model rank data which come in labeled groups. This is well motivated: many data sets have an obvious group structure \citep{Johnson02MonkeyJasa,nagy13pigeonPO,crispino19}; hierarchical models share information across groups and are a principled way to apply shrinkage. The HPO model is arranged as a tree, with posets at the root and leaves. The leaf-posets parameterise assessor-preferences and the root or ``global''-poset parameterises shared preferences. The generative model for a random poset-hierarchy is defined using continuous latent variables which are mapped to posets: we first define a prior over the latent variables at the root; given the root variables, the latent variables at each leaf are conditionally independent of variables at other leaves; we get a random poset-hierarchy by mapping the latent variables at each leaf to a poset. A parameter controlling the correlation between the global and assessor-level variables allows us to control ``shrinkage'' over the space of posets. The model reduces to a hierarchical Plackett-Luce model when the poset dimension is set equal one. The basic latent-variable/poset-map setup is used in timeseries models in \cite{nicholls25AOAS}. We present HPO-models in the terminology of social choice theory: assessors rank items in choice sets and return preference orders.
However, the framework applies in the more general ranking applications we cited in the first paragraph. 

At inference the lists generated by each assessor inform ``their'' poset and jointly inform the global poset. Our analysis is Bayesian and we use MCMC to summarise the posterior. Our main contributions are the hierarchical (HPO, Section~\ref{sec:Hpo-labeled}) and clustering (HCPO, Section~\ref{sec:Hpo-cluster}) models themselves. The model definitions are accompanied by theory relating the models to Plackett-Luce, showing they are marginally consistent and that they parameterise any hierarchy of posets. We also give new observation models (weighted Queue-jumping and Frontier-softmax in Sections~\ref{sec:eta-weighted-QJ} and \ref{sec:front_softmax_qj} respectively) for random orders structured by posets. Model comparison on real data (Sections~\ref{sec:sound} and \ref{sec:ghana-sweet-potato}) using the WAIC \citep{Watanabe2010,VehtariGelmanGabry2017} shows our approach is favored.

The only existing work in the literature on hierarchical models for partial orders is \cite{arcagni2022complexity}. Those authors have three partial orders on the same set of items. This is the data; they wish to compute a central consensus poset. They define a distance between posets and minimise it to get something like a centre-of-mass. The setup there shares with ours the idea of partial orders
distributed around a central poset. However, our framework is based on a generative model and Bayesian inference, where theirs is loss-based and gives a point estimate. We analyse their data in Appendix~\ref{app:arcagni-fixed-leaf} and explain how our framework may be applied (in fact, is simpler) when the data are partial orders.

\section{Background and models}\label{sec:po-foundation-models}

\subsection{Choice sets and posets}
Let $\M=[M],\ [M]=\{1,\dots,M\}$ be the universe of items over which preferences can be given and suppose there are $M=|\M|$ in all. Let $\B_\M$ be the set of all subsets of $\M$ excluding the empty set and let $S\in\B_\M$ be a given {\it choice set} with $m=|S|$ elements.

A \emph{poset} $h$ is a choice set equipped with a \emph{partial order} $\succ_h$. We use partial orders to describe preferences over items in a choice set.
Poset $h=(\M,\succ_h)$ on $\M$ is a set with order relations $j_1\succ_h j_2$ for $j_1,j_2\in \M$. An example is shown in Figure~\ref{fig:PO-LE-example}.
\begin{figure}
    \centering
    \scalebox{0.65}{\begin{tikzpicture}
      \node (1) [param] {$1$};
      \node (2) [param, below=of 1, xshift=-2.5cm, yshift=-0.5cm] {$2$};
      \node (3) [param, below=of 2, yshift=-1.5cm] {$3$};
      \node (4) [param, below=of 1, xshift=2cm, yshift=-2.25cm] {$4$};
      \node (5) [param, below=of 3, xshift=2.5cm, yshift=-1cm] {$5$};

      \node (POh) [text height=1cm, below=of 5, yshift=0.68cm] {\LARGE Partial order $h$};
      
      \node (1q1) [text height=0.5cm, right=of 1, xshift=2.5cm, yshift=0.5cm] {$1$};
      \node (2q1) [text height=0.5cm, below=of 1q1, yshift=-0.1cm] {$2$};
      \node (3q1) [text height=0.5cm, below=of 2q1, yshift=-0.1cm ] {$3$};
      \node (4q1) [text height=0.5cm, below=of 3q1, yshift=-0.1cm ] {$4$};
      \node (5q1) [text height=0.5cm, below=of 4q1, yshift=-0.1cm ] {$5$};

      \node (1q2) [text height=0.5cm, right=of 1q1, xshift=0.5cm] {$1$};
      \node (2q2) [text height=0.5cm, below=of 1q2, yshift=-0.1cm ] {$2$};
      \node (3q2) [text height=0.5cm, below=of 2q2, yshift=-0.1cm ] {$4$};
      \node (4q2) [text height=0.5cm, below=of 3q2, yshift=-0.1cm ] {$3$};
      \node (5q2) [text height=0.5cm, below=of 4q2, yshift=-0.1cm ] {$5$};

      \node (1q3) [text height=0.5cm, right=of 1q2, xshift=0.5cm] {$1$};
      \node (2q3) [text height=0.5cm, below=of 1q3, yshift=-0.1cm ] {$4$};
      \node (3q3) [text height=0.5cm, below=of 2q3, yshift=-0.1cm ] {$2$};
      \node (4q3) [text height=0.5cm, below=of 3q3, yshift=-0.1cm ] {$3$};
      \node (5q3) [text height=0.5cm, below=of 4q3, yshift=-0.1cm ] {$5$};

      \node (LEs) [text height=1cm, below=of 5q2, yshift=1cm, xshift=0.5cm] {\LARGE Linear extensions $\L(h)$};
      
      \edge [->] {1} {2};
      \edge [->] {1} {4};
      \edge [->] {2} {3};
      \edge [->] {4} {5};
      \edge [->] {3} {5};
      \edge [dashed,->] {1} {3};
      \edge [dashed,->] {1} {5};
      \edge [dashed,->] {2} {5};
      \edge [dashed,->] {2} {5};
      
      \edge [->] {1q1} {2q1};
      \edge [->] {2q1} {3q1};
      \edge [->] {3q1} {4q1};
      \edge [->] {4q1} {5q1};

      \edge [->] {1q2} {2q2};
      \edge [->] {2q2} {3q2};
      \edge [->] {3q2} {4q2};
      \edge [->] {4q2} {5q2};

      \edge [->] {1q3} {2q3};
      \edge [->] {2q3} {3q3};
      \edge [->] {3q3} {4q3};
      \edge [->] {4q3} {5q3};

      \node (2s) [param, right=of 3, xshift=12cm, draw=blue] {$2$};
      \node (4s) [param, right=of 2s, xshift=3cm, draw=blue] {$4$};
      \node (5s) [param, below=of 2s, xshift=2.25cm, yshift=-1cm, draw=blue] {$5$};
      
      \edge [->] {2s} {5s};
      \edge [->] {4s} {5s};

      \node (2s1) [text height=0.5cm, right=of 4s, xshift=0.5cm, yshift=1.5cm] {$2$};
      \node (4s1) [text height=0.5cm, below=of 2s1, yshift=-0.1cm] {$4$};
      \node (5s1) [text height=0.5cm, below=of 4s1, yshift=-0.1cm ] {$5$};

      \node (2s2) [text height=0.5cm, right=of 2s1, xshift=0.5cm] {$4$};
      \node (4s2) [text height=0.5cm, below=of 2s2, yshift=-0.1cm ] {$2$};
      \node (5s2) [text height=0.5cm, below=of 4s2, yshift=-0.1cm ] {$5$};
      
      \edge [->] {2s1} {4s1};
      \edge [->] {4s1} {5s1};
      \edge [->] {2s2} {4s2};
      \edge [->] {4s2} {5s2};
      
      \node (SOh) [text height=1cm, above=of 5s, yshift=1.5cm] {\LARGE $S=\{2,4,5\}$};
      \node (SOLE) [text height=1cm, below=of 5s1, yshift=0.75cm, xshift=1cm,] {\LARGE $\L(h[S])$};
      \node (O) [text height=1cm, below=of 5s, yshift=0.65cm] {\LARGE Sub-order $h[S]$};
    \end{tikzpicture} }
    \caption{The example partial order $h$ (left) on $\M=\{1,2,3,4,5\}$ has three linear extensions (left-centre). Dashed lines in $h$ are relations implied by transitivity. Its suborder on $S=\{2,4,5\}$ (right-centre) has two linear extensions (right).}
    \label{fig:PO-LE-example}
\end{figure}
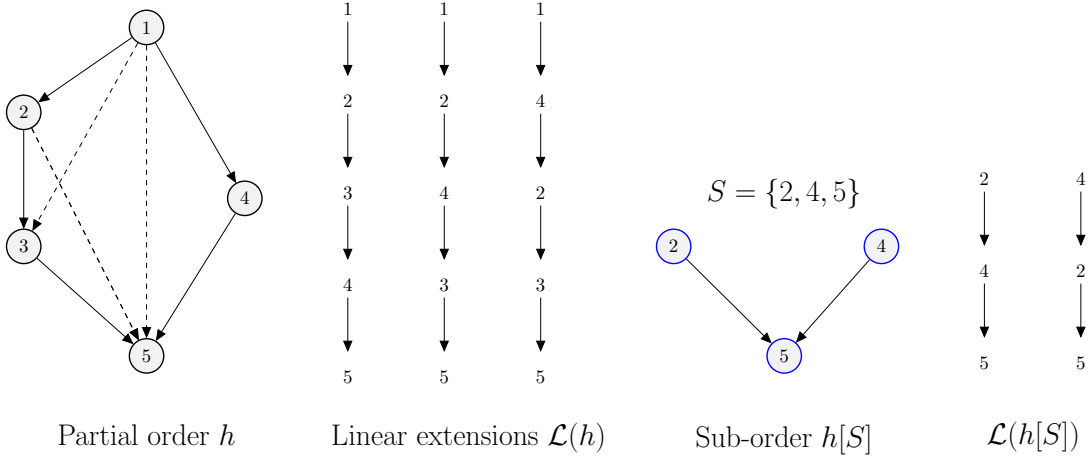
Relations are anti-symmetric (if $j_1\succ_h j_2$ then not $j_2 \succ_h j_1$) and transitive (if $j_1\succ_h j_2$ and $j_2\succ_h j_3$ then $j_1\succ_h j_3$) but need not be complete, so there may be pairs where neither $j_1\succ_h j_2$ nor $j_2\succ_h j_1$. 
We work with strict partial orders, so relations are irreflexive (no relations $j\succ_h j$, and no ties).
Denote by $\max(h)$ the \emph{max-set} of $h$; if $j\in \max(h)$ then $\{j'\in \M; j'\succ j\}$ is empty, so these are the highest ranked elements in $h$. 

If the choice set is restricted from $\M$ to $S\subset\M$, then preferences are given by the {\it suborder} $h[S]=(S,\succ_h)$ with relations $j_1\succ_h j_2$ involving pairs $j_1,j_2\in S$. 
Let $\H_S$ be the set of all posets on $S$.
If for all pairs $j_1,j_2\in S$ either $j_1\succ_h j_2$ or $j_2\succ_h j_1$ then $h$ is a {\it complete order}. 
Let $\C_S$ be the set of all complete orders of $S$. 
Complete orders are one to one with $\P_S$, the set of all permutations of the elements of $S$.

\subsection{Context-independent preference}\label{sec:context-independent-defn}

A generic observation $y=(S,\succ_y),\ y\in \C_S$ is a complete order on its choice set or equivalently
the permutation $y_{1:m}=(y_1,\dots y_m),\ y_{1:m}\in \P_S$ we get by indexing $y_i,\ i=1,\dots m$ so that $y_1\succ_y y_2\succ_y\dots\succ_y y_m$. 
There are two ways to realise an order: an assessor may be given a choice set $S$ and return a preference order considering only the elements of $S$, in which case $y$ is {\it realised on the choice set} $S$; alternatively they may make a preference order $y'\in\C_\M$ on the universe of choices and return the suborder $y=y'[S]$, thinning out any choices not in $S$ while retaining the order of those that remain. 
In this case, $y$ is {\it realised as a suborder} on $S$. 
If preferences are {\it context-independent}, then these two rules give the same observation models and otherwise, preferences are context-dependent. 
\begin{definition}\label{def:context-independence} (Context-independent preference) Let $\{p_S(y),\ y\in \C_S\},\ S\in\B_\M$ be a family of probability distributions over orders and let $y\sim p_\M(\cdot)$. 
If $y[S]\sim p_S(\cdot)$ for all $S\in\B_\M$, then the family expresses context-independent preference, and otherwise, preference is context-dependent.
\end{definition}
Context independence in Definition~\ref{def:context-independence} is weaker than context-independence in the sense of the ``Luce Axiom of Choice'' \citep{luce77}. We return to this in Section~\ref{sec:PL}. 
See Appendix~\ref{app:context-independent-models} for background on context independence and its relation to marginal consistency.

We focus on the case where the data are realised on a choice set.
Our models are context dependent so our methods do not in general apply to data realised as a suborder. They do apply specifically to top-$k$ data, where assessors just return their top-$k$ preferences in order, and they are context independent in the special case where they reduce to Plackett-Luce models. If we wanted to model data realised as a suborder we would follow \cite{vitelli2018} and treat the missing part of the order as missing data. See Appendix~\ref{app:PO_model_single_PO}.

\subsection{Noise-free observation model}\label{sec:noise-free-obs-model}


In this section we give a simple observation model. It makes the often unrealistic assumption that the data perfectly respect all preferences in the unknown underlying poset. However, the model underpins some more general observation models we give in Section~\ref{sec:lkd-noisy} and allows us to move on to Bayesian inference for HPO-models with a concrete likelihood in hand.  
\begin{definition}\label{def:linear-extension} 
A {\it linear extension} of a poset $h\in \H_S$ is any complete order $\ell\in \C_S$ satisfying $j_1\succ_h j_2\Rightarrow j_1\succ_\ell j_2$, so the linear extension ``completes'' the partial order $\succ_h$. Denote by $\L[h]$ the set of all linear extensions of $h$.
\end{definition}
Suppose observed orders are sampled uniformly at random from the linear extensions of a poset $h\in\H_S$.
Let $p_S(y|h)$ be the probability to realise $y\in \L[h]$.
The likelihood for $h$ is
\begin{equation}\label{eq:lkd-noise-free}
p_S(y|h)=|\L[h]|^{-1}\mathbb{I}_{y\in \L[h]}.
\end{equation}
This model is motivated in \cite{nicholls25AOAS} by a stochastic queue process in which $y$ is a draw from the equilibrium of a random walk on $\L[h]$. As shown in Appendix~\ref{app:context-independent-models}, the model in \eqref{eq:lkd-noise-free} expresses context-dependent preferences.
See also remarks in Appendix~\ref{app:PCMC-models} on relations to 
Pairwise choice Markov chain models \citep{ragain16_pcmc}.

Computing $|\L[h]|$ is \#P-complete \citep{brightwell1991counting}, so we cannot evaluate $p_S(y|h)$ for general $h\in \H_S$ and large $m=|S|$. 
Counting is feasible up to about $m = 40$ using $\tt lecount()$, which implements methods in \cite{koivisto16b}. 
For greater $m$-values, we use a repeated selection model (defined in Section~\ref{sec:front_softmax_qj}) which builds up $y_{1:m}$ one element at a time. This is chosen at random from the max-set of the suborder for the elements that remain, so random orders $y$ are still distributed over linear extensions, but not uniformly. The likelihood can be evaluated rapidly. 

In many data sets the choice set varies from one list to another. 
An assessor with preferences expressed by a poset $h = (\M, \succ_h)$ gives $N$ preference orders $Y = (Y_1, \dots, Y_N)$, where $Y_i\in \C_{S_i}$ orders the elements of choice set $S_i\in\B_\M$. For $i\in [N]$ let $m_i=|S_i|$. As lists, $Y_i=(Y_{i,1},\dots,Y_{i,m_i})$.
The assessor's preferences over $S_i$ are determined by the suborder $h[S_i] = (S_i,\succ_h)$. 
Assuming the lists are independent given $h$, the likelihood is
\begin{equation}\label{eq:po_lkd_full_noisefree}
   p_{S_{1:N}}(Y|h) = \prod_{i=1}^N p_{S_i}(Y_i|h[S_i]),
\end{equation}
where $p_{S_i}(Y_i|h[S_i])=|\L[h[S_i]]|^{-1}$ per Equation~\eqref{eq:lkd-noise-free}. 

There is enough here for simple Bayesian inference for posets from order data. For example, if we take a uniform prior over $\H_\M$ then the posterior is proportional to the likelihood in \eqref{eq:po_lkd_full_noisefree} and the MCMC algorithm given in \cite{muirwatt15}, which samples posets uniformly at random, could be used for sample-based Bayesian inference.

\subsection{The Plackett-Luce Model}\label{sec:PL}

The HPO model we write down in Section~\ref{sec:Hpo-labeled} uses the same Gumbel-construction as the Plackett-Luce model, and covariates enter in a similar way. In order to keep our presentation self-contained we give the standard setup here.

In a Plackett-Luce (PL) model with preference weights $\alpha_\M=(\alpha_1,\dots,\alpha_M)$, the weight for choice $i\in\M$ is $\alpha_i\in \R$ and the observation model for an ordering $y\in\C_S$ of a choice set $S\in\B_\M$ with $m$ elements is
\begin{equation}\label{eq:PL-def-full}
    p_S(y|\alpha_S)=\prod_{i=1}^m \frac{e^{\alpha_{y_i}}}{\sum_{i'=i}^m e^{\alpha_{y_{i'}}}},
\end{equation}
where $\alpha_S=(\alpha_j)_{j\in S}$.
We write $y\sim\mbox{PL}(\alpha_S;S)$. This is a sequential choice model in which the order is built up element by element: we can write
\begin{equation}\label{eq:PL-sequential}
p_S(y|\alpha_S)=\prod_{i=1}^{m-1} q_{y_{i:m}}(y_i|\alpha_S),
\end{equation}
where 
\begin{equation}\label{eq:PL-def-one-step}
    q_{y_{i:m}}(y_i|\alpha_S)=\frac{e^{\alpha_{y_i}}}{\sum_{i'=i}^m e^{\alpha_{y_{i'}}}}
\end{equation}
is the probability to choose $y_i$ from the choice set $\{y_{i},y_{i+1},\dots,y_{m}\}$.

We can add covariates to this model. 
For $j\in \M$, let $x_j = (x_{j,1}, \dots, x_{j,p})$ be a vector of $p$ covariates associated with item $j$, and let $X=(x_{j,c})_{j\in \M}^{c\in [p]}$ be the $M\times p$ matrix of all covariates. 
Let $\beta\in \R^p$ be a vector of effects. 
If we include an intercept in the covariate vector, then the vector of preference weights in \eqref{eq:PL-def-full} is $\alpha_\M = X\beta$.

The generative model for orders in the Plackett-Luce model can be given in terms of latent Gumbel random variables \citep{yellot77}. 
\begin{lemma} \citep{yellot77} \label{lem:PL-gumbel} Let $X_j\sim \mbox{Gumbel}(\alpha_j)$ be independent for $j\in \M$, where $\mbox{Gumbel}(\alpha_j)$ is a distribution with CDF $F_{\alpha_j}(g)=\exp(-\exp(-(g-\alpha_j))),\ g\in \R$. Let $X=(X_1,\dots,X_M)$ and
let $y(X)=(\M,\succ_X)$ be the corresponding complete order for the elements of $X$, that is $j_1\succ_X j_2\Leftrightarrow X_{j_1}>X_{j_2}$. \cite{yellot77} shows that $y(X)\sim \mbox{PL}(\alpha_\M;\M)$. 
\end{lemma}
We can think of $X_j$ as a random score for item $j$, biased by its preference weight $\alpha_j$; the random preference order $y(X)$ gives the items ordered by score. 

The Plackett-Luce model is context independent so we don't need to know whether the assessor realised their preference order on the choice set $S$ or made a full ranking on $\M$ and then thinned the list down to the suborder for $S$. See Appendix~\ref{app:PL-further-remarks} for further remarks on context independence in the Plackett-Luce model.


\section{Bayesian inference for a single partial order}
\label{sec:PO-single-inference}

In this section we give a prior for a single random poset, modifying the latent variable setup in \cite{nicholls25AOAS} to capture the Placket-Luce model as a special case.
In the next section we extend this prior to the hierarchical setting.


We parameterise $h$ using continuous latent variables and a product-order embedding \citep{dushnik1941}, inspired by models for random partial orders given by \cite{winkler1985random}. 
Let $\U\in \R^{M\times K}$ be a matrix of preference weights with one row $\U_{j,:}\in \R^K$ for each item $j\in\M$ in the universe of choices and one column $\U_{:,k}$ for each ``feature'' $k=1,\dots,K$.
For a pair of items $j_1,j_2\in\M$, our rule $h(\U) = (\M,\succ_\U)$ mapping $\U$ to a poset will have $j_1 \succ_\U j_2$, if and only if the two rows of weights satisfy $\U_{j_1,k} > \U_{j_2,k}$ for each $k=1,\dots,K$. We can write this map as an intersection over the orders of the columns of $U$. 
Let $h(U_{:,k}) = (\M,\succ_{k})$ be the complete order on column $k$ with 
\[
j_1\succ_k j_2 \Leftrightarrow \U_{j_1,k}>\U_{j_2,k}.
\]
This will almost surely be a complete order because $>$ is a complete order on the real entries in the $k$'th column of $\U$. We define $\succ_\U$ to be the set of order relations \[j_1\succ_U j_2 \Leftrightarrow j_1\succ_k j_2,\ \forall k=1,\dots,K.\] 
This is called an intersection order, and we write 
\begin{equation}\label{eq:UtoPO}
h(\U)=\bigcap_{k=1}^K h(\U_{:,k}), 
\end{equation}
as $h(\U)=(\M,\succ_\U)$  contains just the order relations shared by all $h(\U_{:,k}),\ k=1,\dots,K$. 

This latent variable setup makes it straightforward to add covariates.
In terms of the covariate notation in Section~\ref{sec:PL}, with $\alpha=X\beta$ an $M\times 1$ vector, let
\begin{equation}
    \gamma=\U+\alpha1^T_K,
\end{equation}
where $\alpha1^T_K$ is an outer product of $\alpha$ with a vector of $K$ ones.
All elements in a row, $\gamma_{j,:}=\U_{j,:}+\alpha_j 1^T_K$, are offset by the same $\alpha_j=x_{j}^T\beta$, so if we compute $h=h(\gamma)$ from the offset weights then a large positive effect for item $j$ tends to move $j$ up in the partial order.
We do not include an intercept among the covariates as that duplicates the degree of freedom corresponding to the mean of $\U_{j,:}$.

We get a prior over $h\in \H_\M$ by taking priors over $\U$ and $\beta$. 
We want the prior distribution to be uninformative of the \emph{depth} of the reconstructed partial order as hypotheses about the unknown true $h$ are often framed as statements about depth. 
The depth $d(h)\in \{1,2,\dots, M\}$ of a partial order is 
the number of elements in the longest complete suborder,
\[
d(h)=\max_{S\in\B_\M}\{|S|: h[S]\in \C_S\}.
\]
In order to control $d(h(\gamma))$, \cite{nicholls25AOAS} take the rows of the $\gamma$-matrix to be correlated multivariate normal variables, $\gamma_{j,:}\sim N(\alpha_j 1_K,\Sigma_\rho)$ independent for $j\in\M$. The covariance matrix $\Sigma_\rho$ has a constant unit diagonal $(\Sigma_\rho)_{k,k}=1$ and constant off diagonal $(\Sigma_\rho)_{k,k'}=\rho$ for $\rho\in [0,1)$ and $k\ne k'$. The correlation parameter $\rho$ is positive and controls the typical depth of $h(\gamma)$. 
When $\rho$ is close to one, the values $\gamma_{j,k}$ don't vary much with $k$, so the orders $\succ_k,\ k=1,\dots,K$ are all the same and $d(h(\eta))$ is close to $M$. When $\rho$ is small the preference weights $\eta_{j,k}$ and $\eta_{j,k'}$ are nearly independent so the orders $\succ_k$ and $\succ_{k'}$ share fewer order relations. \cite{nicholls25AOAS} show using simulation that taking $\rho\sim \mbox{Beta}(1,1/6)$ gives a marginal prior distribution for $h$ which is reasonably uninformative of depth.

Up this point we followed \cite{nicholls25AOAS}. In order to get a model for random partial orders nesting PL, we modify this prior. Let $G^{-1}(g)=-\log(-\log(g))$ be the inverse CDF of a standard Gumbel random variable, and let $\Phi$ be the CDF of a standard normal.
\begin{theorem} (Partial Order Model) \label{thm:po-prior-gumbel-U} For $\alpha$ and $\Sigma_\rho$ defined above, if we take\\[-0.25in]
    \begin{align} \label{eq:po-prior-gumbel-U}
    \U_{j,:}&\sim N(0,\Sigma_\rho),\quad\mbox{independent for each $j\in \M$,}\\
    \eta_{j,:}&=G^{-1}(\Phi(\U_{j,:}))+\alpha_j 1^T_K, \quad\mbox{and}
    \label{eq:po-prior-gumbel-eta}\\
    y&\sim p_\M(\cdot|h(\eta(\U,\beta))),\quad\mbox{noise free likelihood in \eqref{eq:lkd-noise-free},}
\end{align}
then $h(\eta_{:,k})\sim PL(\alpha;\M)$ for each $k=1,\dots,K$. 
In particular, if $K=1$ then $y\sim PL(\alpha;\M)$.
\end{theorem}
\begin{proof}
The CDF of a standard normal $\Phi$ is applied to each element of $\U_{j,:}$, so $\Phi(\U_{j,:})$ is a vector of correlated uniform random variables. 
Applying the inverse CDF of the Gumbel distribution to each element of this vector gives a vector of correlated standard Gumbel random variables, which we shift by $\alpha$ to get $\eta_{j,k}\sim \mbox{Gumbel}(\alpha_j),\ k=1,\dots, K$. 
They are independent for each $j\in \M$, so $h(\eta_{:,k})$ is a complete order distributed as $h(\eta_{:,k})\sim PL(\alpha;\M)$ by Lemma~\ref{lem:PL-gumbel}. If $K=1$ then $h(\eta(\U,\beta))=h(\eta_{:,1})$ is a complete order with only has one linear extension, $h(\eta_{:,1})$ itself, so with the noise-free likelihood in \eqref{eq:lkd-noise-free} we must observe $y=h(\eta_{:,1})$ and hence $y\sim PL(\alpha;\M)$.
\end{proof}

Equations~\eqref{eq:po-prior-gumbel-U} and \eqref{eq:po-prior-gumbel-eta} determine a prior distribution $\pi_\M(h|\rho,\beta)$ for $h\in \H_S$ for each choice set 
$S\in \B_\M$ with $m=|S|$. Suppose now $\U$ is an $m\times K$ matrix with one row for each element in $S$. Let 
\[
\eta(\U,\beta)=G^{-1}(\Phi(\U)) + X\beta\, 1^T_K
\]
be the $m\times K$ matrix with rows $\eta_{j,:},\ j\in S$. 
The random partial order $h(\eta(\U,\beta))$ has prior distribution 
\begin{equation}\label{eq:po-single-h-prior}
    \pi_S(h|\rho,\beta)=E_U( \mathbb{I}_{h(\eta(\U,\beta))=h} ),\quad h\in\H_S.
\end{equation}

Every partial order in $\H_\M$ has non-zero prior probability in this model.
\begin{corollary}\label{cor:po-support}
If $M\ge 4$ and $K\ge \lfloor M/2 \rfloor$, then $\pi_S(h|\rho,\beta)>0$ for any $0\le \rho<1,\ \beta\in\R^p$, $h\in \H_S$ and $S\in\B_\M$.
\end{corollary}
Corollary~\ref{cor:po-support} is a simple variant of Proposition~3 in \cite{nicholls25AOAS}, with Normal-$U$ replaced by Gumbel-$U$, and so the proof is essentially the same. It hinges on a result due to \cite{hiraguchi1951dimension}, who shows that any partial order on $M\ge 4$ elements can be written as the intersection of $\lfloor M/2\rfloor$ complete orders or fewer.

The latent variable setup \citep{winkler1985random} ensures that the family of prior distributions $\pi_S(\cdot|\rho,\beta),\ S\in \B_\M$ is marginally consistent. \cite{nicholls25AOAS} show this holds if $\beta=0_p$ in a timeseries model for random partial orders with context-dependent covariates. The result given here holds for all $\beta$ but the proof (given in Appendix~\ref{app:PO_model_single_PO}) is simpler because the covariates in \cite{nicholls25AOAS} are context-dependent.
\begin{corollary}\label{cor:po-MC}
The family of prior distributions $\pi_S(\cdot|\rho,\beta),\ S\in \B_\M$, is marginally consistent, that is if $h\sim \pi_\M(\cdot|\rho,\beta)$, then $h[S]\sim \pi_S(\cdot|\rho,\beta)$.
\end{corollary}

Suppose the data are realised on choice sets $S_i,\ i=1,\dots,N$, and $\M'=\cup_{i} S_i$ with $\M'\subseteq\M$ and $M'=|\M'|$. Items $j\in \M\setminus\M'$ are not present in any choice set so parameters $\U_{j,:}$ are dropped from the posterior (like PL, see remark at end of Appendix~\ref{app:PL-further-remarks}): they don't appear in the likelihood and by Corollary~\ref{cor:po-MC} we have all the prior marginals. We assume for simplicity that items in $\M'\subseteq\M$ are removed from the start so $\cup_{i} S_i=\M$.
The posterior parameters are $\beta\in \R^p,\rho\in [0,1)$ and $\U\in \R^{M\times K}$. 
The posterior is
\begin{align}
    \pi_{\M}(\rho,\U,\beta|Y)&\propto \pi_R(\rho)\,\pi_B(\beta)\,\pi(\U|\rho)\,p_{S_{1:N}}(Y|h(\eta(\U,\beta))),\label{eq:po-posterior-marginal}\\
    \pi(U|\rho) &= \prod_{j\in\M} N(\U_{j,:};0_K,\Sigma_\rho),\nonumber
\end{align}
where $p_{S_{1:N}}(Y|h(\eta(\U,\beta)))$ is given in \eqref{eq:po_lkd_full_noisefree}.
If we have samples $\rho^{(t)},\U^{(t)},\beta^{(t)},\ t = 1, \dots, T$ distributed according to $\pi(\rho,\U,\beta|Y)$, and we want samples from the marginal posterior over partial orders, then we simply set $h^{(t)}=h(\eta(\U^{(t)},\beta^{(t)})),\ t = 1, \dots, T$. 

When the data realised as a suborder we have to add auxiliary variables for the missing ranks, because our observation model is context dependent. The revised setup is outlined briefly in Appendix~\ref{app:PO_model_single_PO}.

In the analysis above, the number of columns $K$ is a fixed hyper-parameter of the prior over partial orders. It controls how ``expressive'' the paramerisation is. If we take $K$ at least $\lfloor M/2\rfloor$ then we can represent any poset in $\H_\M$. However, the parameter dimension increases linearly with $K$ and we would like to use the smallest value that gives a good fit to the data. \cite{jiang24} estimate $K$ using reversible jump MCMC. We extend this to the hierarchical setting. In the present single-poset setting the joint posterior is 
\begin{equation}\label{eq:po-posterior-vary-K}
\pi_{\M}(K,\rho,\U,\beta|Y)\propto\pi_K(K)\,\pi_R(\rho)\,\pi_B(\beta)\,\pi(\U|\rho,K)\,p(Y|h(\eta(\U,\beta))),
\end{equation}
where $\pi_K(K)$ is a prior for $K$ (geometric with mean chosen so that $\pi_K(K>M/2)\simeq 0.5$). The dimension of the $M\times K$ matrix $\U$ is now random. 

\section{A hierarchical model for grouped data}\label{sec:Hpo-labeled}

We now give a hierarchical model for rank-order data which come in labeled groups, one group for each assessor. The setup resembles \cite{crispino19} who give a hierarchical model for Mallows ranking. 
The model can be represented as a tree. It has a root node with label $0$ for the central poset and leaves with labels $1$ through $A$, one leaf for each assessor. The posets on the leaves have a conditional distribution given the root poset, so we can think of the model as ``shrinking'' the posets at the leaves towards the central poset. In the following, we write down priors for all these objects and give the posterior distribution for the hierarchy of posets.

\subsection{Grouped data and a hierarchy of posets}

The notation defined in this subsection is illustrated in Figures~\ref{fig:choice-set-hierarchy} and \ref{fig:po-hierarchy}. Let $\A=\{1,\dots,A\}$ be a set of $A$ ``assessors,'' who rank items in $\M=\{1,\dots,M\}$. 
For $a\in \A$, assessor $a$ is presented with $N_a$ choice sets $S_{a,i}\in\B_\M,\ i=1,\dots,N_a$ of varying size to order. 
Let $m_{a,i} = |S_{a,i}|$ be the number of items in the $i$'th set ordered by assessor $a$ and let $\M_a=\cup_{i=1}^{N_a} S_{a,i}$ be the set of items assessor $a$ was actually asked to rank with $M_a=|\M_a|$. Let $S_a=(S_{a,1},\dots,S_{a,N_a})$ and $S=(S_1,\dots,S_A)$. 
Let $\M_0=\cup_{a\in\A}\M_a$ be the set of items ranked by at least one assessor. For simplicity (and using marginal consistency) we assume $\M_0=\M$ so $|\M_0|=M$. 
Let 
$Y_{a,i}\in \C_{S_{a,i}}$ be the complete order assessor $a$ returned for choice set $S_{a,i}$.
Let $Y_{a}=(Y_{a,1},\dots,Y_{a,N_a})$ be the data returned from assessor $a$ and let $Y=(Y_1,\dots,Y_A)$ be all data. 
We condition on a fixed collection of choice sets which may be chosen in an arbitrary way. 
Relations between these sets are illustrated in Figure~\ref{fig:choice-set-hierarchy}. 
\begin{figure}
    \centering
    \hspace*{0in}\scalebox{1.25}{    \begin{tikzpicture}[thick,scale=0.7, every node/.style={scale=0.25,font=\Huge}]
        \node[draw, circle, minimum width=2.75cm,fill=pink] (0) at (0, 2) {$\M_0$};
        \node[draw, circle, minimum width=2.75cm] (1) at (-3, 0) {$\M_1$};
        \node[draw, circle, minimum width=2.75cm] (2) at (-0.5, 0) {$\M_2$};
        \node[] (3) at (0.75, 0) {};
        \node[] (4) at (1.75, 0) {};
        \node[draw, circle, minimum width=2.75cm] (5) at (3, 0) {$\M_A$};
        \edge [->] {0} {1};
        \edge [->] {0} {2};
        \edge [->] {0} {5};
        \edge [dotted,-] {3} {4};

      \node[draw, circle, minimum width=2.75cm,fill=lightgray] (6) at (-6, -2) {$S_{1,1}$};
      \node[draw, circle, minimum width=2.75cm,fill=lightgray] (7) at (-4, -2) {$S_{1,N_1}$};
      \edge [dotted,-] {6} {7};

      \node[draw, circle, minimum width=2.75cm,fill=lightgray] (8) at (-2, -2) {$S_{2,1}$};
      \node[draw, circle, minimum width=2.75cm,fill=lightgray] (9) at (0, -2) {$S_{2,N_2}$};
      \edge [dotted,-] {8} {9};

      \node[] (8b) at (1.5, -2) {};
      \node[] (9b) at (2.5, -2) {};
      \edge [dotted,-] {8b} {9b};

      \node[draw, circle, minimum width=2.75cm,fill=lightgray] (10) at (4, -2) {$S_{A,1}$};
      \node[draw, circle, minimum width=2.75cm,fill=lightgray] (11) at (6, -2) {$S_{A,N_A}$};
      \edge [dotted,-] {10} {11};

      \edge [->] {1} {6};
      \edge [->] {1} {7};
      \edge [->] {2} {8};
      \edge [->] {2} {9};
      \edge [->] {5} {10};
      \edge [->] {5} {11};
      
      
    \end{tikzpicture}}
    \caption{Caption}
    \label{fig:choice-set-hierarchy}
\end{figure}



Let $\U^{(a)}\in \R^{M\times K},\ a=0,1,\dots,A$ be $A+1$ real matrices with one $K$-dimensional row-vector $\U^{(a)}_j=(\U^{(a)}_{j,1},\dots,\U^{(a)}_{j,K})$ for each item $j\in\M_0$. The ``global'' preference weights are $\U^{(0)}$ and $\U^{(a)}$ are the preference-weights for assessor $a=1,\dots,A$. 
Let $H^{(0)} = h(\U^{(0)})$ with $H^{(0)}\in \H_{\M}$ be the global partial order, and for $a=1,\dots,A$ let $H^{(a)}=h(\U^{(a)}),\ H^{(a)}\in\H_{\M}$ be the partial order of preferences held by the $a$'th assessor. 

We will only need weight parameters $\u^{(a)}=(\U^{(a)}_{j,:})_{j\in\M_a}\in \R^{m_a\times K}$ and partial orders $h^{(a)}=H^{(a)}[\M_a]=h(\u^{(a)})\in\H_{\M_a}$ for items actually ranked because marginal consistency will allow us to drop the rest.
Write $\u = (\u^{(0)},\u^{(1)},\dots,\u^{(A)})$ and let $h = h(\u)$ (applied matrix by matrix), so that $h=(h^{(0)},h^{(1)},\dots,h^{(A)})$. The full parameterisations $\U$ and $H$ are defined in a similar way. 
The partial order hierarchy is illustrated in Figure~\ref{fig:po-hierarchy}.
\begin{figure}
    \centering
    \hspace*{0in}\scalebox{1.25}{    \begin{tikzpicture}[thick,scale=0.7, every node/.style={scale=0.25,font=\Huge}]
        \node[draw, circle, minimum width=2.75cm,fill=pink] (0) at (0, 2) {$h^{(0)}$};
        \node[draw, circle, minimum width=2.75cm] (1) at (-3, 0) {$h^{(1)}$};
        \node[draw, circle, minimum width=2.75cm] (2) at (-0.5, 0) {$h^{(2)}$};
        \node[] (3) at (0.75, 0) {};
        \node[] (4) at (1.75, 0) {};
        \node[draw, circle, minimum width=2.75cm] (5) at (3, 0) {$h^{(A)}$};
        \edge [->] {0} {1};
        \edge [->] {0} {2};
        \edge [->] {0} {5};
        \edge [dotted,-] {3} {4};

      \node[draw, circle, minimum width=2.75cm,fill=lightgray] (6) at (-6, -2) {$Y_{1,1}$};
      \node[draw, circle, minimum width=2.75cm,fill=lightgray] (7) at (-4, -2) {$Y_{1,N_1}$};
      \edge [dotted,-] {6} {7};

      \node[draw, circle, minimum width=2.75cm,fill=lightgray] (8) at (-2, -2) {$Y_{2,1}$};
      \node[draw, circle, minimum width=2.75cm,fill=lightgray] (9) at (0, -2) {$Y_{2,N_2}$};
      \edge [dotted,-] {8} {9};

      \node[] (8b) at (1.5, -2) {};
      \node[] (9b) at (2.5, -2) {};
      \edge [dotted,-] {8b} {9b};

      \node[draw, circle, minimum width=2.75cm,fill=lightgray] (10) at (4, -2) {$Y_{A,1}$};
      \node[draw, circle, minimum width=2.75cm,fill=lightgray] (11) at (6, -2) {$Y_{A,N_A}$};
      \edge [dotted,-] {10} {11};

      \edge [->] {1} {6};
      \edge [->] {1} {7};
      \edge [->] {2} {8};
      \edge [->] {2} {9};
      \edge [->] {5} {10};
      \edge [->] {5} {11};
      
      
    \end{tikzpicture}}
    \caption{Caption}
    \label{fig:po-hierarchy}
\end{figure}

\subsection{A prior for the poset hierarchy}
A prior on the latent weights $\u$ determines a prior for the hierarchy of posets $h$ via
\[
\pi_{\M_{0:A}}(h|\psi)=E_{\u|\psi}(\mathbb{I}_{h(\u)=h}),\quad h\in \H_{\M_{0:A}}.
\]
Here $\psi$ is a set of prior parameters specified below. 
Before we write down the prior, we list the properties it must posses. 
We need a model with a hyperparameter controlling the correlation of $h^{(0)}$ and $h^{(a)}$. 
If $\M_a=\M$ for $a=1,\dots,A$ (all assessors rank all items) then the posets are $H^{(0)},\dots,H^{(A)}$ and the prior is $\pi_{{\M}^{A+1}}(H|\psi),\ H\in\H_{\M^{A+1}}$. We assume assessors are exchangeable samples from a larger population of assessors so their posets, $H^{(1)},\dots,H^{(A)}$ are exchangeable also. For interpretive clarity, the global order $H^{(0)}$ should come from the same population, so marginally, $H^{(0)}\sim H^{(a)},\ a=1,\dots,A$.  

We also require the prior hierarchy to be marginally consistent.
Let $s_{0:A}=(s_0,s_1,\dots,s_A)$ be a generic collection of $A+1$ subsets $s_a\in\B_\M,\ a\in\A$. Define a set of suborders $H[s_{0:A}]=(H^{(0)}[s_0],H^{(1)}[s_1],\dots,H^{(A)}[s_A])$, so that $H[s_{0:A}]\in \H_{s_{0:A}}$.
\begin{definition} ({\it Marginally Consistent Hierarchy})
    The family of probability distributions $\pi_{{s_{0:A}}}(\cdot|\psi),\ s_{0:A}\in \B_\M^{A+1}$ is marginally consistent if $H\sim \pi_{{\M}^{A+1}}(\cdot|\psi)$ implies $H[s_{0:A}]\sim \pi_{{s_{0:A}}}(\cdot|\psi)$ for all $s_{0:A}\in \B_\M^{A+1}$ with $s_0=\cup_{a=1}^A s_a$.
\end{definition}
This says the prior we get by taking suborders of random partial orders on the universe of choices is the same as the prior we defined separately for the suborders.  
Finally, we would like the same control over the distributions of the depths $d(h^{(a)}), a = 0, 1, \dots, A$ that we had for a single partial order in Section~\ref{sec:PO-single-inference}.
We give the prior and set out its properties in Theorem~\ref{thm:Hpo-prior-gumbel-U}. See Appendix~\ref{app:hpo-prior-theorem-proof} for the proof and Figure~\ref{fig:prior-hierarchy-overview} for a graphical display of notation.
\begin{figure}
    \centering
    \begin{tabular}{c}
     $u^{(0)}_{j,:}\!\!\sim\! N(0_K,\Sigma_\rho),\, j\in \M_0$ \\[0.1in]
      \scalebox{1}{    \begin{tikzpicture}[thick,scale=0.7, every node/.style={scale=0.25,font=\Huge}]
        \node[draw, circle, minimum width=2.75cm,fill=pink] (0) at (0, 2) {$u^{(0)}$};
        \node[draw, circle, minimum width=2.75cm] (1) at (-2, 0) {$u^{(1)}$};
        \node[draw, circle, minimum width=2.75cm] (2) at (-0.5, 0) {$u^{(2)}$};
        \node[] (3) at (0.25, 0) {};
        \node[] (4) at (1.25, 0) {};
        \node[draw, circle, minimum width=2.75cm] (5) at (2, 0) {$u^{(A)}$};
        \edge [->] {0} {1};
        \edge [->] {0} {2};
        \edge [->] {0} {5};
        \edge [dotted,-] {3} {4};
    \end{tikzpicture}}\\[0.1in]
      $u^{(a)}_{j,:}\!\!\sim\! N(\tau u^{0}_{j,:},(1\!-\!\tau^2)\Sigma_\rho)$\\$a\in\A,j\in \M_a$
\end{tabular}
\begin{tabular}{c}
      $\eta^{(a)}_{j,:}=G^{-1}(\Phi(u^{(a)}_{j,:}))+\alpha_j 1_K$\\[0.1in]
      \scalebox{1}{    \begin{tikzpicture}[thick,scale=0.7, every node/.style={scale=0.25,font=\Huge}]
        \node[draw, circle, minimum width=2.75cm,fill=pink] (0) at (0, 2) {$\eta^{(0)}$};
        \node[draw, circle, minimum width=2.75cm] (1) at (-2, 0) {$\eta^{(1)}$};
        \node[draw, circle, minimum width=2.75cm] (2) at (-0.5, 0) {$\eta^{(2)}$};
        \node[] (3) at (0.25, 0) {};
        \node[] (4) at (1.25, 0) {};
        \node[draw, circle, minimum width=2.75cm] (5) at (2, 0) {$\eta^{(A)}$};
        \edge [->] {0} {1};
        \edge [->] {0} {2};
        \edge [->] {0} {5};
        \edge [dotted,-] {3} {4};
    \end{tikzpicture}}\\[0.1in]
      $a=0,1,\dots,A,\ j\in\M_a$\\
      \phantom{$a\in\A,j\in \M_a$}
\end{tabular}
\begin{tabular}{c}
        $h^{(0)}=h(\eta^{(0)})$\\[0.1in]
      \scalebox{1}{    \begin{tikzpicture}[thick,scale=0.7, every node/.style={scale=0.25,font=\Huge}]
        \node[draw, circle, minimum width=2.75cm,fill=pink] (0) at (0, 2) {$h^{(0)}$};
        \node[draw, circle, minimum width=2.75cm] (1) at (-2, 0) {$h^{(1)}$};
        \node[draw, circle, minimum width=2.75cm] (2) at (-0.5, 0) {$h^{(2)}$};
        \node[] (3) at (0.25, 0) {};
        \node[] (4) at (1.25, 0) {};
        \node[draw, circle, minimum width=2.75cm] (5) at (2, 0) {$h^{(A)}$};
        \edge [->] {0} {1};
        \edge [->] {0} {2};
        \edge [->] {0} {5};
        \edge [dotted,-] {3} {4};
    \end{tikzpicture}}\\[0.1in]
      $h^{(a)}=h(\eta^{(a)}),\ a=1,\dots,A$\\
      \phantom{$a\in\A,j\in \M_a$}
\end{tabular}
\\[-0.2in]
    \caption{Caption}
    \label{fig:prior-hierarchy-overview}
\end{figure}

\begin{theorem} (Hierarchical Partial Order prior) \label{thm:Hpo-prior-gumbel-U} For $\alpha_\M=X\beta$ and $\Sigma_\rho$ as in Section~\ref{sec:PO-single-inference}, for $0<\tau\le 1$ and sets $\M_a\in\B_\M,\ a\in \A$, let $\M_0=\cup_{a=1}^A \M_a$ and
    \begin{align} 
    \u^{(0)}_{j,:}&\sim N(0,\Sigma_\rho),\quad\mbox{independent for each $j\in \M_0$,}\label{eq:Hpo-prior-U0}\\
    \u^{(a)}_{j,:}|\u^{(0)}_{j,:} &\sim N\left(\tau\u^{(0)}_{j,:},\, (1-\tau^2)\Sigma_\rho\right)\quad\mbox{independent for each $a\in\A$ and $j\in \M_a$,}\label{eq:Hpo-prior-Ua}\\ 
    \eta^{(a)}_{j,:}&=G^{-1}(\Phi(\u^{(a)}_{j,:}))+\alpha_j 1^T_K \quad\mbox{for $a=0,1,\dots,A$ and $j\in \M_a$ and}\label{eq:Hpo-prior-eta}\\
    h&=h(\eta(\u,\beta))\quad\mbox{for $\eta(\u,\beta)=(\eta^{(0)},\eta^{(1)},\dots,\eta^{(A)})$ and $h^{(a)}=h(\eta^{(a)})$.}\label{eq:Hpo-prior-h}
\end{align}
For $h\in \H_{\M_{0:A}}$ let $\pi_{\M_{0:A}}(h|\rho,\beta,\tau)=E_{\u}(\mathbb{I}_{h(\eta(\u,\beta))=h})$ be the resulting prior for the poset hierarchy $h$. The marginal prior for $h^{(a)}$ is $\pi_{\M_{0:A}}(h^{(a)}|\rho,\beta,\tau)$. 
\begin{enumerate}
    \item (single PO marginals) For $a=0,1,\dots,A$, $\pi_{\M_{0:A}}(h^{(a)}|\rho,\beta,\tau)=\pi_{\M_a}(h^{(a)}|\rho,\beta)$ where $\pi_{\M_a}$ is the single partial-order prior given in \eqref{eq:po-single-h-prior};
    \item (PL hierarchy at $K=1$) when $K=1$, 
    $h^{(a)}\sim PL(\alpha_{\M_a},\M_a)$ for $a\in\A$;
    \item (independence at $\tau=0$) when $\tau=0$, $\pi_{\M_{0:A}}(h|\rho,\beta,\tau)=\prod_{a=0}^A\pi_{\M_{a}}(h^{(a)}|\rho,\beta,\tau)$;
    \item (matching posets at $\tau=1$) for each $a\in \A$, $\displaystyle\lim_{\tau\to 1}\pi_{\M_{0:A}}(h^{(a)}|\rho,\beta,\tau)=\mathbb{I}_{h^{(a)}=h^{(0)}[\M_a]}$;
    \item (support) if $K\ge \lfloor M/2\rfloor$ then $\pi_{\M_{0:A}}(h|\rho,\beta,\tau)>0$ for all $h\in\H_{\M_{0:A}}$; 
    \item (marginal consistency) If $H\sim \pi_{\M^{A+1}}(\cdot|\rho,\beta,\tau)$ then $H[s_{0:A}]\sim \pi_{s_{0:A}}(\cdot|\rho,\beta,\tau)$ for every $s_{0:A}\in \B_\M^{A+1}$ with $s_0=\cup_{a=1}^A s_a$.
\end{enumerate}
\end{theorem}

 The relation between $\u^{(0)}_{j,:}$ and $\u^{(a)}_{j,:}$ is equivalent to running a multivariate Ornstein-Uhlenbeck process $dX_t=-X_t+V^{1/2}dW_t$ with $V=2\Sigma_\rho$ starting at $X_0=\u^{(0)}_{j,:}$ to get $X_t\sim \u^{(a)}_{j,:}$ at $t=-2\log(\tau)$. The stationary distribution of this process is $N(0_K,\Sigma_\rho)$ so if we run the process for no time at all ($\tau=1$) we get back $\u^{(0)}_{j,:}$ and if we run it for an infinite amount of time ($\tau\to 0$) we get an independent draw from $N(0_K,\Sigma_\rho)$.


\subsection{Observation Model and posterior distribution} \label{sec:Hpo-labeled-lkd-post}

The parameters of the poset-prior are $\rho,\beta,\tau$ and $\u$. The noise-free likelihood in \eqref{eq:lkd-noise-free} adds no further parameters, so the dependence relations in the prior and likelihood are
\[\scalebox{1}{    \begin{tikzpicture}[thin,scale=1, every node/.style={scale=1,font=\large}]
       \node[] (1) at (0, 2) {$\rho$};
       \node[] (2) at (0, 0) {$\tau$};
       \node[] (3) at (2, 2) {$\u^{(0)}$};
       \node[] (4) at (2, 0) {$\u^{(a)}$};
       \node[] (5) at (4, 2) {$\beta$};
       \node[] (6) at (4, 0) {$\eta$};
       \node[] (7) at (6, 0) {$h$};
       \node[] (8) at (8, 0) {$Y$};
       \edge [thin,->] {1} {3};
       \edge [->] {1} {4};
       \edge [->] {2} {4};
       \edge [->] {3} {4};
       \edge [dashed,->] {3} {6};
       \edge [dashed,->] {4} {6};
       \edge [dashed,->] {5} {6};
       \edge [dashed,->] {6} {7};
       \edge [->] {7} {8};
    \end{tikzpicture}}\]
Deterministic steps given by \eqref{eq:Hpo-prior-eta} and \eqref{eq:Hpo-prior-h} are shown with dashed lines.
The priors for $\rho$ and $\beta$ are unchanged from Section~\ref{sec:PO-single-inference}. The prior for $\tau$ is $\tau\sim U(0,1)$. The prior for $\u=(\u^{(0)},\u^{(1)},\dots,\u^{(A)})$ is given in \eqref{eq:Hpo-prior-U0} and \eqref{eq:Hpo-prior-Ua}. 

The likelihood for these parameters given the grouped data is
\begin{align}
    p_S(Y|\beta,\u)=\prod_{a=1}^A \prod_{i=1}^{N_a} p_{S_{a,i}}(Y_{a,i}|h^{(a)}[S_{a,i}]) 
    \nonumber
\end{align}
with $p_{S_{a,i}}(Y_{a,i}|h^{(a)}[S_{a,i}])$ given in \eqref{eq:lkd-noise-free} and $h^{(a)}[S_{a,i}]$ a suborder depending only on $\beta$ and $\u$ through $h=h(\eta)$ and $\eta=\eta(\u,\beta)$. 
The posterior is
\begin{align}
    \pi_S(\rho,\beta,\tau,\u|Y)&\propto \pi_R(\rho)\pi_B(\beta)\pi_T(\tau)\pi(\u|\rho,\tau)p_S(Y|h(\u,\beta))\\
    &=
    \pi_R(\rho)\pi_B(\beta)\pi_T(\tau) \prod_{j\in\M_0} N\left(\u^{(0)}_{j,:};0_{K},\Sigma_\rho\right)\times \nonumber\\
    &\quad \prod_{a=1}^A \left[p_{S_{a}}(Y_{a}|h(\eta(\u^{(a)},\beta))) \prod_{j\in\M_a} N\left(\u^{(a)}_{j,:};\tau\u^{(0)}_{j,:},(1-\tau^2)\Sigma_\rho\right)\right].\label{eq:Hpo-posterior}
\end{align}
As in the single-poset analysis in Section~\ref{sec:PO-single-inference}, 
the latent dimension $K$ is inferred jointly with the partial order, 
following the varying-$K$ posterior in \eqref{eq:po-posterior-vary-K}.

\section{Clustering unlabeled orders and clustering assessors}
\label{sec:Hpo-cluster}
\def\N{{\mathcal{N}}}
\subsection{Clustering unlabeled orders}\label{sec:cluster-lists}
In Section~\ref{sec:Hpo-labeled} the data were $A$ groups $Y_a,\ a\in \A$ of orders. If we simply have a collection of $N$ lists $Y=(Y_1,\dots,Y_N)$ ordering $N$ choice sets $S=(S_1,\dots,S_N)$ and we think there may be some latent group structure, we add a partition-parameter $C=(C_1,\dots,C_G)$ clustering $\N=\{1,\dots,N\}$ into $G$ groups of size $n_g=|C_g|,\ g=1,\dots,G$. Let $\Xi_\N$ be the set of all partitions of $\N$. We give $C$ a Pitman-Yor prior with discount parameter $d$, concentration $\vartheta$ and $\gamma=(d,\vartheta)$,
\begin{equation}\label{eq:Hpo-cluster-c-prior}
      \pi_{\gamma,\N}(C)=\frac{\Gamma(\vartheta)}{\Gamma(\vartheta+N)}
                        \frac{d^{G}\Gamma(\vartheta/d+G)}{\Gamma(\vartheta/d)}
                        \frac{\prod_{g=1}^G \Gamma(n_g-d)}{\Gamma(1-d)},\quad C\in\Xi_\N.
\end{equation}
The main change here is that the number of groups $G$ and the orders $Y_i,\ i\in C_g$ associated with group $g=1,\dots,G$ are now random. Each group has matrix of preference weights $\U^{(g)}\in\R^{M\times K}$ and since $G$ is random the dimension $G\times M\times K$ of $\U=(\U^{(0)},\U^{(1)},\dots,\U^{(G)})$ is random. Notice that we don't drop unranked items from individual groups. If we took $\M_g=\cup_{i\in C_g} S_i$ and worked with $\u^{(g)}=\U^{(g)}_{\M_g,:}$ then the row-dimension of $\u^{(g)}$ would be random. This would lead to the same model, as marginal consistency still holds. To keep things simple in the MCMC we set $\M_g=\M$ for each group $g=1,\dots,g$; the row-dimension of $\U^{(g)}$ is always $M$ and the posets $h^{(g)}\in\H_\M$ all order the universe of choices (we assume any items missing from all orders were removed at the start, so $\M=\M_0$). We are leaving parameters in the model which we could marginalise exactly. If we are estimating $K$ using reversible jump MCMC then this adds another element of randomness to the overall parameter dimension. 

The likelihood now depends on the partition $C$, that is
\begin{align}\label{eq:Hpo-cluster-c-lkd}
    p_S(Y|C,h(\U,\beta))=\prod_{g=1}^G \prod_{i\in C_g} p_{S_i}(Y_i|h^{(g)}[S_i]),
\end{align}
with $h^{(g)}=h(\eta(\U^{(g)},\beta))$. The posterior becomes
\begin{align}
    \pi_S(\rho,\beta,\tau,\U,C|Y)\propto \pi_R(\rho)\pi_B(\beta)\pi_T(\tau)\,\pi(\U|\rho,\tau)\,\pi_{\gamma,\N}(C)\,p_S(Y|C,h(\U,\beta)),\label{eq:Hpo-cluster-posterior}
\end{align}
with $\pi_{\gamma,\N}$ given in \eqref{eq:Hpo-cluster-c-prior} and $p_S(Y|C,h(\U,\beta))$ in \eqref{eq:Hpo-cluster-c-lkd} and otherwise as \eqref{eq:Hpo-posterior}. As for the single poset analysis in Section~\ref{sec:PO-single-inference} and the hierarchical model for labeled lists in Section~\ref{sec:Hpo-labeled}, the likelihood may be extended to allow for noise using the models in the next section, and $K$ may also be variable in the posterior, as in \eqref{eq:po-posterior-vary-K}. 

\subsection{Clustering assessors}\label{sec:cluster-assessors}
When the number of assessors is large (as in the sound data, with $A=46$ assessors) it may be useful to cluster them in groups of assessors with similar preferences.
The setup is similar to the previous section.
The partition $C\in \Xi_\A$ now clusters $\A=\{1,\dots,A\}$ into $G$ groups with $n_g=|C_g|,\ g=1,\dots,G$ as before. The Pitman-Yor prior $\pi_{\gamma,\A}(C)$ for $C$ is unchanged. 

Recall that the data associated with assessor $a\in\A$ is $N_a$ orders $Y_a=(Y_{a,1},\dots,Y_{a,N_a})$. The data in cluster $g$ are now $\{Y_a: a\in C_g\}$, so $\sum_{a\in C_g} N_a$ orders in all. In the MCMC, when we move an assessor from one cluster to another, \emph{all} their orders travel with them, rather than moving one order at a time as in the previous section. In other respects the parameterisation and priors are the same.

The likelihood becomes
\begin{align}\label{eq:Hpo-assessor-cluster-c-lkd}
    p_S(Y|C,h(\U,\beta))=\prod_{g=1}^G \prod_{a\in C_g} \prod_{i=1}^{N_a} p_{S_{a,i}}(Y_{a,i}|h^{(g)}[S_{a,i}]),
\end{align}
with $h^{(g)}=h(\eta(\U^{(g)},\beta))$. The posterior for assessor-clustering in given in \eqref{eq:Hpo-cluster-posterior} now uses the assessor-clustering likelihood in \eqref{eq:Hpo-assessor-cluster-c-lkd} (and we cluster assessors, not lists, so $C\in\Xi_\A$ and and the $C$-prior is $\pi_{\gamma,\A}(C)$).

\section{Extending the observation model to allow for noise}\label{sec:lkd-noisy}

In this section we give some further observation models. For ease of exposition, we drop the assessor labels and return to the ``single-partial-order'' setup of Sections~\ref{sec:po-foundation-models} and \ref{sec:PO-single-inference}.

\subsection{Repeated selection models}\label{sec:sequential-choice-models}

A draw $y\sim p_S(\cdot|h)$ from the model in \eqref{eq:lkd-noise-free} can be realised sequentially, building up the list one element at a time. Models of this kind are called {\it repeated selection} models. Many models for preference orders are repeated selection models: all our observation models, the Plackett-Luce model and some Mallows-$\phi$ models. 

Let $y_{i:m}=(y_i,\dots,y_m)$ so that $h[y_{i:m}]=(y_{i:m},\succ_h)$ is the suborder of $h$ restricted to $y_{i:m}$.
For $j\in S$ let 
\[
\L_j[h] = \{\ell \in \L[h]: \max(\ell) = j\}
\]
be the set of all linear extensions of $h$ with maximal element $j$.
We can write $p_S(y|h)$ in \eqref{eq:lkd-noise-free} as a telescoping product over suborders
\begin{align}\label{eq:po_lkd_telescope}
    p_S(y|h)
    =\prod_{i=1}^{m-1} q_{y_{i:m}}(y_i|h[y_{i:m}]),
\end{align}
where
\begin{equation}\label{eq:po-lkd-select-first}
    q_{y_{i:m}}(y_i|h[y_{i:m}])=\frac{|\L_{y_i}[h[y_{i:m}]]|}{|\L[h[y_{i:m}]]|}
\end{equation}
is the probability $y_i$ is selected next from the remaining choices. 
The product in \eqref{eq:po_lkd_telescope} is equal to $p_S(y|h)$ in \eqref{eq:lkd-noise-free} because the number of LEs $|\L_{y_i}[h[y_{i:m}]]|$ headed by $y_i$ is the number of LEs $=|\L[h[y_{i+1:m}]]|$ in the suborder that remains after $y_i$ is removed. 
The likelihood for top-$k$ data truncates the product in \eqref{eq:po_lkd_telescope} at $k$.


\subsection{Queue jumping noise}\label{sec:noise-qj}

\cite{nicholls25AOAS} allow for noise in the model for observed orders by allowing individuals to ``jump the queue'' in a repeated selection model. 
Before the $i$'th entry of $y_{1:m}$ is chosen, there are $m-i+1$ elements of $S$, with labels $y_{i:m}$, yet to be placed. With probability $p$ the next entry (i.e., $y_i$) is chosen at random, ignoring any order constraints and otherwise $y_i$ is  $y_i$ is chosen next with probability $q_{y_{i:m}}(y_i|h[y_{i:m}])$ in \eqref{eq:po-lkd-select-first}. 
Working from the top down,
\begin{align}\label{eq:lkd-noisy-down-one-list}
p_S(y_{1:m}|h,p)&=\prod_{i=1}^{m-1}q(y_{i:m}|h[y_{i:m}],p)\\
\intertext{where}
q(y_{i:m}|h[y_{i:m}],p)&=\frac{p}{m-j+1}+(1-p)\,q_{y_{i:m}}(y_i|h[y_{i:m}]).
\label{eq:lkd-noisy-down-one-list-q}
\end{align}
When we use this model we replace the likelihood $p_{S_{i}}(Y_{i}|h[S_{i}])$ in \eqref{eq:po_lkd_full_noisefree} with $p^{(D)}_{S_{i}}(Y_{i}|h[S_{i}],p)$ in \eqref{eq:lkd-noisy-down-one-list}. The posteriors \eqref{eq:po-posterior-marginal}, \eqref{eq:Hpo-posterior} and \eqref{eq:Hpo-cluster-posterior} add a parameter $p\in [0,1]$ with a uniform prior.

\subsection{$\eta$-weighted queue-jumping}\label{sec:eta-weighted-QJ}



The queue-jumping model has the weakness that it doesn't take into account how many places an item jumps, or the number of preference relations which are contradicted by the error.
Referring to \eqref{eq:lkd-noisy-down-one-list-q}, an error occurs with probability $p$. When this happens the next element in the list is chosen uniformly at random from this that remain. We let this choice depend on $\eta(\U,\beta)$:
for $j\in \M$ let $\bar\eta_j=K^{-1}\sum_{k=1}^K \eta_{j,k}$; when an error occurs, the next element $y_i$ in the list is chosen with probability proportional to $\exp(\bar\eta_{y_i})$ as in \eqref{eq:PL-def-one-step}; when there is no error the next element is chosen according to the noise-free partial-order model in \eqref{eq:po-lkd-select-first}; the modified queue-jumping likelihood is
\begin{align}\label{eq:lkd-noisy-down-one-list-eta-weighted}
p_S(y_{1:m}|\eta,p)&=\prod_{i=1}^{m-1}q(y_{i:m}|\eta_{y_{i:m},:},p)\\
\intertext{where}
q(y_{i:m}|\eta_{y_{i:m}},p)&=p\,q(y_i|\bar\eta_{y_{i:m}})+(1-p)\,q_{y_{i:m}}(y_i|h(\eta_{y_{i:m}}))
\label{eq:lkd-noisy-down-one-list-eta-weighted-q}
\intertext{and}
q(y_i|\bar\eta_{y_{i:m}})&=\frac{e^{\bar \eta_{y_i}}}{\strut\sum_{j=i}^m e^{\bar \eta_{y_{j}}}}
\end{align}
is the one-step Plackett-Luce probability to select $y_i$ from $y_{i:m}$ given weights $\bar\eta_{y_{i:m}}$. The choice probability for the next element in the list is a mixture of Plackett-Luce and the noise-free partial order observation model \eqref{eq:lkd-noise-free}. 

This captures the idea that errors become less likely as the distance between $j_1$ and $j_2$ on $h$ increases. Let $h=h(\U)$ be a poset on $\M$, $s=(j_1,j_2)$ a choice set with $j_1\succ_h j_2$ and let $\delta=|\{j\in \M: j_1\succ_h j\succ_h j_2\}|$ be the number of items between $j_1$ and $j_2$. If the data are $y=(j_2,j_1)$, which is only possible if there was a jump, then
\[
p_S(y|\eta,p)=p\frac{1}{1+e^{\bar\eta_{j_1}-\bar\eta_{j_2}}}.
\]
If $\delta$ is large then $\bar\eta_{j_1}\gg \bar\eta_{j_2}$ (typically, as $\eta_{j_1,k}>\eta_{j,k}>\eta_{j_2,k}$ for each $k$ and many $j$).
The probability for an error goes to zero as $\bar\eta_{j_1}-\bar\eta_{j_2}\to\infty$. If $\delta$ is small then $\bar\eta_{j_1}-\bar\eta_{j_2}$ is typically small (similar reasoning) and in this case the probability for error approaches $p/2$.

When $p=1$ the mixture in \eqref{eq:lkd-noisy-down-one-list-eta-weighted} reduces to Plackett-Luce in \eqref{eq:PL-def-one-step} with weights $\alpha_j=\bar\eta_j$ and in particular when $K=1$ and $\eta_{j,:}=\lambda_j+x_{j}\beta$ with $\lambda_j=G^{-1}(\Phi(\U_{j,:}))$ we have Plackett-Luce in familiar form. This gives an alternative way to nest Plackett-Luce within our model. 

\subsection{Frontier-softmax likelihood}
\label{sec:front_softmax_qj}
This is a variant of the queue-jumping likelihood, proposed in \cite{li26delinearizing} to model execution traces from LLMs. Consider a poset $h\in\H_S$ and an order $y_{1:m}$. The $i$'th item $y_i$ is chosen from the \emph{frontier} 
$\max(h[y_{i:m}])$ with probability weighted by the number of items
\[
D(y_i; h[y_{i:m}])=1+|\{j\in y_{i:m}: y_i\succ_h j\}|
\]
which $y_i$ dominates (counting itself). This expresses the idea that a well-trained agent is biased toward resolving the hardest paths first, because it’s a waste of resources to work on the easy part if the hard part fails. In \cite{li26delinearizing} the utility for choosing $j\in y_{i:m}$ is
\begin{equation}
Q(j;h[y_{i:m}])=
\begin{cases}
\log(D(j;h[y_{i:m}])), & j\in\max(h[y_{i:m}]),\\
-\infty, & \text{otherwise}.
\end{cases}
\label{eq:q_succ}
\end{equation}
and the noise-free selection probability is
\begin{equation}\label{eq:lkd-noisy-frontier}
q_{y_{i:m}}(y_i|h[y_{i:m}],\varphi)=\frac{\exp\!\left(\varphi \, Q(y_i;h[y_{i:m}])\right)}
{\sum_{j\in \max(h[y_{i:m}])} \exp\!\left(\varphi\, Q(j;h[y_{i:m}])\right)}
\end{equation}
with $\varphi\ge 0$ a parameter controlling how sharply the distribution
concentrates on high-scoring items in the frontier.
This is embedded in a queue-jumping model for the probability $p_S(y_{1:m}|h,p,\varphi)$ to realise $y_{1:m}$ by replacing $q_{y_{i:m}}(y_i|h[y_{i:m}])$ in \eqref{eq:lkd-noisy-down-one-list-q} with $q_{y_{i:m}}(y_i|h[y_{i:m}],\varphi)$ in \eqref{eq:lkd-noisy-frontier}.
\cite{li2026delinearizing} call this the \emph{frontier-softmax model}. Although all our likelihoods are in effect defining a utility like $Q$ for choosing $j$ when $y_{i:m}$ remain, they are all generative models. Frontier-softmax approaches modeling $via$ utility. For a wider discussion of random utility models for orders and background literature see \cite{seshadri21}.

The frontier-softmax likelihood is computationally attractive: given the
frontier and dominance counts, each step only requires a softmax over the
frontier. It is fast to evaluate, while still providing a simple and useful generative model for partially ordered data.

\section{Experiments}
\label{sec:experiments}

This section illusterates and evaluates the models introduced in Sections~\ref{sec:Hpo-labeled}--\ref{sec:lkd-noisy} in four ways.
First, we compare the computational cost of the three observation models.
Second, on synthetic data, we study when HPO and HCPO recover the latent partial-order structure and (for HCPO) the latent clustering.
Third, on agent-trace data, we ask whether hierarchical pooling improves recovery of executable dependency graphs.
Fourth, on human preference data, we compare HCPO with total-order baselines in terms of predictive fit and interpretability.
Unless otherwise stated, structural recovery is evaluated on the transitive closure of the inferred posets; definitions of precision, recall, F1, feasibility, and incomparability-pair coverage (IP-Cov) are given in Appendix~\ref{sec:metrics} and Appendix~\ref{app:metrics_ipcov}.
Throughout this section we use \emph{frontier-softmax} as shorthand for the frontier-softmax likelihood of Section~\ref{sec:front_softmax_qj}. We set the penalty parameter $\varphi=1$ in \eqref{eq:lkd-noisy-frontier}.

\subsection{Likelihood speed benchmarks}
\label{app:likelihood_speed_benchmark}

Before comparing statistical performance, we benchmark the cost of the three likelihoods from Section~\ref{sec:lkd-noisy}.
The goal is not to optimize implementations, but to compare the two exact linear-extension-based likelihoods (queue-jump and weighted queue-jump) with frontier-softmax. 

Because frontier-softmax has better scaling with list-length, and the same support on linear extensions, it provides a computationally attractive alternative to the queue-jump likelihood. We report per-iteration MCMC runtime as a function of the number of items $M$ in both an HPO setting and a clustered HCPO setting.

In the single-poset benchmark we use synthetic data generated at $\tau=0.5$ with $A=5$ assessors, $N_a=10$ full-length lists per assessor, so $S_{a,i}=\M_a=\M_0$ and $m_{a,i}=M$ for $a=1,\dots,A$ and all $i=1,\dots,N_a$, with $M$ varying from $5$ to $50$. Runtime will grow linearly with $A$ and $N_a$, hence the focus on varying list-length. We run MCMC targeting the posterior, averaging over 1000 MCMC iterations and 3 simulated data sets. In Table~\ref{tab:likelihood_speed_benchmark_hpo}, queue-jump and weighted queue-jump have similar runtimes, while frontier-softmax is substantially faster at every problem size and is the only computable likelihood at $M=50$.

\begin{table}[h]
\centering
\small
\caption{\textbf{Per-iteration (milliseconds) runtime benchmark by number of nodes $M$ in the single-poset (HPO-style) setting.}}
\label{tab:likelihood_speed_benchmark_hpo}
\setlength{\tabcolsep}{6pt}
\renewcommand{\arraystretch}{1.05}
\begin{tabular}{rccc}
\toprule
\textbf{Nodes $M$} & \textbf{Queue-jump} & \textbf{Weighted queue-jump} & \textbf{frontier-softmax} \\
\midrule
5  & $1.09 \pm 0.23$   & $1.19 \pm 0.28$   & $0.80 \pm 0.03$ \\
10 & $11.30 \pm 0.76$   & $11.3 \pm 2.2$   & $1.47 \pm 0.07$ \\
20 & $59 \pm 21$   & $66 \pm 22$   & $3.03 \pm 0.35$ \\
30 & $100 \pm 12$ & $103 \pm 11$ & $4.82 \pm 0.74$ \\
50 & \texttt{timeout}  & \texttt{timeout}  & $11.4 \pm 2.9$ \\
\bottomrule
\end{tabular}

\vspace{0.25em}
\footnotesize
\end{table}

The same qualitative conclusion holds in HCPO.
Here we use $G=2$ latent groups, $A=5$ assessors, $N_a=10$ lists per assessor, and two representative coupling values, $\tau\in\{0.2,0.8\}$.
Results for clustering runs are shown in Table~\ref{tab:likelihood_speed_benchmark_hcpo}.
The two exact likelihoods time out at large $M$, whereas frontier-softmax stays below 20 ms per iteration even at $M=50$.

\begin{table}[t]
\centering
\small
\caption{\textbf{Per-iteration MCMC runtime (milliseconds) for the clustered HCPO setting under two representative coupling regimes.}
Benchmarks use $G=2$ clusters, $A=5$ assessors, and 50 lists.}
\label{tab:likelihood_speed_benchmark_hcpo}
\setlength{\tabcolsep}{6pt}
\renewcommand{\arraystretch}{1.05}
\begin{tabular}{crrr}
\toprule
$M$ & \textbf{Queue-jump} & \textbf{Weighted queue-jump} & \textbf{Frontier-softmax} \\
\midrule
\multicolumn{4}{l}{\textit{Coupling }\boldmath$\tau=0.2$} \\
5  & $1.20 \pm 0.36$   & $1.56 \pm 0.49$   & $0.98 \pm 0.11$ \\
10 & $23.6 \pm 10.9$   & $29.3 \pm 11.4$   & $2.19 \pm 0.15$ \\
20 & $134.9 \pm 43.8$  & $99.7 \pm 22.4$   & $3.51 \pm 0.41$ \\
30 & $457 \pm 284$     & $265 \pm 100$     & $11.6 \pm 8.7$ \\
50 & \texttt{timeout} & \texttt{timeout} & $14.8 \pm 2.5$ \\
\midrule
\multicolumn{4}{l}{\textit{Coupling }\boldmath$\tau=0.8$} \\
5  & $1.05 \pm 0.34$   & $1.19 \pm 0.32$   & $0.84 \pm 0.01$ \\
10 & $19.7 \pm 5.4$    & $18.2 \pm 3.6$    & $1.68 \pm 0.07$ \\
20 & $80.1 \pm 26.4$   & $78.8 \pm 18.4$   & $3.53 \pm 0.50$ \\
30 & $169 \pm 45$      & $179 \pm 62$      & $6.44 \pm 1.05$ \\
50 & \texttt{timeout} & \texttt{timeout} & $16.7 \pm 3.1$ \\
\bottomrule
\end{tabular}

\vspace{0.25em}
\end{table}

\subsection{Experiment~A: Synthetic HPO study}
\label{sec:synth_hpo_core}

Experiment~A asks when HPO can recover a latent hierarchy of posets when the group labels are observed.
We generate data with $M=10$ items and $A=5$ assessors, and vary five design factors: coupling $\tau\in\{0,0.5,0.9\}$, noise probability $p\in\{0.01,0.1\}$, number of lists per assessor $L\in\{5,10,20\}$, choice-set sizes/list lengths $m\in\{2,5,10\}$, and observation model (queue-jump, weighted queue-jump, or frontier-softmax).
For each cell of the resulting $3\times2\times3\times3\times3=162$-configuration grid we fit the likelihood-matched HPO model, so the fitted model is always well specified.
All runs use $5\times 10^5$ MCMC iterations with 50\% burn-in; Appendix Table~\ref{tab:synth_runtime} reports runtimes. For each fitted posterior we consider four poset-estimators: the marginal posterior-mode poset (MPM-poset, which \cite{li2026delinearizing} show is the Bayes estimator for a loss counting false positive and false negative edges) and consensus posets \citep{nicholls25AOAS} which keep edges with marginal posterior probabilities above a threshold $t\in\{0.33,0.50,0.70\}$. In the main text we always report MPM-posets. 

Recovering the true hierarchy of posets is decided mainly by how much local ordering information each observation contains.
Figure~\ref{fig:blocka_precision_fpr} plots the precision against the false-positive rate for poset edges in the recovered hierachy: reconstructions with $m=2$ are noisy and variable, while those with $m=10$ in the top left corner have high precision and low false-positive rates. In Tables~\ref{tab:blocka_marginal_means} and~\ref{tab:blocka_effect_sizes}, increasing $m$ from 2 to 10 raises the mean F1 from $0.457$ to $0.811$, a larger change than the corresponding marginal effects of $\tau$, $L$, likelihood choice, or noise probability.
Stronger coupling and more lists also help, but their impact is secondary.

Performance across likelihood in Table~\ref{tab:blocka_marginal_means} shows they provide similar levels of information in these well-specified analyses, though the extension-based likelihoods dominate at large list lengths (see Figure~\ref{fig:blocka_precision_fpr}, where Frontier-softmax has few points at top left). Appendix Figure~\ref{fig:blocka_threshold_sensitivity} shows that these conclusions are unchaged if we replace the MPM-poset estimator with posets estimated by thresholding posterior edge probabilities. Appendix Figures~\ref{fig:blocka_true_vs_inferred_poset} and~\ref{fig:blocka_posterior_traces}  show a
representative high-signal run. The inferred global and assessor-specific
posets recover most of the edges, but also exhibit
local discrepancies from the truth. The accompanying posterior traces
provide a diagnostic check on mixing after burn-in.

\begin{table}[t]
\centering
\small
\caption{\textbf{Marginal mean recoverability metrics by design factor in Experiment~A.}
Precision, recall, and F1 are computed on ordered comparable pairs in the transitive closure of the estimated and ground-truth posets, and then averaged over all remaining dimensions of the synthetic design grid. See the definition in Appendix \ref{sec:metrics}}
\label{tab:blocka_marginal_means}
\setlength{\tabcolsep}{8pt}
\renewcommand{\arraystretch}{1.05}
\begin{tabular}{llccc}
\toprule
\textbf{Factor} & \textbf{Level} & \textbf{Precision} & \textbf{Recall} & \textbf{F1} \\
\midrule
\multirow{3}{*}{task size}
  & 2  & 0.343 & 0.684 & 0.457 \\
  & 5  & 0.649 & 0.669 & 0.659 \\
  & 10 & 0.809 & 0.814 & 0.811 \\
\midrule
\multirow{3}{*}{$L$}
  & 5  & 0.441 & 0.667 & 0.531 \\
  & 10 & 0.527 & 0.734 & 0.614 \\
  & 20 & 0.635 & 0.751 & 0.688 \\
\midrule
\multirow{3}{*}{$\tau$}
  & 0.0 & 0.426 & 0.648 & 0.514 \\
  & 0.5 & 0.511 & 0.668 & 0.579 \\
  & 0.9 & 0.627 & 0.810 & 0.707 \\
\midrule
\multirow{2}{*}{$p$}
  & 0.01 & 0.532 & 0.751 & 0.623 \\
  & 0.10 & 0.519 & 0.683 & 0.590 \\
\midrule
\multirow{3}{*}{likelihood}
  & \texttt{queue\_jump}                  & 0.536 & 0.797 & 0.641 \\
  & \texttt{weighted\_queue\_jump}        & 0.570 & 0.686 & 0.622 \\
  & \texttt{frontier\_softmax\_queue\_jump} & 0.486 & 0.702 & 0.574 \\
\bottomrule
\end{tabular}
\end{table}

\begin{table}[t]
\centering
\small
\caption{\textbf{Effect sizes for Experiment~A design factors.}
For each factor, the effect size is defined as the difference between the largest and smallest marginal mean across its levels.}
\label{tab:blocka_effect_sizes}
\setlength{\tabcolsep}{8pt}
\renewcommand{\arraystretch}{1.05}
\begin{tabular}{lccc}
\toprule
\textbf{Factor} & \textbf{$\Delta \mathrm{F1}$} & \textbf{$\Delta$ Precision} & \textbf{$\Delta$ Recall} \\
\midrule
task size  & 0.354 & 0.465 & 0.144 \\
$\tau$     & 0.193 & 0.202 & 0.162 \\
$L$        & 0.157 & 0.194 & 0.083 \\
likelihood & 0.067 & 0.084 & 0.112 \\
$p$        & 0.033 & 0.013 & 0.068 \\
\bottomrule
\end{tabular}
\end{table}

\begin{figure}[h]
\centering
\includegraphics[width=0.9\linewidth]{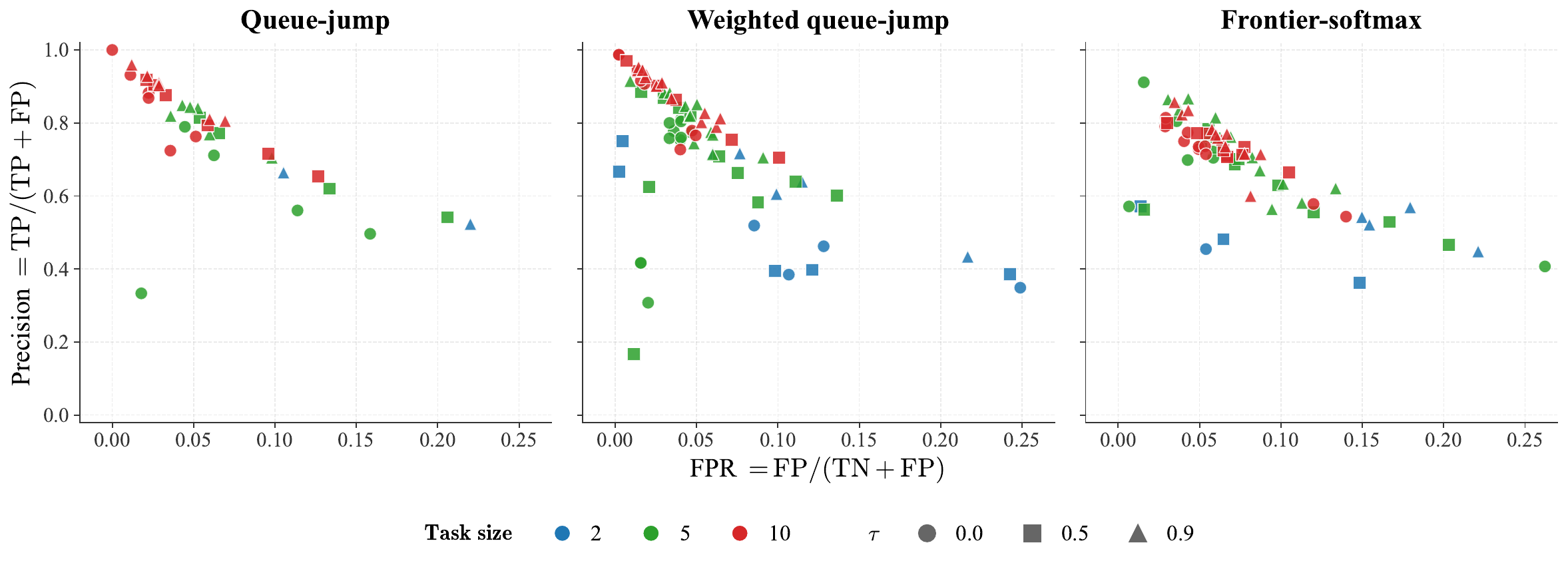}
\caption{\textbf{Experiment~A: precision versus false-positive rate for the posterior-mode estimator.}
Each point corresponds to one fitted configuration from the core HPO grid. Colors indicate the observation model,
and marker shapes indicate the choice-set size $s$. The upper-left corner corresponds to the most accurate structural
recovery (high precision, low false-positive rate).}
\label{fig:blocka_precision_fpr}
\end{figure}

In summary, Experiment~A shows that longer observed lists and stronger coupling matter more than modest differences in noise level, as these factors determine the amount of information in the data.
When the observed choice sets are informative, both exact likelihoods recover the latent structure well; frontier-softmax remains useful in lower-signal regimes, but its main advantage is computational rather than statistical.

\subsection{Experiment B: Synthetic HCPO study with clusters}
\label{sec:expB}

We now consider unlabeled data and ask two questions: can HCPO in Section~\ref{sec:Hpo-cluster} recover latent group structure, and does modeling that structure improve poset recovery relative to a forced single-poset fit using the posterior in \eqref{eq:po-posterior-vary-K}. We use a $2\times2\times2\times3$ design with coupling $\tau\in\{0.2,0.9\}$, likelihood in \{weighted queue-jump, frontier-softmax\}, clustering mode in \{list,assessor\} (the models in Sections~\ref{sec:cluster-lists} and \ref{sec:cluster-assessors} respectively), and three independent replicates per design cell.
Each of the 24 synthetic datasets has $M=10$ items, $A=10$ assessors, $G_{\mathrm{true}}=3$ true latent groups, noise probability $p=0.05$, $L=10$ lists per assessor, and random choice sets so list lengths are random from 5 and 10. Noise is low and lists are quite long relative to poset-size to these data are fairly informative in the sense of Experiment~A. The true model is always HCPO. After fitting HCPO, we extract a VI-optimal partition using the estimator of \cite{WadeGhahramani2018} (see Appendix~\ref{sec:VI-lower-bound-define}) and refit HPO conditional on that fixed partition to summarize the cluster-specific posets.

Figure~\ref{fig:blockb_precision_fpr} shows that both likelihoods recover clustered partial-order structure with high precision and low FPR, and many runs recover the correct number of clusters. Each plot has 12 points for two $\tau$-levels, assessor or list clustering and three replicates. The greatest variation across replicates is seen in the most heterogenous setting, weighted queue-jump with $\tau=0.2$ and list-level clustering, but is still small (all points FPR in $(0,0.05)$ and precision in $(0.9,1)$). Table~\ref{tab:block_b_results_main} shows the gains over the single-poset fit.
When coupling is weak ($\tau=0.2$), the benefit of clustering is substantial, because there is alot of heterogeneity across clusters. HCPO improves over the single-poset baseline by between $0.163$ and $0.529$ in cluster-weighted F1.
When coupling is strong ($\tau=0.9$), the groups become more similar and the gain from clustering is smaller, though still positive in every setting.
This is the regime in which one would expect a non-hierarchical model to be competitive.

\begin{figure}[t]
\centering
\includegraphics[width=0.8\linewidth]{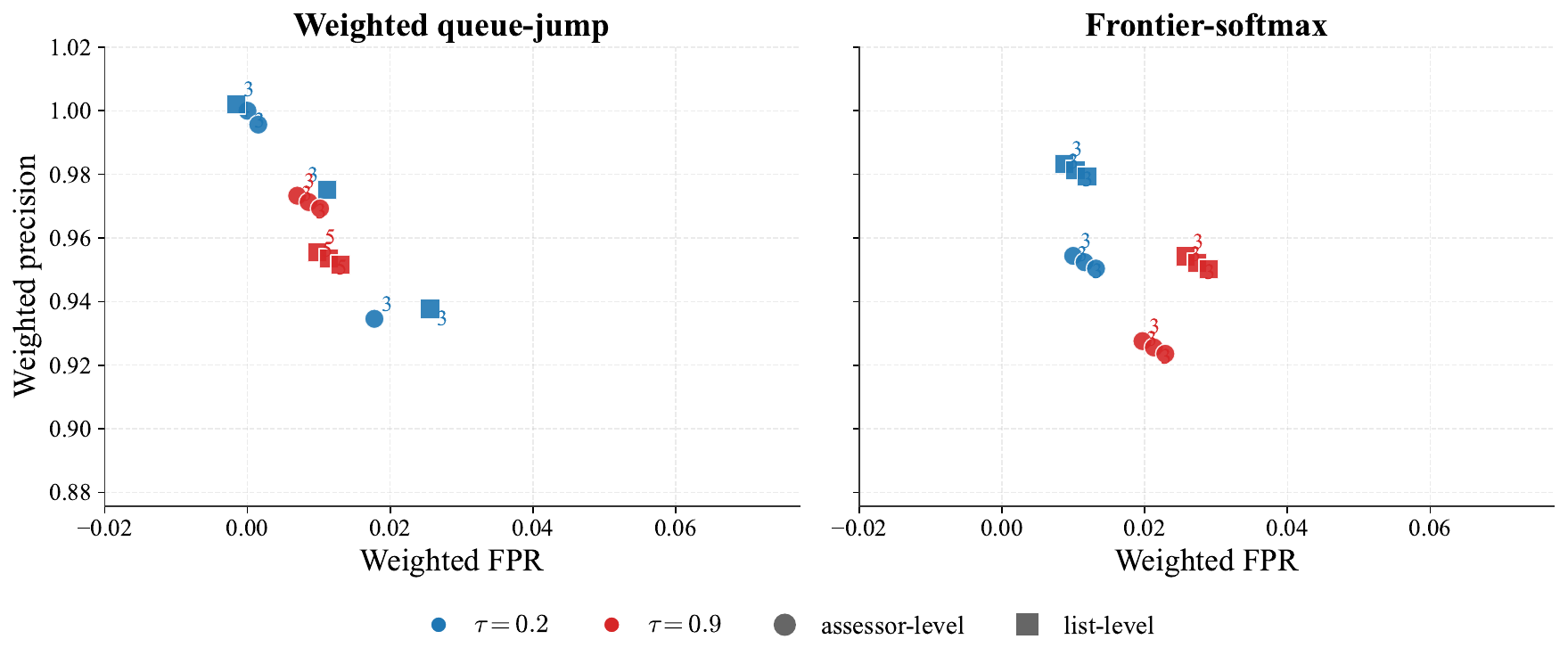}
\caption{\textbf{Experiment B: weighted precision versus weighted false-positive rate for HCPO.}
Each point corresponds to one fitted configuration in the Experiment~B grid. Panels correspond to the observation model,
and point labels indicate the inferred number of clusters. Configurations closer to the upper-left corner correspond
to more accurate structural recovery.}
\label{fig:blockb_precision_fpr}
\end{figure}

In real applications, the clustering unit is not something we get to choose - it depends on the data-lists having assessor labels.
List clustering gives slightly higher mean cluster-weighted F1 in Table~\ref{tab:block_b_results_main} and a larger average improvement over the single-poset baseline. Assessor clustering is easier as lists move in groups so there are fewer configurations: it recovers the correct number of groups $G_{\mathrm{true}}=3$ more often and is computationally cheaper on average.
Figure~\ref{fig:blockb_representative_example} illustrates a representative weighted queue-jump fit in assessor mode, where HCPO recovers both the true partition and the corresponding cluster-specific posets almost perfectly.
As a robustness check, Table~\ref{tab:blockb_vilb_vs_medoid_summary} compares the standard n-estimator in Appendix~\ref{sec:VI-lower-bound-define} with the frequency-weighted VI medoid from Section~\ref{sec:vi_freq_method}. When $\widehat{VI}_{UB}=0$ the partition is recovered exactly. When $\hat G=3$ but $\widehat{VI}_{UB}>0$, the number of groups is correct but the assignment of lists to clusters does (quite) match the truth. The medoid tends to select a slightly finer partition, whereas the VI upper bound estimator favors a more aggregated summary based on pairwise co-clustering structure. The substantive conclusions are unchanged.

\begin{table}[h]
\centering
\small
\caption{\textbf{Experiment~B: VI upper bound estimator (Appendix~\ref{sec:VI-lower-bound-define}) versus frequency-weighted VI medoid (Appendix~\ref{sec:vi_freq_method})}. 
Results (for weighted queue-jump only) are averaged over the three synthetic replications for each combination of coupling $\tau$ and clustering mode.}
\label{tab:blockb_vilb_vs_medoid_summary}
\setlength{\tabcolsep}{6pt}
\renewcommand{\arraystretch}{1.05}
\begin{tabular}{ccrrrrrr}
\toprule
$\tau$ & \textbf{Mode} & \textbf{VI-UB $\hat G$} & \textbf{Medoid $\hat G$} &
\textbf{$VI(\hat c_{\mathrm{VI}}, c_{\mathrm{true}})$} &\\
\midrule
0.2 & assessor & 3.00 & 3.33 & 0.00 & \\
0.2 & list     & 3.00 & 3.67 & 1.16 & \\
0.9 & assessor & 3.00 & 3.00 & 0.00 & \\
0.9 & list     & 5.00 & 5.00 & 1.10 &  \\
\bottomrule
\end{tabular}
\end{table}

\begin{figure}[t]
\centering
\includegraphics[width=\linewidth]{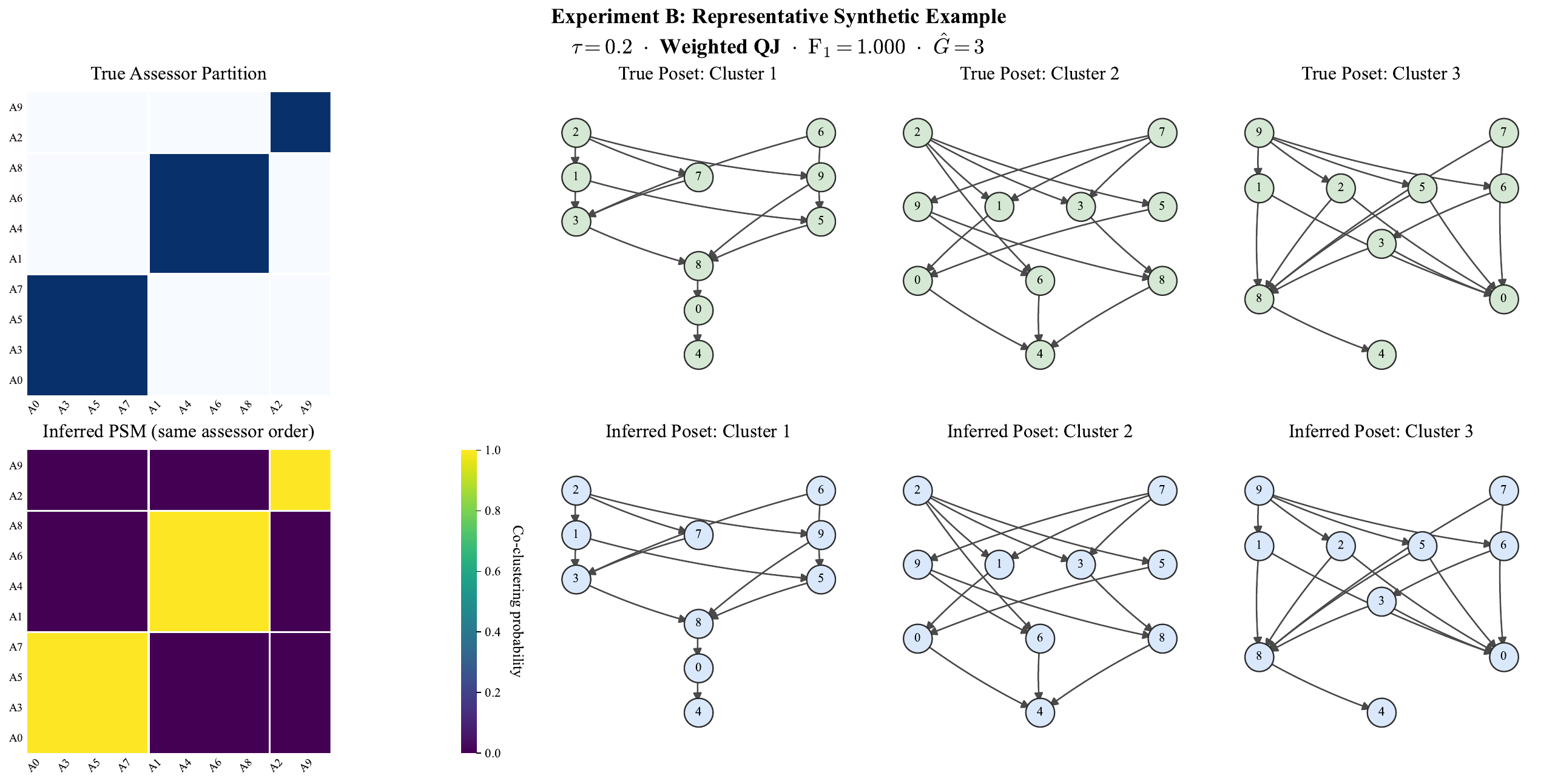}
\caption{\textbf{Representative synthetic example from Experiment B.}
Top left: true assessor partition. Bottom left: inferred posterior similarity matrix (PSM, defined in \eqref{eq:psm-define} in Appendix~\ref{sec:VI-lower-bound-define}) in the same assessor
order. The remaining panels compare the true and inferred cluster-specific partial orders for the three latent
groups. In this configuration (weighted queue-jump, assessor clustering, $\tau=0.2$), the inferred partition
recovers the true three-group structure and the inferred posets closely match the ground truth.}
\label{fig:blockb_representative_example}
\end{figure}

\begin{table}[t]
\centering
\small
\caption{\textbf{Experiment B: cluster recoverability and structure recovery under heterogeneous assessors.}
Values are means over three independently generated synthetic datasets per setting.
All runs use $G_{\text{true}}=3$ latent clusters.
$\hat{G}$ is the mean inferred number of clusters,
$F_1^{\mathrm{HCPO}}$ is HCPO cluster-weighted edge $F_1$,
$F_1^{\mathrm{PO}}$ is the single-poset baseline,
and $\Delta F_1 = F_1^{\mathrm{HCPO}} - F_1^{\mathrm{PO}}$.}
\label{tab:block_b_results_main}
\setlength{\tabcolsep}{5pt}
\renewcommand{\arraystretch}{1.05}
\begin{tabular}{cccrrrr}
\toprule
\textbf{$\tau$} & \textbf{Likelihood} & \textbf{Mode} & \textbf{$\hat{G}$} & \textbf{$F_1^{\mathrm{HCPO}}$} & \textbf{$F_1^{\mathrm{PO}}$} & \textbf{$\Delta F_1$} \\
\midrule
\multicolumn{7}{l}{\textit{Lower coupling: $\tau=0.2$}} \\
0.2 & Weighted QJ & Assessor & 3.33 & 0.966 & 0.629 & 0.338 \\
0.2 & Weighted QJ & List     & 3.00 & 0.959 & 0.430 & 0.529 \\
0.2 & Frontier-softmax QJ & Assessor & 3.00 & 0.850 & 0.687 & 0.163 \\
0.2 & Frontier-softmax QJ & List     & 3.00 & 0.979 & 0.735 & 0.243 \\
\midrule
\multicolumn{7}{l}{\textit{Higher coupling: $\tau=0.9$}} \\
0.9 & Weighted QJ & Assessor & 3.00 & 0.970 & 0.848 & 0.122 \\
0.9 & Weighted QJ & List     & 5.00 & 0.912 & 0.848 & 0.063 \\
0.9 & Frontier-softmax QJ & Assessor & 3.00 & 0.935 & 0.884 & 0.051 \\
0.9 & Frontier-softmax QJ & List     & 3.00 & 0.972 & 0.885 & 0.087 \\
\bottomrule
\end{tabular}
\end{table}

Experiment~B shows HCPO can recover latent clustering and cluster-specific partial orders with high accuracy under both likelihoods in this experiment with informative data. The statistical value of clustering is largest when the assessor groups are genuinely different, that is, when coupling between the latent group posets is weak.

\subsection{Experiment C: Aliyun cloud provisioning: Cloud-IaC-6 for the HPO model}
\label{sec:exp_aliyun}

We next consider agent traces, which motivate the partial-order framework from a workflow perspective rather than a preference-learning one.
Cloud-IaC-6 is a benchmark of six cloud-provisioning scenarios derived from an internal Aliyun platform, with 5--12 actions per scenario and 60 successful traces in total (54 LLM-generated traces plus 6 expert traces); see Table~\ref{tab:cloud_iac_6} and Appendix~\ref{app:cloud_iac_6_dataset}. There are in all $M=15$ different actions appearing across these traces. Some (like \texttt{CreateVpc}, \texttt{CreateSG}, \texttt{AuthorizeSG}, \texttt{CreateVSwitch}, \texttt{RunInstances}) appear across scenarios. Others (like \texttt{EIP}, \texttt{SLB}, \texttt{RDS}, \texttt{Redis}) are scenario-specific. These actions are the list items which we wish to order. The problem is a natural application of HPO because each task/scenerio is a group label,
and traces generated under the same scenerio are more similar to each other than to traces generated under other scenerios.
For each scenario a ground-truth dependency graph is available, so we can evaluate structural recovery directly.
The six scenarios are treated as related groups in an HPO fit, and we compare HPO with an independent single-PO baseline obtained by setting $G=1$ and $\tau=0$.

We study two trace regimes distinguished by IP-Cov, a measure of data information content defined in \eqref{eq:IP-Cov-define} in Appendix~\ref{app:metrics_ipcov}.
The first (reported in Figure~\ref{fig:cloud_iac_per_scenario}) uses the observed traces exactly as collected, where empirical IP-Cov varies widely across scenarios.
The second (reported in Figure~\ref{fig:aliyun_qualitative_recovery}) augments each scenario with additional sampled linear extensions until full incomparability coverage is reached.
This separation lets us distinguish failure caused by poor trace coverage from failure caused by the model itself.
We evaluate cover-edge F1, critical-pair F1, and trace feasibility (would poset-extensions complete the task); the formal definitions are given in Appendix~\ref{app:metrics_ipcov}.

Figure~\ref{fig:cloud_iac_per_scenario} reports the per-scenario quantitative comparison on Cloud-IaC-6 for the primary weighted queue-jump run with $10^6$ iterations. HPO matches or improves upon the single-PO baseline on nearly all scenario--metric pairs and achieves uniformly near-perfect recovery across the benchmark. The clearest gains appear in cover-edge F1 and trace feasibility on scenarios where the single-PO baseline under-recovers key ordering constraints. Overall, these results indicate that hierarchical pooling yields more reliable recovery of the workflow backbone across heterogeneous cloud-provisioning tasks.

\begin{figure}[t]
    \centering
    \includegraphics[width=\linewidth]{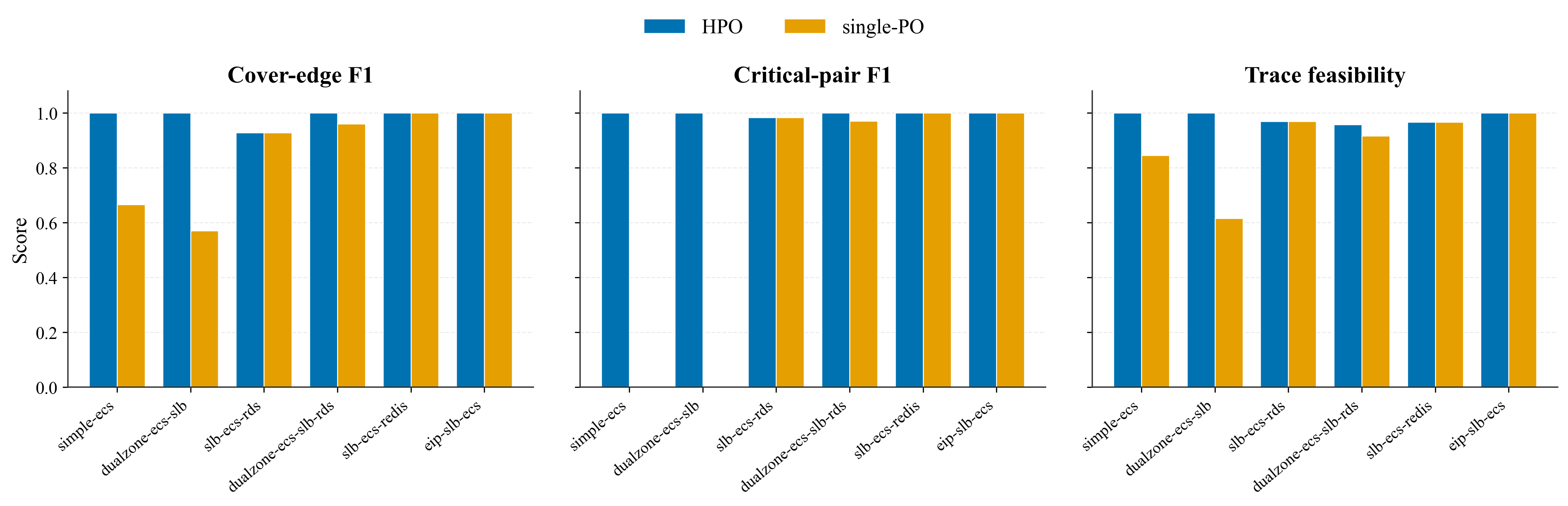}
    \caption{\textbf{Per-scenario comparison on Cloud-IaC-6 for the primary $10^6$-iteration weighted queue-jump run.}
    Comparison of HPO and an independent single-PO baseline across six Aliyun cloud-provisioning scenarios using cover-edge F1, critical-pair F1, and trace feasibility. The observed trace sets for these scenarios span empirical IP-Cov values from 0\% to 100\% (Table~\ref{tab:cloud_iac_6}). HPO matches or improves upon the single-PO baseline on nearly all scenario--metric pairs and achieves near-perfect recovery throughout.}
    \label{fig:cloud_iac_per_scenario}
\end{figure}

Figure~\ref{fig:aliyun_qualitative_recovery} provides the corresponding qualitative comparison under full incomparability-pair coverage ($\mathrm{IP\text{-}Cov}=1.0$). With the weighted queue-jump likelihood and a $10^6$-iteration inference budget, HPO recovers the dependency structure almost exactly: five of the six scenarios are recovered without error, and the remaining scenario differs only by a minor local discrepancy. This shows that, given sufficient ordering evidence, the hierarchical model reconstructs the precedence backbone of these workflows essentially perfectly.

\begin{figure}[t]
    \centering
    \includegraphics[width=\textwidth]{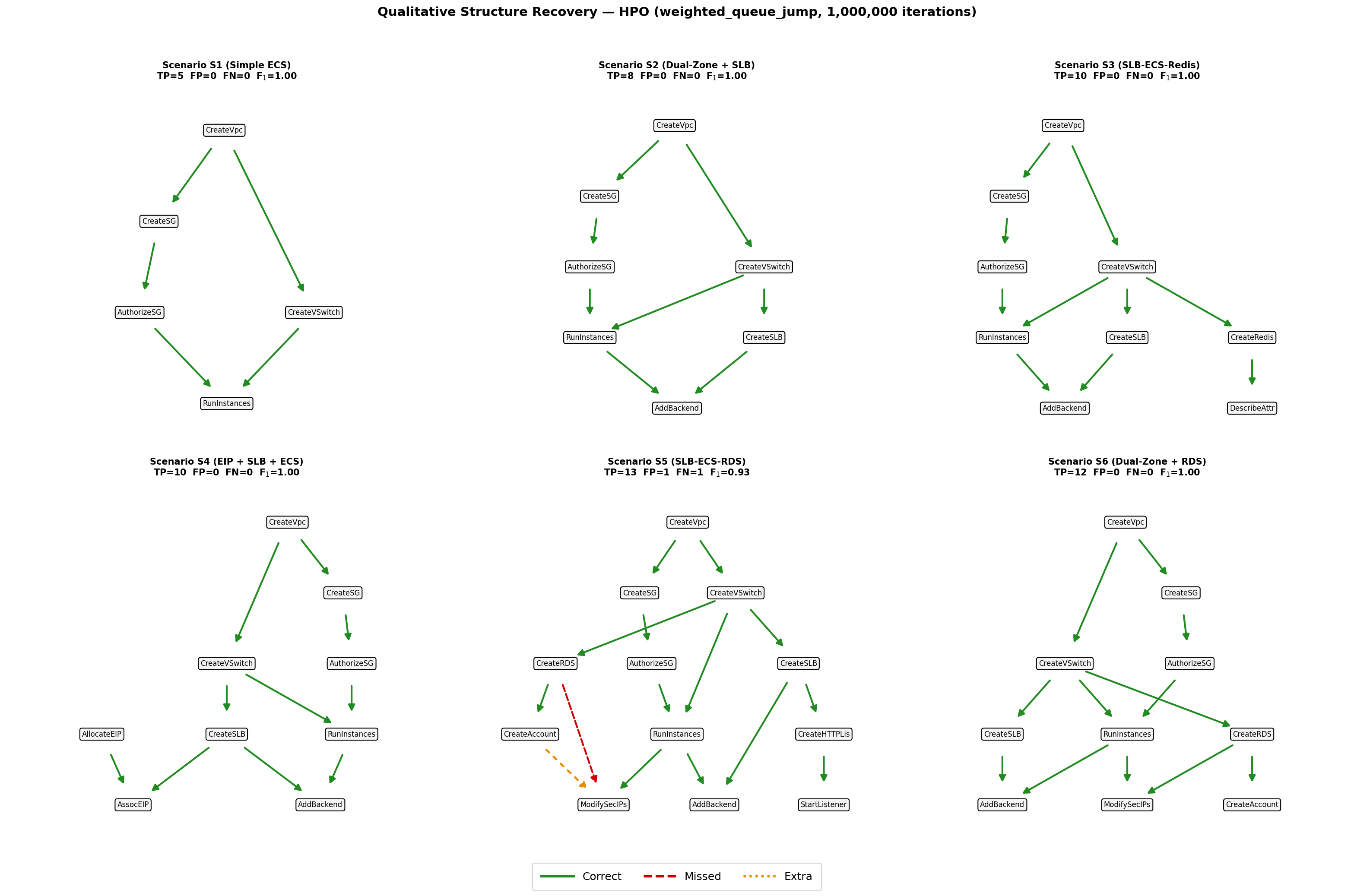}
    \caption{\textbf{Qualitative structure recovery on Cloud-IaC-6.}
    Comparison between the HPO-estimated cover graph and the ground-truth dependency graph for all six Aliyun cloud-provisioning scenarios under full incomparability-pair coverage ($\mathrm{IP\text{-}Cov}=1.0$), using the weighted queue-jump likelihood after $10^6$ iterations. Green solid edges denote correctly recovered cover relations, red dashed edges denote missed ground-truth edges, and orange dotted edges denote extra inferred edges. Each panel reports the numbers of true positives (TP), false positives (FP), false negatives (FN), and the resulting cover-edge F1 score. HPO achieves perfect recovery on five scenarios and near-perfect recovery on the remaining one.}
    \label{fig:aliyun_qualitative_recovery}
\end{figure}

The appendix compares all three likelihoods at the same $10^6$-iteration budget.
Both exact queue-jump variants attain essentially identical top performance, whereas frontier-softmax cuts runtime substantially (from 14.6 hours to 5.9 hours) but yields weaker structural recovery and lower feasibility. In this application the extension-based likelihoods are preferable when accurate workflow reconstruction is the main objective and traces are short enough to allow us to fit the Queue-Jump models. Overall, Experiment~C shows that, with sufficient ordering evidence, the extension-based likelihoods can recover the latent dependency structure with
high accuracy. The HCPO fit provides a coherent pooled summary across related scenarios.

\subsection{Experiment D: 3D bow-stroke sound preferences for the HCPO model}
\label{sec:sound}

The sound experiment returns to human preference data.  Listeners compared twelve short sound stimuli generated by controlled bow-stroke motions and rendered in 3D Ambisonics.
Each of the 46 assessors contributed 30 pairwise judgments, giving 1,380 A/B comparisons in total.
The dataset was introduced by \citet{crispino19}, who analyzed it with a Bayesian Mallows model; approximately 80\% of listeners exhibit cyclic pairwise patterns, so a model that permits incomparabilities is well motivated. For further details of the data see Appendix~\ref{subsec:data-sound}.

We fit HCPO with assessor-clustering under two observation models: weighted queue-jump and frontier-softmax.
Assessor heterogeneity is modeled with a Dirichlet-process prior with concentration $\vartheta=0.8$ and $d=0$ in \eqref{eq:Hpo-cluster-c-prior}.
Each chain runs for $3\times 10^6$ iterations, with the first half discarded as burn-in and thinning by 100 for posterior summaries.
The weighted queue-jump fit mixes more cleanly for the key structural and noise parameters (Appendix Table~\ref{tab:sound_obsmodel_diag}) and gives a better VI objective ($\widehat{VI}_{\mathrm{UB}}(\hat{c})$ in \eqref{eq:VI-objective}) in Table~\ref{tab:sound_partition_summary}, so we use it for the main structural summaries.
Its VI-optimal partition has $\widehat G_{\mathrm{VI}}=5$ clusters with sizes $(18,14,7,4,3)$, while frontier-softmax yields $\widehat G_{\mathrm{VI}}=6$ clusters with sizes $(16,14,5,4,4,3)$.
Figure~\ref{fig:sound_cluster_count} summarizes the posterior distribution of the number of clusters, Figure~\ref{fig:sound_psm} shows the posterior
similarity matrix with rows and columns reordered by the VI point estimate of the partition to make the block structure visible, and Figure~\ref{fig:sound_loglik_trace} shows the post-burn-in log-likelihood
traces.

For predictive comparison we use pairwise WAIC, treating each A/B judgment as one observation.
We compare HCPO with two total-order baselines, both fitted with three mixture components following \citet{crispino19}: a Plackett--Luce mixture and a Mallows mixture.
See Appendix~~\ref{sec:sound-model-comp-detail} for further details of the comparison setup and discussion of results. Appendix Figure~\ref{fig:sound_pairwise_waic} shows that HCPO dominates both baselines across a range of cluster counts $G$.
For the final fits, the weighted queue-jump HCPO model achieves $\mathrm{WAIC}=1593.52$, compared with $1628.44$ for frontier-softmax, $2046.72$ for the PL mixture, and $3182.48$ for the Mallows mixture.

Because cluster occupancy is sparse in some posterior draws, we also refit HPO conditional on the inferred cluster assignments using four chains (seeds 42/143/244/345).
Appendix Table~\ref{tab:sound_refit_restart_ess} shows satisfactory convergence for these refits, and the seed-345 refit, used for the Hasse summaries in Figure~\ref{fig:sound_final_inference}, attains a slightly lower WAIC of $1557.29$. This is consistent with a good choice of conditioning partition. See the Figure \ref{fig:sound_refit_posterior_diagnostics} for MCMC traces in this analysis.
After alignment with the PL and Mallows baselines, the HCPO summaries in Figure~\ref{fig:sound_final_inference} reveal a feature that total-order mixtures cannot represent directly: several clusters contain broad top layers together with only a small number of clearly inferior stimuli.
In other words, the model uses incomparability where the total-order baselines are forced to invent a strict within-cluster ranking.

\begin{figure}[H]
\centering
\includegraphics[width=\linewidth]{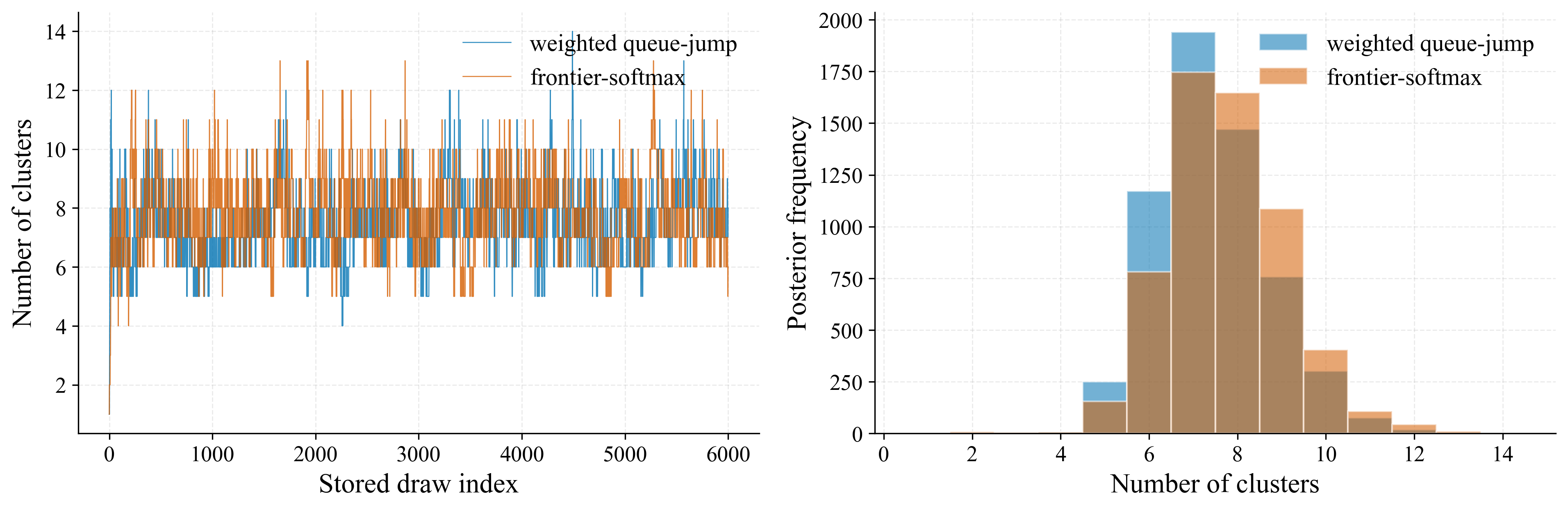}
\caption{\textbf{Sound data: posterior cluster count diagnostics.}
(\emph{Left}) Trace of the number of clusters across retained post burn-in draws fitting HCPO.
(\emph{Right}) Histogram of the posterior distribution of the number of clusters.
Blue: weighted queue-jump; orange: log-assessor/Frontier-softmax.}
\label{fig:sound_cluster_count}
\end{figure}

\begin{figure}[t]
\centering
\begin{subfigure}[t]{0.49\linewidth}
  \centering
  \includegraphics[width=\linewidth]{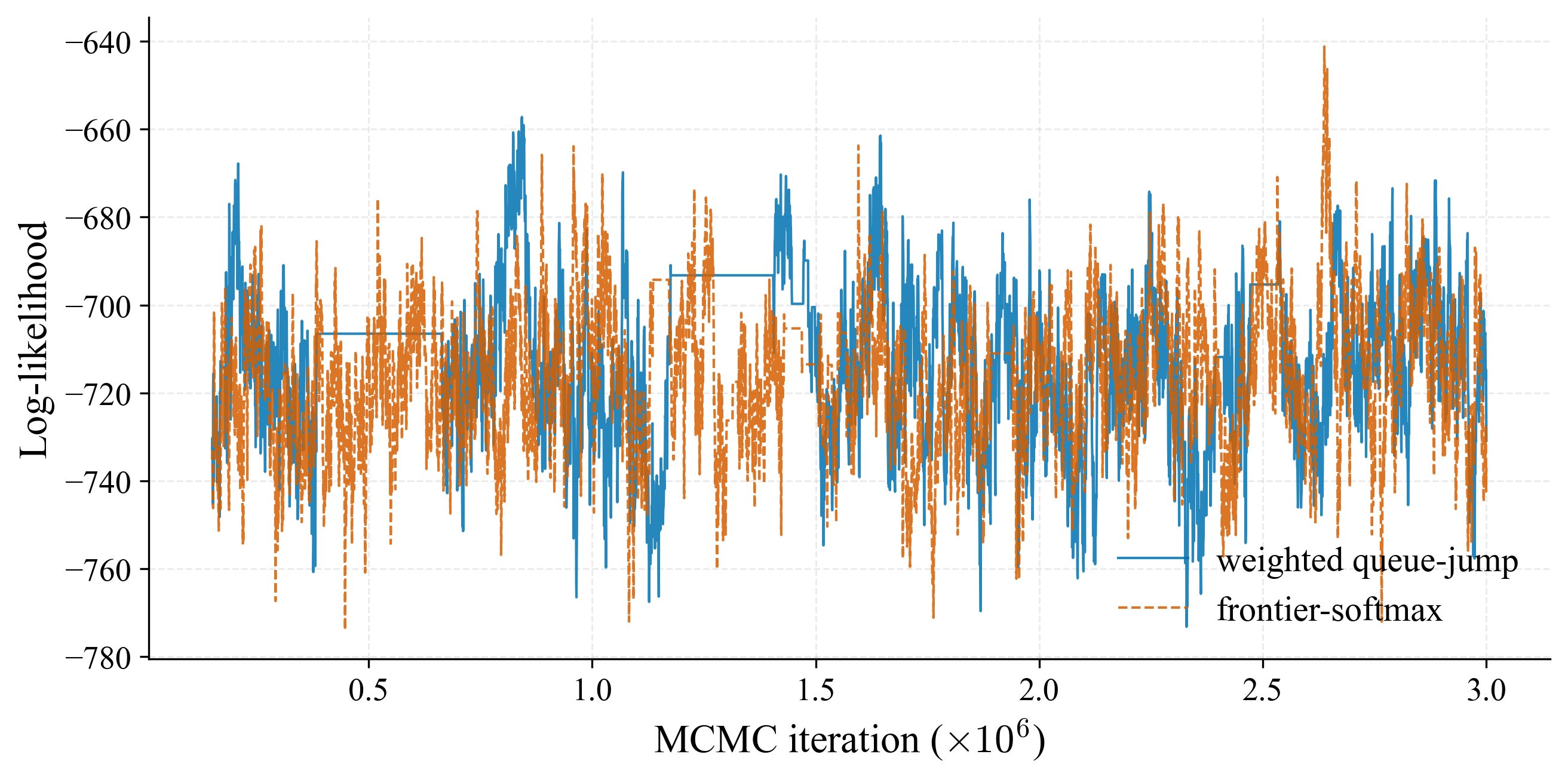}
  \caption{\textbf{Sound data: post burn-in log-likelihood traces.}
  Log-likelihood traces for the HCPO model under the two observation models
  (weighted queue-jump vs.\ log-assessor/Frontier-softmax). These traces provide a basic stationarity check
  and indicate that both likelihood variants explore similar fit regimes.}
  \label{fig:sound_loglik_trace}
\end{subfigure}\hfill
\begin{subfigure}[t]{0.49\linewidth}
  \centering
  \includegraphics[width=\linewidth]{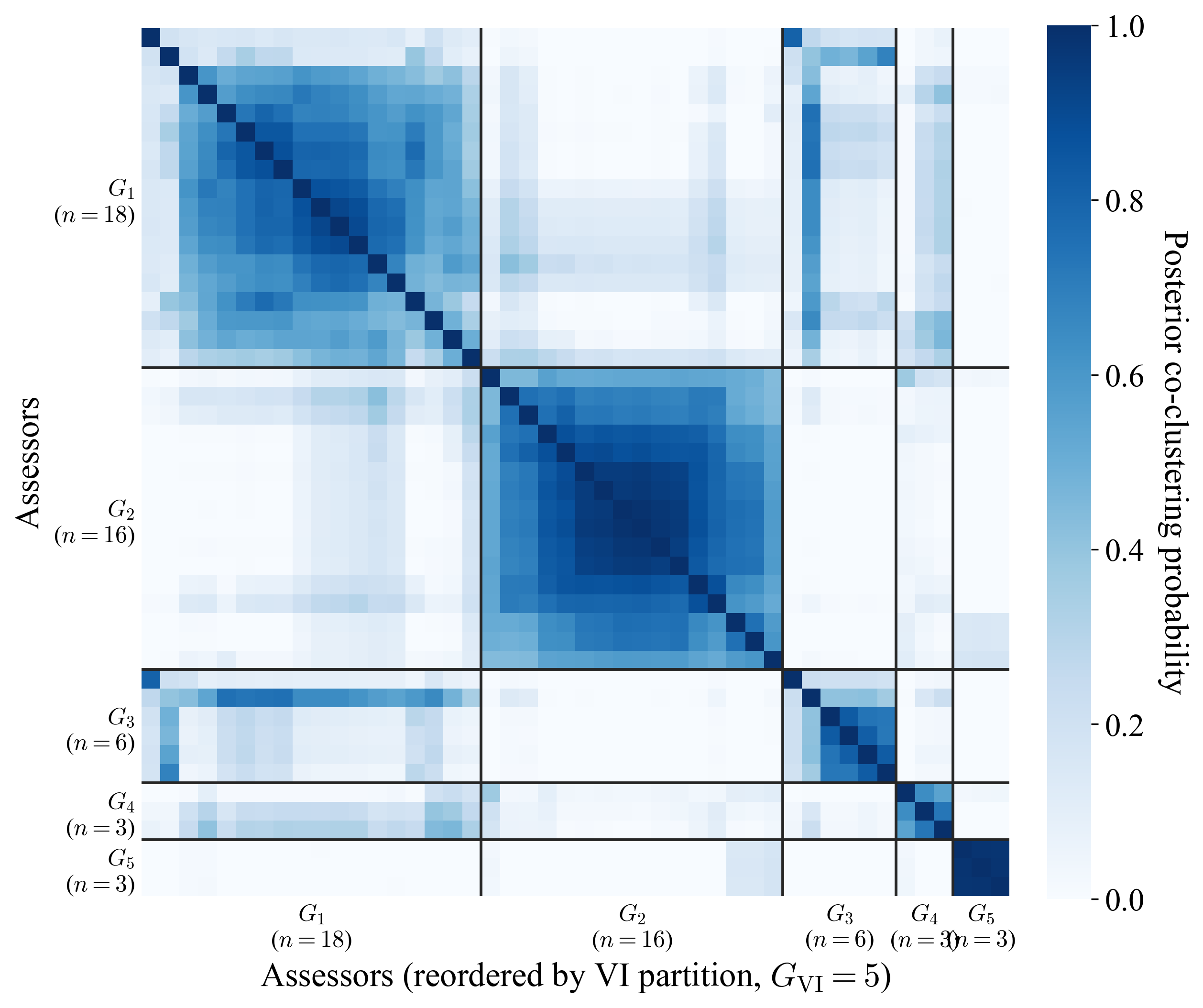}
  \caption{\textbf{Sound data: posterior similarity matrix (PSM).}
  Assessors reordered by the VI-optimal partition; darker cells indicate higher posterior co-clustering.}
  \label{fig:sound_psm}
\end{subfigure}
\end{figure}

\begin{figure}[t]
\centering
\includegraphics[width=\linewidth]{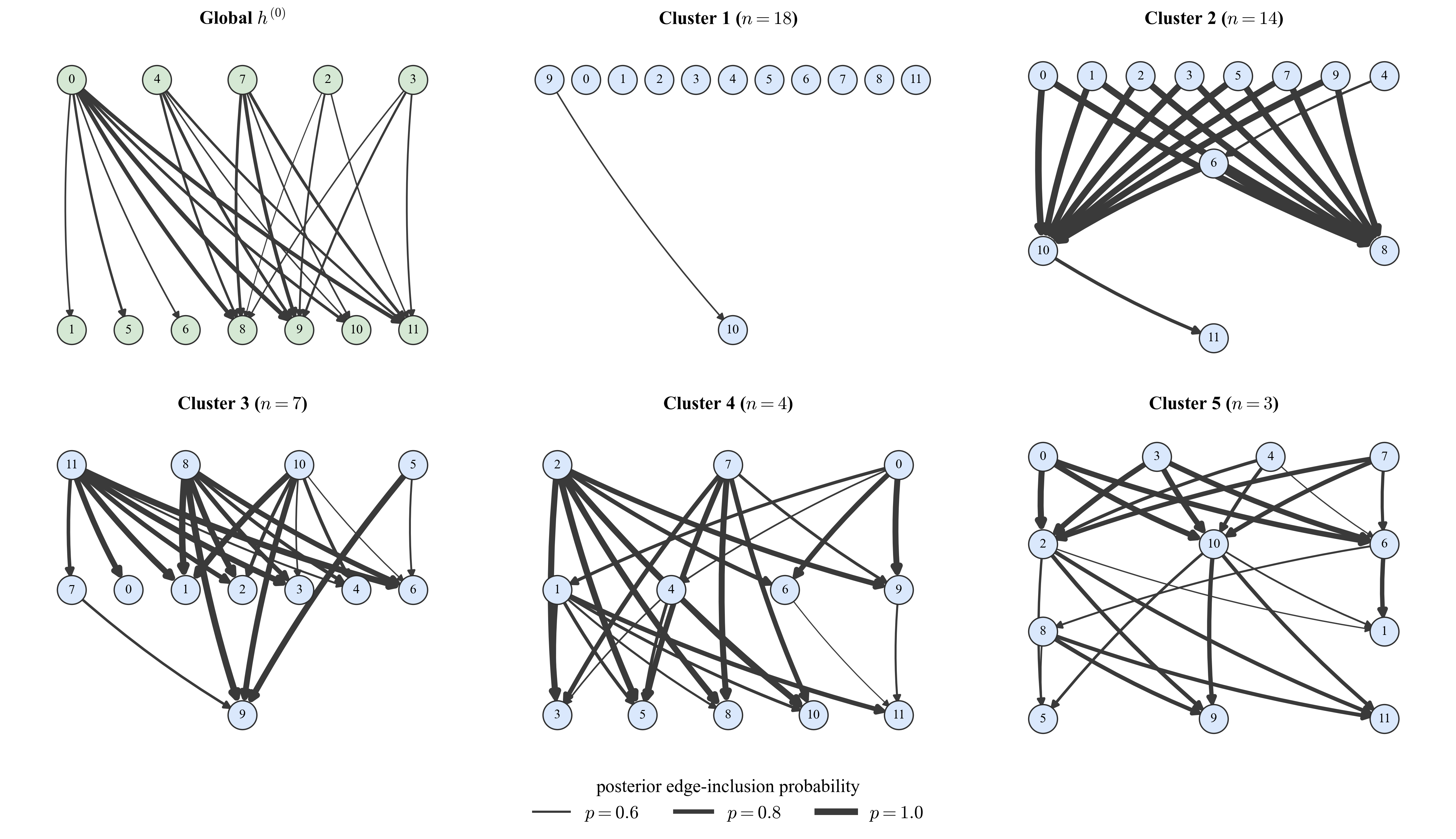}
\caption{\textbf{Sound data: final inferred cluster-level structure (seed 345 refit).}
Final HCPO summaries reported in the main text are computed from the seed 345 refit run; see Appendix Table~\ref{tab:sound_refit_restart_ess} for the corresponding multi-chain convergence/mixing diagnostics. Edge thickness encodes posterior inclusion probability.}
\label{fig:sound_final_inference}
\end{figure}

The sound experiment therefore highlights the main modeling advantage of HCPO on preference data: it improves predictive fit relative to total-order baselines while producing cluster summaries that are easy to interpret when listeners do not behave as if they were following strict global rankings.

\subsection{Experiment E: Ghana sweet-potato tricot preferences for the HCPO model}
\label{sec:ghana-sweet-potato}

Our final experiment studies the Ghana sweet-potato tricot data from the \texttt{SFPL} package \citep{hermes2024heterogeneous}.
The study involves $A=111$ participants (64 male and 47 female) and 13 local varieties.
Each participant reports rankings within small tricot choice sets, so the data are globally sparse but still informative about robust local comparisons.
This is a natural test case for HCPO: the earlier analyses of \citet{hermes2024heterogeneous} and \citet{MollicaTardella2017} both point to substantial heterogeneity, but they summarize that heterogeneity in terms of total-order or worth-based models.

We fit HCPO using a split--merge warm start.
Two short pre-training chains are run on the male and female subsets separately, their terminal states are merged, and the resulting pooled state initializes a $3\times 10^6$-iteration chain on the full dataset. 
We also use a Pitman--Yor partition prior with concentration parameter
$\vartheta=0.8$ and discount $d=0$, corresponding to the Dirichlet-process. We compare queue-jump, weighted queue-jump, and frontier-softmax. The log-likelihood traces and posterior cluster-count diagnostics are shown in Appendix Figures~\ref{fig:ghana_blockd_loglik} and~\ref{fig:ghana_blockd_cluster_count}; Appendix Table~\ref{tab:ghana_blockd_mcmc_diag} reports the corresponding posterior means and ESS values.
Across these diagnostics, weighted queue-jump mixes better and attains the highest log-likelihood values, so we use it as the primary model for interpretation.

Predictive comparison again favours HCPO over a single-poset baseline (PO). Table~\ref{tab:ghana_hcpo_hpo_vi_waic_compact} reports WAIC using the participant as the observation unit: each contribution is the likelihood of the full tricot response recorded for one participant, giving $n_{\mathrm{obs}}=111$.

HCPO improves on PO under all three likelihoods, with the largest gain under weighted queue-jump ($\Delta\mathrm{WAIC}=74.09$), followed by queue-jump ($37.15$) and frontier-softmax ($28.10$).
Only weighted queue-jump supports a non-trivial VI partition: it yields $G_{\mathrm{VI}}=3$ clusters of sizes $(63,40,8)$, whereas the other two likelihoods give $G_{\mathrm{VI}}=1$ (all 
$A=111$ assessors grouped together).
The corresponding PSM, shown in Appendix Figure~\ref{fig:app_ghana_psm_weighted_vi3}, shows a weakly supported but clear block structure after VI reordering.
Figure~\ref{fig:ghana_triad_waic_curves} shows WAIC estimates including comparisons with hierarchical Mallows and Plackett-Luce models. The conclusions are qualitatively similar to those reported in Section~\ref{sec:sound}.

\begin{table}[t]
\centering
\small

\caption{\textbf{Ghana: predictive fit (HCPO vs.\ PO) and VI partition summaries.}
WAIC is computed on $n_{\mathrm{obs}}=111$ assessor-level units using
$n_{\mathrm{draws}}=800$ draws. For each likelihood, PO and HCPO use the same
observation model, and positive
$\Delta=\mathrm{WAIC}_{\mathrm{PO}}-\mathrm{WAIC}_{\mathrm{HCPO}}$ favours
HCPO. $\widehat{VI}_{\mathrm{UB}}(\hat{c})$  is the Wade--Ghahramani upper bound for the selected partition;
smaller values are preferred within a likelihood, but are not used to compare
likelihoods.}
\label{tab:ghana_hcpo_hpo_vi_waic_compact}
\setlength{\tabcolsep}{5pt}
\renewcommand{\arraystretch}{1.05}
\begin{tabular}{l r r r c l r}
\toprule
\textbf{Lik.} &
\textbf{WAIC$_{\mathrm{HCPO}}$} &
\textbf{WAIC$_{\mathrm{PO}}$} &
\textbf{$\Delta$} &
$G_{\mathrm{VI}}$ &
\textbf{VI sizes} &
\textbf{$\widehat{VI}_{\mathrm{UB}}(\hat{c})$} \\
\midrule
WQJ  & 313.71 & 387.80 & 74.09 & 3 & 63,40,8 & $-206.17$ \\
QJ   & 357.47 & 394.62 & 37.15 & 1 & 111     & $-339.14$ \\
FSM & 365.97 & 394.08 & 28.10 & 1 & 111     & $-333.55$ \\
\bottomrule
\end{tabular}

\vspace{0.25em}
\footnotesize
Abbreviations: WQJ = weighted queue-jump; QJ = queue-jump; FSM = Frontier-softmax.
\end{table}

The inferred partition is not explained by gender.
A $\chi^2$ test finds no evidence of association between gender and cluster membership ($\chi^2=0.054$, $p=0.816$), and the cluster-specific male proportions are close to the overall cohort proportion.
Thus the latent groups identified by HCPO cut across the male/female split emphasized in earlier analyses.

After selecting weighted queue-jump, we refit HPO conditional on the VI partition to obtain cleaner cluster-level structural summaries.
The fixed-partition refit further improves predictive fit, reaching $\mathrm{WAIC}=236.49$; see Table~\ref{tab:ghana_refit_waic} and Appendix Table~\ref{tab:ghana_refit_waic_original}.
The refit diagnostics in Appendix Table~\ref{tab:ghana_refit_ess} and Figure~\ref{fig:app_ghana_refit_loglik_compare} show adequate mixing for these summary runs.
Figure~\ref{fig:ghana_hcpo_posets} displays the resulting global and cluster-specific partial orders.
The two larger clusters share a common backbone but differ in their top layers, while the smallest cluster is more weakly resolved and exhibits several competing top candidates.

\begin{table}[h]
\centering
\small
\caption{\textbf{Ghana: WAIC for fixed-$G$ HPO refits (post-selection).}
WAIC is computed on $A=111$ assessor-level units using $n_{\mathrm{draws}}=800$ draws. Lower is better.}
\label{tab:ghana_refit_waic}
\setlength{\tabcolsep}{6pt}
\renewcommand{\arraystretch}{1.05}
\begin{tabular}{l l r r r r}
\toprule
\textbf{Likelihood} & \textbf{Refit model} & \textbf{WAIC} & \textbf{elpd$_{\mathrm{WAIC}}$} & \textbf{$p_{\mathrm{WAIC}}$} & \textbf{s.e.(elpd)} \\
\midrule
weighted queue-jump       & HPO refit ($G=3$) & 236.49 & $-118.25$ & 32.10 & 11.87 \\
frontier-softmax & HPO refit ($G=1$) & 369.10 & $-184.55$ & 16.28 & 6.71 \\
queue-jump                & HPO refit ($G=1$) & 371.76 & $-185.88$ & 18.51 & 6.77 \\
\bottomrule
\end{tabular}
\end{table}

Although the Weighted Queue Jump model out-performed all others by WAIC, the reconstructed partial orders in Figure~\ref{fig:ghana_hcpo_posets} are close to being total orders. This, and the model complexity, make it surprising that the model fits better by a predictive criterion.

\begin{figure}[t]
\centering
\includegraphics[width=\linewidth]{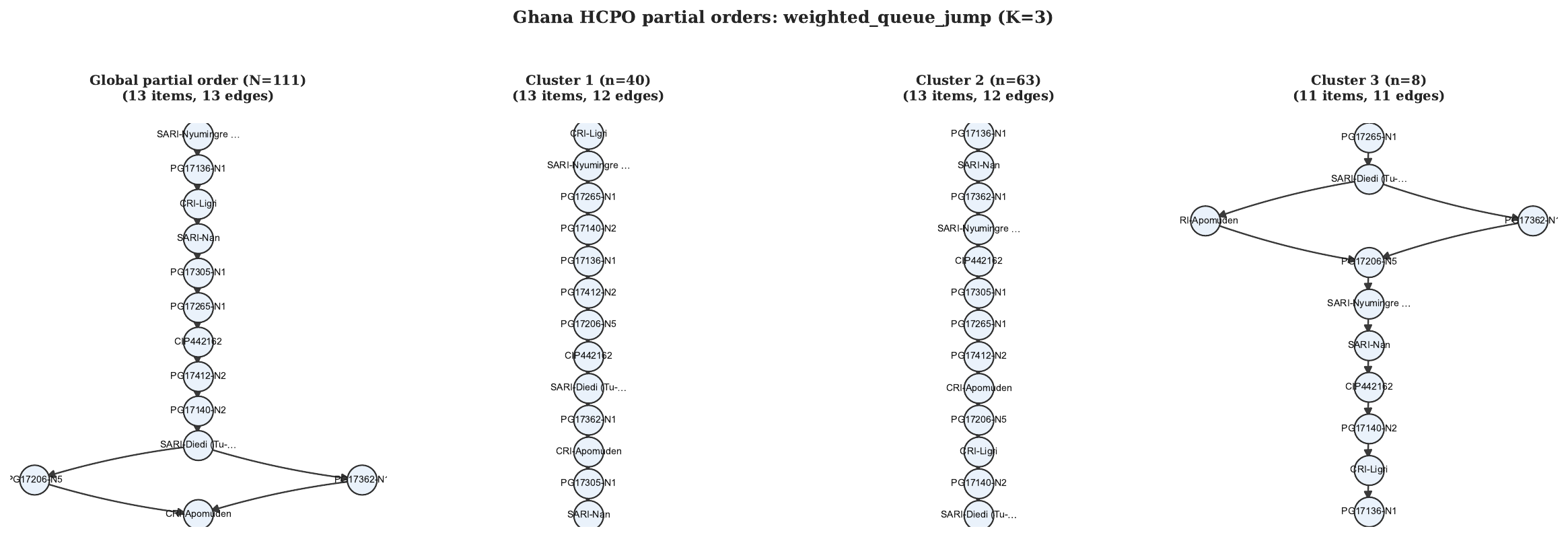}
\caption{\textbf{Ghana (weighted queue-jump): global and cluster-level HCPO partial orders.}
Global partial order and VI-optimal cluster-specific partial orders under the weighted queue-jump likelihood
($G_{\mathrm{VI}}=3$, cluster sizes $(63,40,8)$).}
\label{fig:ghana_hcpo_posets}
\end{figure}

\section{Conclusions}

We have introduced HPO and HCPO, hierarchical Bayesian models for rank data in which the latent consensus objects are partial orders rather than total orders.
The framework covers grouped data with observed labels, extends naturally to latent clustering when labels are unknown, and nests Plackett--Luce as a special case.
This makes it possible to borrow strength across related groups while still allowing genuine incomparabilities in the latent preference structure.

The experiments support three main conclusions.
First, recoverability depends primarily on the information content of the observed lists: longer choice sets and stronger coupling improve reconstruction much more than modest changes in noise level.
Second, explicit clustering is most helpful when latent groups are genuinely different; when coupling is strong and heterogeneity is low, a simple single-poset model can do well.
Third, on real data the partial-order formulation is not only statistically competitive but also practically useful: on Cloud-IaC-6 it improves recovery of executable workflow graphs, and on the sound and Ghana data it yields better predictive fit than total-order baselines.

The experiments also clarify the role of the observation models.
Queue-jump and weighted queue-jump are the most faithful to the underlying linear-extension semantics and usually give the strongest structural recovery, but they are computationally expensive.
Frontier-softmax is much faster and remains useful in larger problems, although the speed gain typically comes with some loss of structural accuracy.

Several limitations remain.
Inference using our generally favored linear-extension-based Queue-Jumping models is intractable when the observed list length becomes large. In an MCMC setting with millions of likelihood evaluations this caps the list lengths at about 20 depending on other details of the problem. However, looking across published data, long lists are exceptional. In Table~\ref{tab:intro-data-summary} we give data set sizes for the work cited in the introduction. These problems are all tractable. Other issues are that recovery degrades sharply when the data provide little local ordering information, and cluster summaries can require post-processing or conditional refits when posterior occupancy is diffuse.
Future work is focusing on new applications for animal dominance hierarchies, richer covariate structures, integration with differentiable relaxations for fast posterior approximation on long traces, and broader comparisons with alternative structured ranking models.

\acks{We thank Valeria Vitelli and the Oslo group for helpful conversations.}

\bibliography{ref} 

\appendix

\section{Notation and Definitions}
\label{sec:notation}

Let \(\mathcal{M}\) denote the universe of items, and let
\(\mathcal{M}_0 \subseteq \mathcal{M}\) be the observed item set, with
\[
m_0 \coloneqq |\mathcal{M}_0|.
\]
Assessors are indexed by \(a=1,\dots,A\).

\subsubsection*{Observed data}

\begin{description}[leftmargin=4.2cm,labelwidth=3.8cm,style=multiline]

    \item[\(A\)]
        Number of assessors.

    \item[\(\mathcal{M}\), \(\mathcal{M}_0\)]
        Full item universe and observed item set, respectively.

    \item[\(\mathcal{B}_{\mathcal{M}_0}\)]
        Collection of all non-empty subsets of \(\mathcal{M}_0\):
        \[
        \mathcal{B}_{\mathcal{M}_0}
        \coloneqq
        \{B \subseteq \mathcal{M}_0 : B \neq \emptyset\}.
        \]

    \item[\(L_a\)]
        Number of observed lists contributed by assessor \(a\).
        In balanced designs, we may write \(L_a = L\).

    \item[\(S_{a\ell}\)]
        Choice set presented to assessor \(a\) in task/list \(\ell\),
        where \(S_{a\ell} \in \mathcal{B}_{\mathcal{M}_0}\).

    \item[\(s_{a\ell}\)]
        Size of the choice set \(S_{a\ell}\), i.e.
        \[
        s_{a\ell} \coloneqq |S_{a\ell}|.
        \]

    \item[\(y_{a\ell}\)]
        Observed complete order returned by assessor \(a\) on task \(\ell\):
        \[
        y_{a\ell} = (y_{a\ell 1},\dots,y_{a\ell s_{a\ell}})
        \in \mathcal{L}(S_{a\ell}).
        \]

    \item[\(\mathcal{M}_a\)]
        Set of items seen by assessor \(a\):
        \[
        \mathcal{M}_a \coloneqq \bigcup_{\ell=1}^{L_a} S_{a\ell}.
        \]

    \item[\(\mathcal{L}(B)\)]
        Set of all complete (linear) orders on a finite set \(B\).

    \item[\(\mathcal{D}\)]
        Full observed dataset:
        \[
        \mathcal{D} \coloneqq \{(S_{a\ell}, y_{a\ell}) :
        a=1,\dots,A,\ \ell=1,\dots,L_a\}.
        \]

\end{description}

\subsubsection*{Latent partial-order structure}

\begin{description}[leftmargin=4.2cm,labelwidth=3.8cm,style=multiline]

    \item[\(K\)]
        Latent dimension of the utility representation.

    \item[\(\mathbf{U}^{(0)}_j \in \mathbb{R}^K\)]
        Global latent utility vector for item \(j \in \mathcal{M}_0\).

    \item[\(\mathbf{U}^{(0)} \in \mathbb{R}^{m_0 \times K}\)]
        Global latent utility matrix, obtained by stacking
        \(\mathbf{U}^{(0)}_j\) over \(j \in \mathcal{M}_0\).

    \item[\(\mathbf{U}^{(a)}_j \in \mathbb{R}^K\)]
        Assessor-specific latent utility vector for item \(j\).

    \item[\(\mathbf{U}^{(a)} \in \mathbb{R}^{m_0 \times K}\)]
        Assessor-specific latent utility matrix.

    \item[\(\mathcal{H}(B)\)]
        Set of all partial orders on a finite set \(B\).

    \item[\(\Phi(\mathbf{U})\)]
        Mapping from a latent utility representation \(\mathbf{U}\) to the
        induced partial order on \(\mathcal{M}_0\).

    \item[\(h^{(0)}\)]
        Global / consensus partial order:
        \[
        h^{(0)} \coloneqq \Phi(\mathbf{U}^{(0)}) \in \mathcal{H}(\mathcal{M}_0).
        \]

    \item[\(h^{(a)}\)]
        Assessor-specific partial order:
        \[
        h^{(a)} \coloneqq \Phi(\mathbf{U}^{(a)}) \in \mathcal{H}(\mathcal{M}_0).
        \]

    \item[\(\mathbf{A}(h)\)]
        Adjacency matrix of a partial order \(h\), with entries
        \[
        A_{ij}(h) = \mathbf{1}\{i \succ_h j\}.
        \]

    \item[\(\mathrm{LE}(h)\)]
        Set of linear extensions of the partial order \(h\).

    \item[\(N_{\mathrm{LE}}(h)\)]
        Number of linear extensions of \(h\):
        \[
        N_{\mathrm{LE}}(h) \coloneqq |\mathrm{LE}(h)|.
        \]

    \item[\(N_{\mathrm{LE},j}(h)\)]
        Number of linear extensions of \(h\) in which item \(j\) appears first.

\end{description}

\subsubsection*{Likelihood-related quantities}

\begin{description}[leftmargin=4.2cm,labelwidth=3.8cm,style=multiline]

    \item[\(h=(A_h,\succ_h)\)]
        Generic partial order on support set \(A_h \subseteq \mathcal{M}_0\).

    \item[\(y=(y_1,\dots,y_T)\)]
        Generic observed list over the support set \(A_h\).

    \item[\(y_{<t}\)]
        Prefix of the list \(y\) up to time \(t-1\):
        \[
        y_{<t} \coloneqq (y_1,\dots,y_{t-1}).
        \]

    \item[\(R_t\)]
        Remaining items at step \(t\):
        \[
        R_t \coloneqq A_h \setminus \{y_1,\dots,y_{t-1}\}.
        \]

    \item[\(F_t(h;y_{<t})\)]
        Frontier at step \(t\), i.e. the set of admissible undominated items.

    \item[\(S_t(a)\)]
        Number of remaining successors dominated by item \(a\) at step \(t\).

    \item[\(\beta\)]
        Inverse-temperature parameter in the observation model.

    \item[\(\epsilon\)]
        Observation-model noise probability.

    \item[\(p\)]
        Synthetic data-generating noise probability, when used in simulation studies.

    \item[\(\tau\)]
        Coupling parameter controlling similarity between global and local structures.

\end{description}

\subsubsection*{Clustering and partition summaries}

\begin{description}[leftmargin=4.2cm,labelwidth=3.8cm,style=multiline]

    \item[\(u=1,\dots,n\)]
        Generic clustering unit. Depending on the model, a unit may be an assessor or a list.

    \item[\(c=(c_1,\dots,c_n)\)]
        Cluster labels for the \(n\) units.

    \item[\(\mathcal{G}(c)\)]
        Set of occupied clusters:
        \[
        \mathcal{G}(c) \coloneqq \{g : n_g > 0\}.
        \]

    \item[\(G\)]
        Number of occupied clusters:
        \[
        G \coloneqq |\mathcal{G}(c)|.
        \]

    \item[\(n_g\)]
        Size of cluster \(g\):
        \[
        n_g \coloneqq |\{u : c_u = g\}|.
        \]

    \item[\((d,\vartheta)\)]
        Pitman--Yor discount and concentration hyperparameters.

    \item[\(\mathcal{C}_n\)]
        Space of all partitions of \(n\) units.

    \item[\(P_{uv}\)]
        Posterior similarity matrix entry:
        \[
        P_{uv} \coloneqq \Pr(c_u = c_v \mid \mathcal{D}).
        \]

    \item[\(VI(c,c')\)]
        Variation-of-information distance between two partitions.

    \item[\(\hat c\)]
        Representative clustering, typically chosen as a VI-optimal summary.

\end{description}

\subsubsection*{Trace notation for Cloud-IaC-6}

\begin{description}[leftmargin=4.2cm,labelwidth=3.8cm,style=multiline]

    \item[\(\pi\)]
        Execution trace.

    \item[\(A_\pi\)]
        Set of actions appearing in trace \(\pi\).

    \item[\(\operatorname{pos}_\pi(a)\)]
        Position of action \(a\) in trace \(\pi\).

\end{description}

\section{Background}

Here we include some background with the aim of making the paper self-contained.

\subsection{Context-independent preference}
\label{app:context-independent-models}

If preferences are context independent then we must have
\begin{equation}\label{eq:context-independence-as-sum}
p_S(y)=\sum_{z\in\C_\M} p_\M(z)\mathbb{I}_{z(S)=y},\ S\in\B_\M.    
\end{equation}
It further holds that all the marginals are consistent, so we can replace $\M$ by $S'$ in \eqref{eq:context-independence-as-sum} for any $S'\supseteq S$ and \eqref{eq:context-independence-as-sum} still holds. This kind of {\it marginal consistency} appears in other forms below but has the special meaning of context-independence for distributions over complete preference orders. 

Context independence need not hold if we simply write down $p_S(\cdot)$ separately for each $S\in\B_\M$.  
For example, the Mallows model \citep{mallows1957}, the contextual repeated selection (CRS) model \citep{seshadri20} and the partial order model in this section do not, in general, satisfy \eqref{eq:context-independence-as-sum}. 
On the other hand, the Plackett-Luce model in the next section is context-independent in the sense of Definition~\ref{def:context-independence} and in the equivalent sense of the Luce Axiom of Choice \citep{luce1959possible}.
\cite{ragain18} and \cite{seshadri19,seshadri20} develop models for context-dependent choice, which generalise the Plackett-Luce model. 
Our approach builds transitivity into a context-dependent preference model whilst transitivity plays no role in \cite{seshadri20}. 

The model in \eqref{eq:po_lkd_full_noisefree} expresses context-dependent preferences. 
This is illustrated by the example in Figure~\ref{fig:not-consistent-example}.
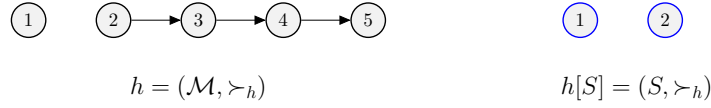
\begin{figure}
    \centering
    \scalebox{0.65}{\begin{tikzpicture}
      \node (1) [param] {$1$};
      \node (2) [param, right=of 1, xshift=0.5cm] {$2$};
      \node (3) [param, right=of 2, xshift=0.5cm] {$3$};
      \node (4) [param, right=of 3, xshift=0.5cm] {$4$};
      \node (5) [param, right=of 4, xshift=0.5cm] {$5$};

      \edge [->] {2} {3};
      \edge [->] {3} {4};
      \edge [->] {4} {5};

    \node (LEs) [text height=1cm, below=of 3, yshift=1cm] {\Large $\displaystyle h=(\M,\succ_h)$};

      \node (1s) [param, right=of 1, xshift=10cm, draw=blue] {$1$};
      \node (2s) [param, right=of 1s, xshift=0.5cm, draw=blue] {$2$};

      \node (LEs4) [text height=1cm, below=of 1s, xshift=1.15cm,yshift=1cm] {\Large $\displaystyle h[S]=(S,\succ_h)$};
    \end{tikzpicture} }
    \caption{In the partial order (left) on $\M=\{1,2,3,4,5\}$,  $1$ is maximal in one of five possible linear extensions. In the suborder on $S=\{1,2\}$, 1 is maximal on one of two.}
    \label{fig:not-consistent-example}
\end{figure}
Suppose $h \in \H_\M$ is the poset $(\{1, \dots, M\},\succ_h)$, in which the suborder $(\{2, \dots, M\},\succ_h)$ is a complete order $2\succ_h 3 \succ_h\dots\succ_h M$, but item $1$ has no order relation to any other elements.
Now $h$ has $M$ linear extensions (item $1$ can go in any position in the complete order, and the order of the rest is fixed), so item $2$ is the maximal element in $M-1$ out of $M$ linear extensions. 
If the choice set is $S = \{1,2\}$, and the list is realised as a suborder on $S$, then we pick a random linear extension $y'\sim p_\M(\cdot|h)$ and get $\Pr(1\succ_{y'[S]} 2)=1/M$. 
However, the suborder $h[S]=(\{1,2\},\succ_h)$ is empty (it has two unordered elements), so it has two linear extensions, $\ell^{(1)}$ and $\ell^{(2)}$, with $1 \succ_{\ell^{(1)}} 2$ and $2 \succ_{\ell^{(2)}} 1$. 
If the list is realised on $S$, then $y\sim p_S(\cdot|h[S])$ and $\Pr(1\succ_y 2)=1/2$, so $y'[S]$ and $y$ do not have the same distribution so preferences are context dependent. 

\subsection{PCMC models}\label{app:PCMC-models}


The likelihood $p_S(y|h)$ can be parameterised as a Pairwise choice Markov chain (PCMC, \cite{ragain16_pcmc}).
In PCMC models, an order $y$ is built sequentially by drawing $y_{1}$ from the equilibrium of a stochastic process on $S$ with $m\times m$ rate matrix $Q$ and repeating for $S\setminus\{y_1\}$ and so on. Taking $Q_{j_1,j_2}(\alpha)\propto |\L_{j_2}[h]|^{1-\alpha}|\L_{j_2}[h]|^{-\alpha}$ gives detailed balance, 
$q_S(j_1|h)Q_{j_1,j_2}(\alpha)=q_S(j_2|h)Q_{j_2,j_1}(\alpha)$, for any fixed $\alpha\in [0,1]$ so the $i$'th 
element $y_i$ in $y$ is selected with probability $q_{y_{i:m}}(y_i|h[y_{i:m}])$. This could be used to 
embed an underlying partial order structure in PCMC and thereby build transitivity into the PCMC setup, or conversely, to relax the observation model in \eqref{eq:po_lkd_full_noisefree}, by taking a prior on entries of $Q$ which is concentrated 
on $Q(\alpha)$. We have not pursued this generalisation. 

\subsection{Further remarks on the Plackett-Luce model}\label{app:PL-further-remarks}

The PL model is context independent: if $y\sim p_\M(\cdot|\alpha_\M)$ then $y[S]\sim p_S(\cdot|\alpha_S)$. We prove this using Lemma~\ref{lem:PL-gumbel}. \cite{hunter04} demonstrates the required relationship between marginal distributions, \eqref{eq:PL-MC} below, by direct computation.

\begin{theorem} \label{thm:PL-MC} \citep{hunter04}
For all $S\in\B_\M$ and all $y\in \C_S$
\begin{equation}\label{eq:PL-MC}
    p_S(y|\alpha_S)=\sum_{z\in \C_\M} p_\M(z|\alpha_\M)\,\mathbb{I}_{z[S]=y}.
\end{equation}
\end{theorem}
\begin{proof}
This follows from the Gumbel construction. 
If $G_j\sim \mbox{Gumbel}(\alpha_j)$ independent for $j=1,\dots,M$, then $y(G)\sim \mbox{PL}(\alpha_\M;\M)$ with $y(G)\in \C_\M$ by Lemma~\ref{lem:PL-gumbel}. 
In our notation, $y(G)[S]\in \C_S$ is the $S$-suborder of $y(G)$, so $\Pr(y(G)[S] = y)$ is given by the RHS of \eqref{eq:PL-MC}. 
Let $G_S=\{G_j: j\in S\}$. 
Removing elements from $G$ does not change the relative ordering for the elements that remain, so $y(G)[S]=y(G_S)$
with $G_j\sim \mbox{Gumbel}(\alpha_j),\ j\in S$ jointly independent and so by Lemma~\ref{lem:PL-gumbel}, $y(G_S)\sim p_S(\cdot|\alpha_S)$, the distribution on the LHS of \eqref{eq:PL-MC}.
\end{proof}

Context independence in the sense of Definition~\ref{def:context-independence} is a weaker condition than context-independence in the sense of the ``Luce Axiom of Choice'' (LAC, \cite{luce77}). As we build up a list by choosing the sequence of elements one at a time according to the factorisation in \eqref{eq:PL-sequential}, the odds of choosing $j_1$ next over $j_2$ in the sequence is the same for every choice set $S$ containing $j_1$ and $j_2$. In the notation of \eqref{eq:PL-def-one-step} that is
\begin{equation}\label{eq:PL-LAC}
\frac{q_S(j_1|\alpha_S)}{qF_S(j_2|\alpha_S)}=\frac{q_\M(j_1|\alpha_\M)}{q_\M(j_2|\alpha_\M)}.  
\end{equation}
This is easy to check for $q_S(j|\alpha_S)$ in \eqref{eq:PL-def-one-step} and in fact the converse is also true: the Plackett-Luce distribution in \eqref{eq:PL-def-full} is the only distribution over orders satisfying the LAC \citep{luce1959possible}. 
Since LAC $\Rightarrow$ \eqref{eq:PL-def-full} $\Rightarrow$  Lemma~\ref{lem:PL-gumbel} $\Rightarrow$ Theorem~\ref{thm:PL-MC}, it follows that LAC implies context independence in the sense of Definition~\ref{def:context-independence}. However, there are many marginally consistent families of distributions over orders which satisfy \eqref{eq:context-independence-as-sum} but not \eqref{eq:PL-LAC}: for example if we take $p_\M(z)=|\L[h]|^{-1}\mathbb{I}_{z\in\L[h]}$ for some $h\in\H_\M$ and {\it define} $p_S(y)$ by \eqref{eq:context-independence-as-sum} then the family $p_S,\ S\in\B_\M$ is marginally consistent by construction, but it is not Plackett-Luce.


Theorem~\ref{thm:PL-MC} is helpful when data $Y=(Y_1,\dots,Y_N)$ are observed as suborders with $Y_i=Y'_i[S_i]$, and $Y'_i\sim p_{\M}(\cdot|\alpha)$ jointly independent for $i=1,\dots,N$ given $\alpha$, as this is the same as generating $\tilde Y_i$ on the choice set $S_i$. 
Suppose the family of priors $\{\pi_{\alpha,S}(\alpha_S),\ \alpha_S\in \R^{m}\}_{S\in\B_\M}$ for $\alpha\in \R^M$ is marginally consistent, so $\alpha\sim \pi_{\alpha,\M}$ implies $\alpha_S\sim \pi_{\alpha,S}$ (for example, when the components of $\alpha$ are a priori independent).
We can drop $\alpha_j$ from the analysis if $j$ does not appear in at least one $S_i$.  
Let $\M'=\cup_i S_i$. 
 The marginal posterior for $\alpha\in \R^{|\M'|}$ is
\begin{align}
    \pi_{\alpha,\M'}(\alpha_{\M'}|Y)&=\int \pi_{\alpha,\M}(\alpha_\M|Y) d\alpha_{\M\setminus\M'}\nonumber\\
    &\propto \int \pi_\alpha(\alpha_\M) \prod_{i=1}^N \left[\sum_{Y'_i\in\C_\M} p_\M(Y'_i|\alpha)\mathbb{I}_{Y'[S_i]=Y_i}\right]\,d\alpha_{\M\setminus\M'}\nonumber\\
    \intertext{after applying Theorem~\ref{thm:PL-MC} to do the sum, the $\alpha_j,\ j\in {\M\setminus\M'}$ are no longer in the product,}
    &\propto \int \pi_\alpha(\alpha_\M) \,d\alpha_{\M\setminus\M'} \prod_{i=1}^N p_{S_i}(Y_i|\alpha_{S_i})\nonumber\\
    \intertext{so by the assumed marginal consistency of the $\alpha$-prior}
    &\propto \pi_{\alpha,\M'}(\alpha_{\M'})\prod_{i=1}^N p_{S_i}(Y_i|\alpha_{S_i}).\label{eq:PL-posterior}
\end{align}
The posterior $\pi_{\alpha,\M'}(\alpha_{\M'}|Y)$ only involves parameters for items in $\M$ that we have data for, so sampling and estimation will be more efficient than if we had to target $\pi_{\alpha,\M}(\alpha_\M|Y)$.  



\section{Proofs and further remarks for Sections~\ref{sec:PO-single-inference} and \ref{sec:Hpo-labeled}}
\subsection{Proof of Corollary~\ref{cor:po-MC} and data realised as a suborder in Section~\ref{sec:po-foundation-models}}\label{app:PO_model_single_PO}

Let $\pi_S(\cdot|\rho,\beta),\ S\in \B_\M$ be defined as in Theorem~\ref{thm:po-prior-gumbel-U} with $S$ replacing $\M$.

\noindent{\bf Corollary~\ref{cor:po-MC}}
{\it The family of prior distributions $\pi_S(\cdot|\rho,\beta),\ S\in \B_\M$, is marginally consistent, that is if $h\sim \pi_\M(\cdot|\rho,\beta)$, then $h[S]\sim \pi_S(\cdot|\rho,\beta)$.}

\begin{proof}
    For a $M\times K$ matrix $X$ let $X_{S,:}$ denote the sub-matrix with rows $X_{j,:},\ j\in S$. Since the maps are all applied element by element, $h[S]=h(\eta(\U,\beta)_{S,:})$, and
    \[
    \eta(\U,\beta)_{S,:}=G^{-1}(\Phi(\U_{S,:})) + X_{S,:}\beta\, 1^T_K.
    \]
    As $\U_{j,:}\sim N(0,\Sigma_\rho), j\in\M$ are iid over $j$, so are $\U_{j,:},\ j\in S$ (no dependence on index $j\in S$) so these are just the random variables and maps defining $\pi_S(\cdot|\rho,\beta)$.
\end{proof}

If the data are realised as a suborder then the assessor formed an order $Y'_i\sim p(\cdot|h(\eta(\U,\beta)))$ on $\M$, but we only observe $Y_i=Y'_i[S_i],\ i=1,\dots,N$. In this case we have parameters $\rho$ and $\beta$ as before, but now $\U\in \R^{M\times K}$ and we introduce auxiliary variables $Y'_i\in \{y\in \C_\M: y[S_i]=Y_i\},\ i=1,\dots,N$ that agree with $Y_i$ on the suborder. The posterior is
\begin{equation}\label{eq:po-posterior-realised-as-suborder}
    \pi_{\M}(\rho,\U,\beta,Y'|Y)\propto \pi_R(\rho)\pi_B(\beta)\left[\prod_{j\in\M} N(\U_{j,:};0_K,\Sigma_\rho)\right]\ \times\ \left[\prod_{i=1}^N p(Y'_i|h(\eta(\U,\beta)))\right],
\end{equation}
and we have to work on the full space $\U\in \R^{M\times K}$ and further marginalise over $Y'$.

\subsection{Proof of Theorem~\ref{thm:Hpo-prior-gumbel-U}}\label{app:hpo-prior-theorem-proof}

\noindent{\bf Theorem~\ref{thm:Hpo-prior-gumbel-U}} {\it (Hierarchical Partial Order prior) For $\alpha_\M=X\beta$ and $\Sigma_\rho$ as in Section~\ref{sec:PO-single-inference}, for $0<\tau\le 1$ and sets $\M_a\in\B_\M,\ a\in \A$, let $\M_0=\cup_{a=1}^A \M_a$ and
    \begin{align} 
    \u^{(0)}_{j,:}&\sim N(0,\Sigma_\rho),\quad\mbox{independent for each $j\in \M_0$,}\label{eq:Hpo-prior-U0-app}\\
    \u^{(a)}_{j,:}|\u^{(0)}_{j,:} &\sim N\left(\tau\u^{(0)}_{j,:},\, (1-\tau^2)\Sigma_\rho\right)\quad\mbox{independent for each $a\in\A$ and $j\in \M_a$,}\label{eq:Hpo-prior-Ua-app}\\ 
    \eta^{(a)}_{j,:}&=G^{-1}(\Phi(\u^{(a)}_{j,:}))+\alpha_j 1^T_K \quad\mbox{for $a=0,1,\dots,A$ and $j\in \M_a$ and}\label{eq:Hpo-prior-eta-app}\\
    h&=h(\eta(\u,\beta))\quad\mbox{for $\eta(\u,\beta)=(\eta^{(0)},\eta^{(1)},\dots,\eta^{(A)})$ and $h^{(a)}=h(\eta^{(a)})$.}\nonumber
\end{align}
For $h\in \H_{\M_{0:A}}$ let $\pi_{\M_{0:A}}(h|\rho,\beta,\tau)=E_{\u}(\mathbb{I}_{h(\eta(\u,\beta))=h})$ be the resulting prior for the poset hierarchy $h$. The marginal prior for $h^{(a)}$ is $\pi_{\M_{0:A}}(h^{(a)}|\rho,\beta,\tau)$. 
\begin{enumerate}
    \item (single PO marginals) For $a=0,1,\dots,A$, $\pi_{\M_{0:A}}(h^{(a)}|\rho,\beta,\tau)=\pi_{\M_a}(h^{(a)}|\rho,\beta)$ where $\pi_{\M_a}$ is the single partial-order prior given in \eqref{eq:po-single-h-prior};
    \item (PL hierarchy at $K=1$) when $K=1$, 
    $h^{(a)}\sim PL(\alpha_{\M_a},\M_a)$ for $a\in\A$;
    \item (independence at $\tau=0$) when $\tau=0$, $\pi_{\M_{0:A}}(h|\rho,\beta,\tau)=\prod_{a=0}^A\pi_{\M_{a}}(h^{(a)}|\rho,\beta,\tau)$;
    \item (matching suborders at $\tau=1$) for each $a\in \A$, $\displaystyle\lim_{\tau\to 1}\pi_{\M_{0:A}}(h^{(a)}|\rho,\beta,\tau)=\mathbb{I}_{h^{(a)}=h^{(0)}[\M_a]}$;
    \item (support) if $K\ge \lfloor M/2\rfloor$ then $\pi_{\M_{0:A}}(h|\rho,\beta,\tau)>0$ for all $h\in\H_{\M_{0:A}}$; 
    \item (marginal consistency) If $H\sim \pi_{\M^{A+1}}(\cdot|\rho,\beta,\tau)$ then $H[s_{0:A}]\sim \pi_{s_{0:A}}(\cdot|\rho,\beta,\tau)$ for every $s_{0:A}\in \B_\M^{A+1}$ with $s_0=\cup_{a=1}^A s_a$.
\end{enumerate}
}
\begin{proof}
We can write
\begin{equation}\label{eq:Hpo-Ua-U0-epsilon}
    \u^{(a)}_{j,:}=\tau\u^{(0)}_{j,:}+\epsilon^{(a)}_{j,:}
\end{equation}
with 
\[
\epsilon^{(a)}_{j,:}\sim N(0_K,(1-\tau^2)\Sigma_\rho)
\]
independent of everything else so the (marginal) covariance of $\u^{(a)}_{j,:}$ is
\[
\mbox{cov}(\tau\u^{(0)}_{j,:}) + (1-\tau^2)\Sigma_\rho=\Sigma_\rho,
\]
and the mean of $\u^{(a)}_{j,:}$ is zero. It follows that
marginally,
\[
\u^{(a)}_{j,:}\sim N(0_{m_a},\Sigma_\rho).
\]
Since this is the distribution of $\U_{j,:}$ in the single-PO model in Theorem~\ref{thm:po-prior-gumbel-U} and all else is the same, this establishes results 1 and 2 of Theorem~\ref{thm:Hpo-prior-gumbel-U}. Result 3 follows from \eqref{eq:Hpo-Ua-U0-epsilon} also, as $\mbox{cov}(\u^{(a)}_{j_1,:},\u^{(0)}_{j_2,:})=\tau\Sigma_\rho$ and $\mbox{cov}(\u^{(a)}_{j_1,:},\u^{(a')}_{j_2,:})=\tau^2\Sigma_\rho$ so if $\tau=0$ then $h(\eta(\u^{(a)},\beta))$ and $h(\eta(\u^{(a')},\beta))$ are functions of jointly normal independent random variables for $a\ne a'$. 

Define $f(x)=\Phi^{-1}(G(x))$ for $x\in \R$ and $f^{(a)}_{j,k}=f(\eta^{(a)}_{j,k}-\alpha_j)$ so that
\[
f^{(a)}_{j,k}-f^{(0)}_{j,k}=\u^{(a)}_{j,k}-\u^{(0)}_{j,k}.
\] 
Conditional on $\u^{(0)}$ we have  
\[
f^{(a)}_{j,k}-f^{(0)}_{j,k}|u^{(0)}\sim N\bigl((\tau-1)\u^{(0)}_{j,k},\, (1-\tau^2)\bigr).
\]
When $\tau\to 1$ this gives $f^{(a)}_{j,k}\stackrel{P}{\to}f^{(0)}_{j,k}$ and hence $\eta^{(a)}_{j,k}\stackrel{P}{\to}\eta^{(0)}_{j,k}$ as $f$ is continuous and invertible. This gives $h^{(a)}\stackrel{P}{\to} h^{(0)}$ (result 4). Results 5 and 6 in Theorem~\ref{thm:Hpo-prior-gumbel-U} may be shown using the same reasoning as Corollaries~\ref{cor:po-support} and \ref{cor:po-MC} respectively. 
\end{proof}

\section{Derivation of the VI Upper Bound Objective}
\label{app:vi_lb_derivation}

\subsection{VI Upper Bound Objective}\label{sec:VI-lower-bound-define}

Following the Bayesian cluster analysis framework of \citet{WadeGhahramani2018}, 
we minimize the posterior expected variation of information 
\citep{meila2007comparing}. Since this objective is computationally 
intractable, we instead optimize the VI upper bound derived using 
Jensen's inequality \citep{WadeGhahramani2018}.

Given posterior samples of clusterings, we compute the posterior similarity matrix (PSM), whose entries 
\begin{equation}\label{eq:psm-define}
P_{ij} = \Pr(c_i = c_j \mid D)
\end{equation}
represent the posterior probability that items $i$ and $j$ are co-assigned to the same cluster. 

Following the Bayesian decision-theoretic framework for clustering \citep{WadeGhahramani2018}, the optimal partition $\hat{c}^*$ minimizes the posterior expected variation of information (VI):
\[
\hat{c}^* = \arg\min_{\hat{c}} \mathbb{E}[VI(c, \hat{c}) \mid D].
\]

Because directly evaluating this expectation across the vast discrete space of partitions is computationally prohibitive, we instead optimize a tractable upper bound derived via Jensen's inequality \citep{WadeGhahramani2018}.
The tractable upper bound, $VI_{\mathrm{crit}}(\hat{c})$, can be computed directly from the PSM:
\begin{equation}\label{eq:VI-objective}
\widehat{VI}_{\mathrm{UB}}(\hat{c}) = \sum_{i=1}^n \log |C_{\hat{c}_i}| - 2 \sum_{i=1}^n \log \left( \sum_{j \in C_{\hat{c}_i}} P_{ij} \right),
\end{equation}
where $C_{\hat{c}_i}$ denotes the cluster containing item $i$ under the candidate partition $\hat{c}$. Intuitively, the first term acts as a regularization penalty against partitions with numerous small clusters, while the second term rewards configurations where members exhibit high posterior co-clustering probabilities. A derivation of this bound is detailed in the Appendix \ref{sec:vi-lower-bound-derive}.

\subsection{Frequency-weighted posterior VI Upper Bound partition estimator}
\label{sec:vi_freq_method}

Rather than optimizing the VI upper bound derived from the posterior similarity matrix (PSM),
we want the inference among the partitions the sampler actually visited, which one is closest on average to the whole posterior distribution.
Let
\[
c^{(1)},\dots,c^{(S)}
\]
be the post burn-in partition samples.
We estimate the posterior VI risk of a candidate partition $\hat{c}$ by the Monte Carlo average
\[
\widehat{R}(\hat{c})
=
\frac{1}{S}\sum_{s=1}^S VI\!\left(c^{(s)},\hat{c}\right).
\]

A direct optimization over the full space of partitions is still infeasible, so we restrict the
search to the set of unique sampled partitions.
To do this, each sampled partition is first converted to a canonical label-invariant representation,
for example by relabeling clusters in order of first appearance.
Let
\[
u_1,\dots,u_M
\]
denote the resulting unique sampled partitions, with multiplicities
\[
n_m = \sum_{s=1}^S \mathbf{1}\!\left(c^{(s)}=u_m\right),
\qquad
w_m=\frac{n_m}{S}.
\]
Then the empirical posterior VI risk can be written as
\[
\widehat{R}(\hat{c})
=
\sum_{m=1}^M w_m\, VI(u_m,\hat{c}),
\]
and our representative partition is chosen as
\[
\hat{c}^*
=
\arg\min_{\hat{c}\in\{u_1,\dots,u_M\}}
\sum_{m=1}^M w_m\, VI(u_m,\hat{c}).
\]

This estimator can be interpreted as a \emph{VI medoids} of the sampled posterior partitions.
Its main advantage over the PSM-based VI upper bound is that it explicitly accounts for the
posterior frequency of full partitions.
In particular, a partition that is only rarely visited by the MCMC sampler cannot be selected
simply because it appears favorable under pairwise co-clustering probabilities alone.

\subsection{Bayesian Clustering Decision Rule and Variation of Information}\label{sec:vi-lower-bound-derive}

In this appendix we derive the VI upper bound estimator in Appendix~\ref{sec:VI-lower-bound-define}.

Let $c = (c_1, \dots, c_n)$ denote the true clustering and $P(c \mid D)$ denote the posterior distribution over partitions given data $D$. Under Bayesian decision theory, the optimal point estimate $\hat{c}^*$ minimizes the expected posterior loss. When employing the variation of information \citep{meila2007comparing} as our loss function—i.e., $L(c, \hat{c}) = VI(c, \hat{c})$—the objective is to minimize $\mathbb{E}[VI(c, \hat{c}) \mid D]$.

The variation of information between a true partition $c$ and a candidate partition $\hat{c}$ is defined using Shannon entropy $H$ and mutual information $I$:
\[
VI(c, \hat{c}) = H(c) + H(\hat{c}) - 2I(c, \hat{c}).
\]


To express this in terms of cluster sizes, let $C_k = \{i : c_i = k\}$ denote cluster $k$ in partition $c$. The probability that a randomly selected item belongs to $C_k$ is $p_k = |C_k|/n$. The entropy of partition $c$ is therefore:
\[
H(c) = -\sum_k p_k \log p_k = -\sum_k \frac{|C_k|}{n} \log \frac{|C_k|}{n}.
\]
Similarly, by defining $\hat{C}_l = \{i : \hat{c}_i = l\}$ for the candidate partition, the candidate entropy $H(\hat{c})$ is:
\[
H(\hat{c}) = -\sum_l \frac{|\hat{C}_l|}{n} \log \frac{|\hat{C}_l|}{n}.
\]

To compute the mutual information $I(c, \hat{c})$, we define the overlap counts $n_{kl} = |C_k \cap \hat{C}_l|$, representing the number of items shared between true cluster $k$ and candidate cluster $l$. The joint probability of an item falling into both clusters is $p_{kl} = n_{kl}/n$. The mutual information is defined as:
\[
I(c, \hat{c}) = \sum_{k,l} p_{kl} \log \frac{p_{kl}}{p_k p_l}.
\]
Substituting the cluster-count expressions into the mutual information yields:
\[
I(c, \hat{c}) = \sum_{k,l} \frac{n_{kl}}{n} \log \frac{n_{kl}/n}{(|C_k|/n)(|\hat{C}_l|/n)}.
\]
Simplifying the ratio inside the logarithm gives:
\[
\frac{n_{kl}/n}{(|C_k|/n)(|\hat{C}_l|/n)} = \frac{n_{kl} n}{|C_k| |\hat{C}_l|}.
\]
Therefore, the mutual information becomes:
\[
I(c, \hat{c}) = \sum_{k,l} \frac{n_{kl}}{n} \log \frac{n_{kl} n}{|C_k| |\hat{C}_l|}.
\]


Substituting the explicit expressions for entropy and mutual information into the VI definition and multiplying the entire objective by $n$ provides the scaled formulation:
\[
n \, VI(c, \hat{c}) = \sum_k |C_k| \log \frac{n}{|C_k|} + \sum_l |\hat{C}_l| \log \frac{n}{|\hat{C}_l|} - 2 \sum_{k,l} n_{kl} \log \frac{n_{kl} n}{|C_k||\hat{C}_l|}.
\]
Taking the exact posterior expectation of this formulation, $\mathbb{E}[VI(c, \hat{c}) \mid D]$, is computationally intractable because it requires summing over the combinatorial space of all possible partitions. To derive a scalable alternative, we isolate the terms dependent on $\hat{c}$ and convert the cluster-level sums into item-level sums over $i=1, \dots, n$.

Notice that a cluster-level sum of the form $\sum_l |\hat{C}_l| f(|\hat{C}_l|)$ can be rewritten as an item-level sum $\sum_{i=1}^n f(|\hat{C}_{\hat{c}_i}|)$. Applying this identity to the mutual information overlap term, the count $n_{kl}$ represents the size of the intersection $C_{c_i} \cap \hat{C}_{\hat{c}_i}$ for an item $i$. We can express this intersection size using indicator functions as $\sum_{j \in \hat{C}_{\hat{c}_i}} \mathbb{I}(c_i = c_j)$. 

This allows us to rewrite the expected value of the key intractable cross-term as an expectation over individual items:
\[
\mathbb{E} \left[ \sum_{k,l} n_{kl} \log n_{kl} \mid D \right] = \mathbb{E} \left[ \sum_{i=1}^n \log \left( \sum_{j \in \hat{C}_{\hat{c}_i}} \mathbb{I}(c_i = c_j) \right) \mid D \right].
\]
Because the logarithm is a concave function, we can apply Jensen's inequality, $\mathbb{E}[\log X] \le \log \mathbb{E}[X]$, to push the expectation inside the logarithmic term, providing a lower bound:
\[
\mathbb{E} \left[ \sum_{i=1}^n \log \left( \sum_{j \in \hat{C}_{\hat{c}_i}} \mathbb{I}(c_i = c_j) \right) \mid D \right] \le \sum_{i=1}^n \log \left( \sum_{j \in \hat{C}_{\hat{c}_i}} \mathbb{E}[\mathbb{I}(c_i = c_j) \mid D] \right).
\]
Noting that the expectation of the indicator function $\mathbb{E}[\mathbb{I}(c_i = c_j) \mid D]$ corresponds exactly to our posterior similarity matrix entries $P_{ij} = \Pr(c_i = c_j \mid D)$, we arrive at our computable surrogate objective:
\[
\widehat{VI}_{\mathrm{UB}}(\hat{c}) =
\sum_{i=1}^n \log |\hat{C}_{\hat{c}_i}|
- 2 \sum_{i=1}^n \log \left(
\sum_{j \in \hat{C}_{\hat{c}_i}} P_{ij}
\right).
\]

Minimizing this equation serves as a rigorous, computationally efficient approximation to minimizing the true posterior expected VI loss, utilizing only the pairwise similarity matrix rather than enumerating the full posterior space.

\section{Experiments}

\subsection{Evaluation Metrics}
\label{sec:metrics}

For synthetic experiments, we evaluate each estimated poset \(\hat h\) against the ground-truth
poset \(h^\star\) using precedence-pair precision, recall, and F1. Let the transitive closure of h be 
\[
TC(h)=\{(i,j): i\neq j,\ i \prec_h j\}
\]
denote the set of all ordered comparable pairs in the transitive closure of poset \(h\). We define 
\[
\mathrm{TP}=|TC(h^\star)\cap TC(\hat h)|,\qquad
\mathrm{FP}=|TC(\hat h)\setminus TC(h^\star)|,\qquad
\mathrm{FN}=|TC(h^\star)\setminus TC(\hat h)|.
\]
We then compute
\[
\mathrm{Precision}=\frac{\mathrm{TP}}{\mathrm{TP}+\mathrm{FP}},\qquad
\mathrm{Recall}=\frac{\mathrm{TP}}{\mathrm{TP}+\mathrm{FN}},\qquad
\mathrm{F1}=\frac{2\,\mathrm{Precision}\,\mathrm{Recall}}
{\mathrm{Precision}+\mathrm{Recall}}.
\]
Thus, these metrics measure recovery of the underlying precedence relation rather than only the
minimal Hasse diagram representation.

\subsection{Experiment~A: HPO Synthetic}
Experiment~A is a full-factorial synthetic HPO recoverability study. We vary the assessor--global coupling strength
$\tau$, the queue-jump noise level $p$, the number of lists per assessor $L$, the choice-set size
$s$ (reported in code as \texttt{task\_size}), and the observation model (likelihood). The fixed settings are
$m=10$ items, $A=5$ assessors, $5\times 10^5$ MCMC iterations, burn-in $2.5\times 10^5$, and seed $=1$.
The resulting grid contains
\[
3 \times 2 \times 3 \times 3 \times 3 = 162
\]
fitted configurations.

\begin{table}[t]
\centering
\small
\caption{\textbf{Experiment~A: synthetic HPO configuration grid.}
The full grid is obtained by taking the Cartesian product of the values shown below.}
\label{tab:block_a_config_grid}
\setlength{\tabcolsep}{8pt}
\renewcommand{\arraystretch}{1.05}
\begin{tabular}{lll}
\toprule
\textbf{Parameter} & \textbf{Values} & \textbf{Interpretation} \\
\midrule
$\tau$ & $\{0.0,\ 0.5,\ 0.9\}$ & assessor--global coupling strength \\
$p$ & $\{0.01,\ 0.1\}$ & queue-jump noise probability \\
$L$ & $\{5,\ 10,\ 20\}$ & lists/tasks per assessor \\
$s$ (\texttt{task\_size}) & $\{2,\ 5,\ 10\}$ & items per choice set \\
likelihood &
Weighted/Log Successors/Queue-Jump &
observation model \\
\midrule
\multicolumn{3}{l}{\textbf{Fixed settings:} $m=10$, $A=5$, iterations $=5\times 10^5$, burn-in $=2.5\times 10^5$, seed $=1$} \\
\multicolumn{3}{l}{\textbf{Total configurations:} $162$} \\
\bottomrule
\end{tabular}
\end{table}

\begin{table}[h]
\centering
\small
\caption{\textbf{Mean runtime (minutes) for the synthetic experiments by likelihood and task size $m\in\{2,5,10\}$.}
Each run uses $5\times 10^5$ MCMC iterations with burn-in $2.5\times 10^5$.}
\label{tab:synth_runtime}
\setlength{\tabcolsep}{8pt}
\renewcommand{\arraystretch}{1.05}
\begin{tabular}{lccc}
\toprule
\textbf{Likelihood} & \textbf{Task size = 2} & \textbf{Task size = 5} & \textbf{Task size = 10} \\
\midrule
\texttt{frontier\_softmax\_queue\_jump} & 12.5 & 16.0 & 23.7 \\
\texttt{queue\_jump}                  & 21.7 & 33.9 & 449.8 \\
\texttt{weighted\_queue\_jump}        & 12.1 & 21.5 & 464.1 \\
\bottomrule
\end{tabular}
\end{table}

\begin{figure}[h]
\centering
\includegraphics[width=\linewidth]{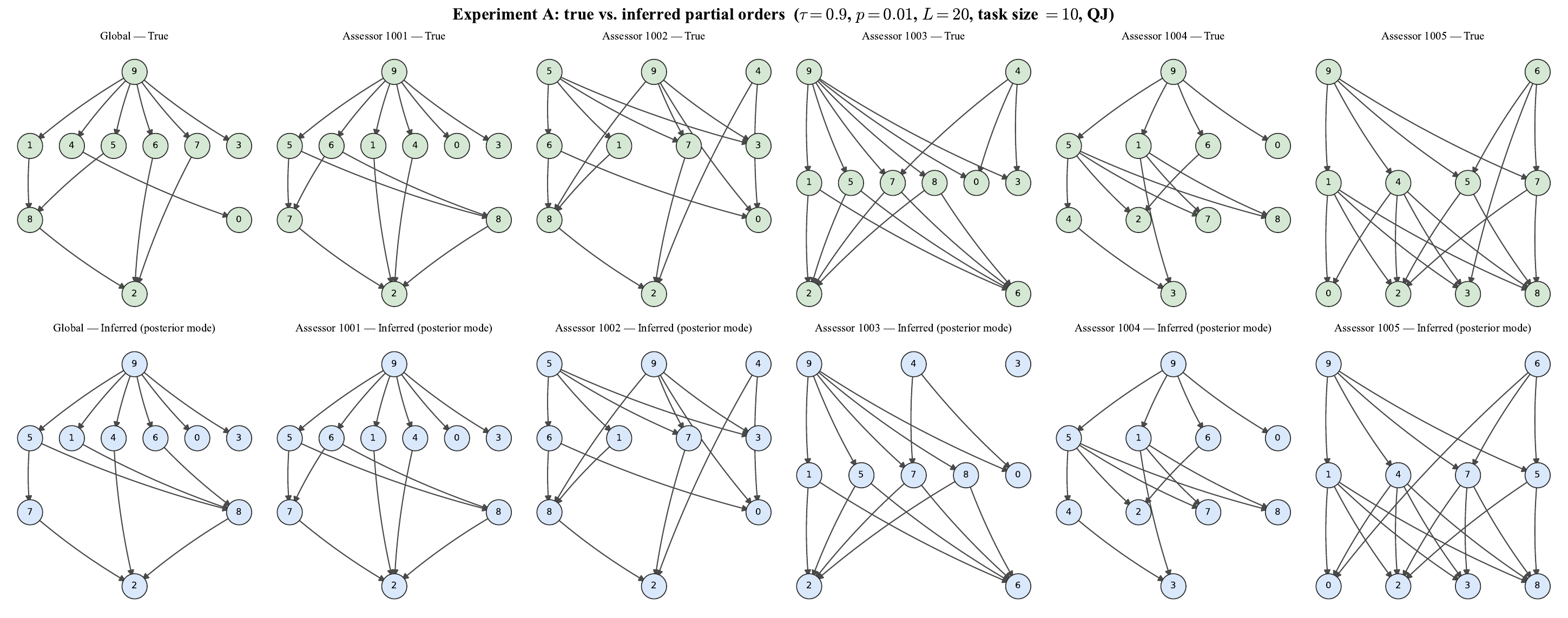}
\caption{\textbf{Experiment~A: true vs.\ inferred partial orders in a representative high-signal setting.}
Comparison between the ground-truth and posterior-mode inferred partial orders for a representative synthetic configuration
($\tau=0.9$, $p=0.01$, $L=20$, task size $m=10$, queue-jump). The inferred global and assessor-level structures recover
the main ordering relations closely.}
\label{fig:blocka_true_vs_inferred_poset}
\end{figure}

\begin{figure}[h]
\centering
\includegraphics[width=\linewidth]{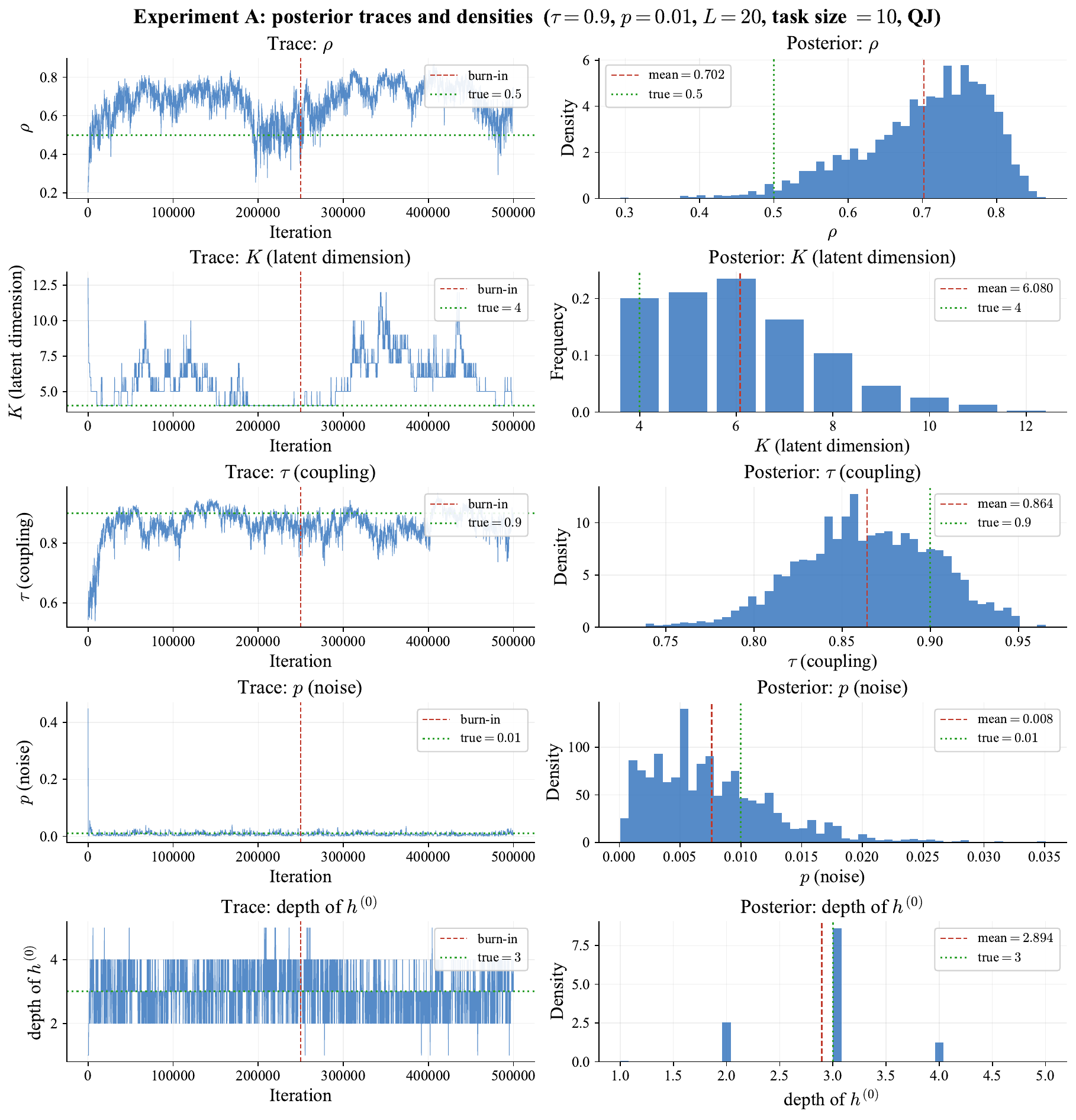}
\caption{\textbf{Experiment~A: posterior trace and density summaries in a representative high-signal setting.}
Posterior trace and marginal density plots for the main scalar parameters under the same synthetic configuration as
Figure~\ref{fig:blocka_true_vs_inferred_poset}. The traces show stable post burn-in behavior, and the posterior
distributions concentrate around the true parameter values.}
\label{fig:blocka_posterior_traces}
\end{figure}

\begin{figure}[h]
\centering
\includegraphics[width=\linewidth]{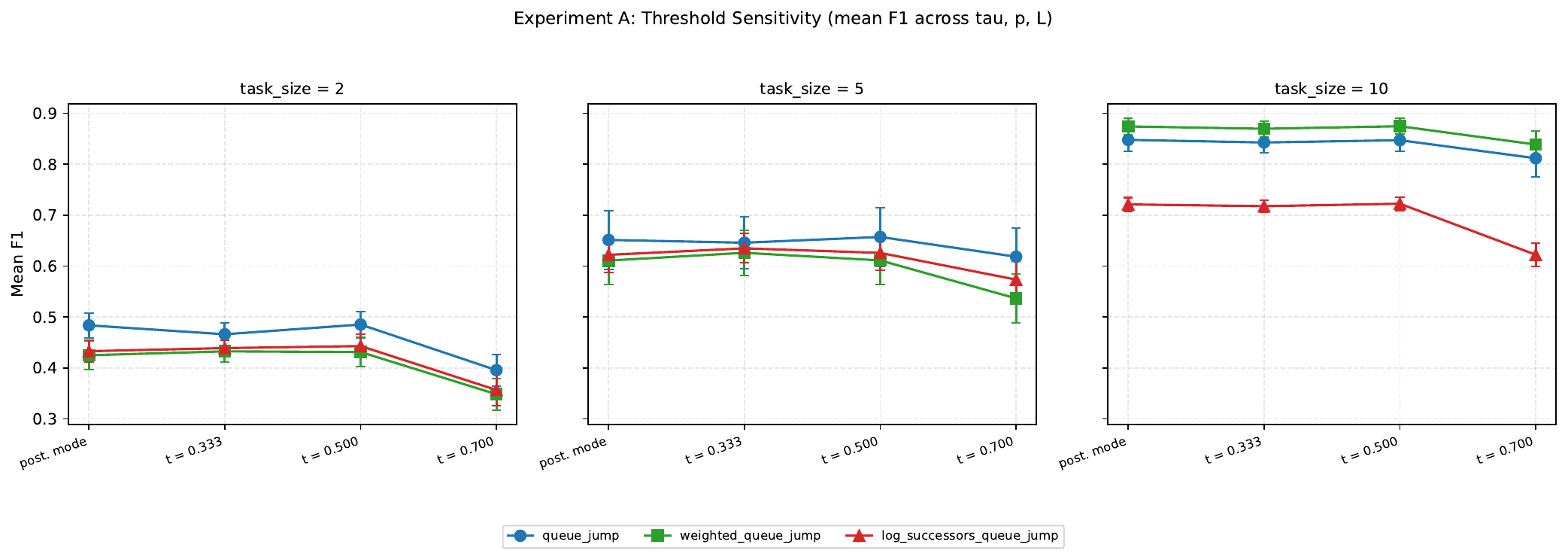}
\caption{\textbf{Experiment~A: threshold sensitivity of structural recovery.}
Mean F1, averaged over $\tau$, $p$, and $L$, for four posterior summaries: posterior mode and thresholded edge-probability estimators with thresholds $t=0.333$, $0.500$, and $0.700$. Panels correspond to task sizes $s\in\{2,5,10\}$, and curves compare the three observation models (\texttt{queue\_jump}, \texttt{weighted\_queue\_jump}, and \texttt{frontier\_softmax\_queue\_jump}).}
\label{fig:blocka_threshold_sensitivity}
\end{figure}

\subsection{Experiment B: HCPO Synthetic}

\begin{figure}[h]
\centering
\includegraphics[width=\linewidth]{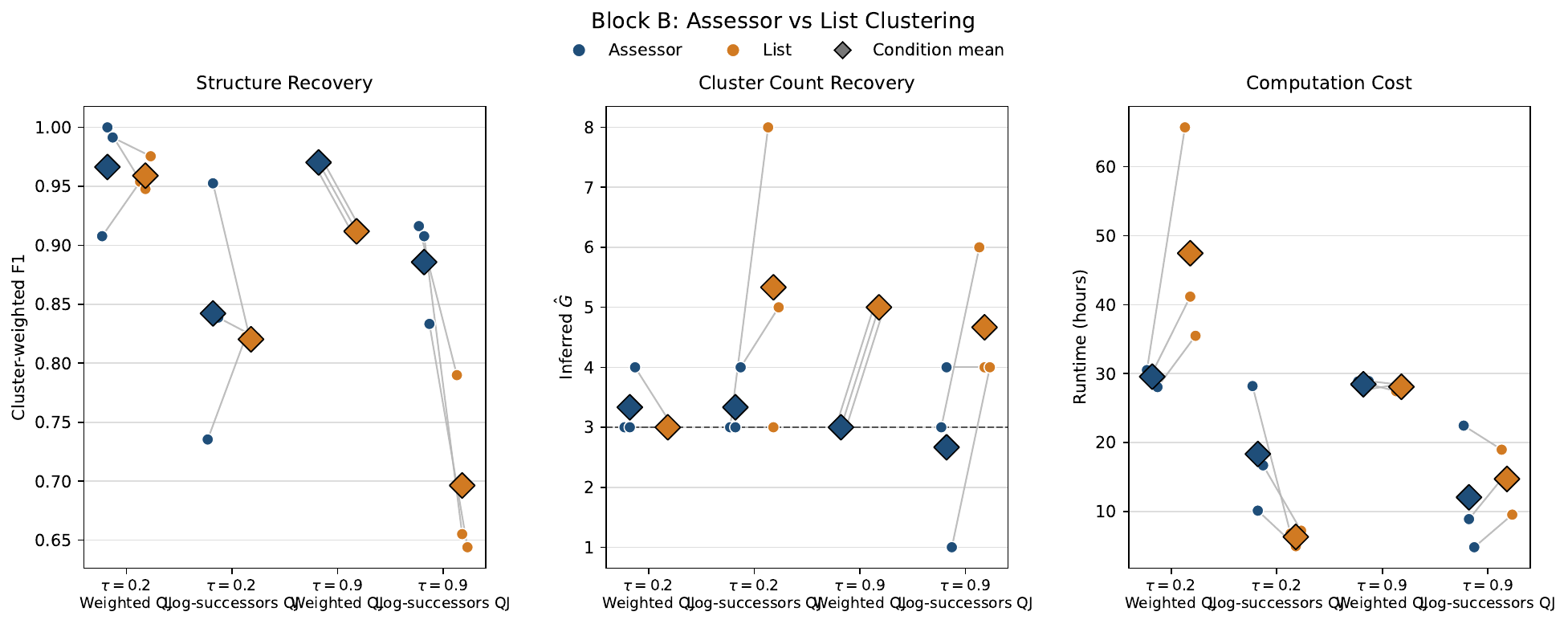}
\caption{\textbf{Experiment~B: assessor-level versus list-level clustering.}
Comparison of the two clustering modes across four representative conditions
($\tau\in\{0.2,0.9\}$ and likelihood $\in\{\texttt{weighted\_queue\_jump},\texttt{frontier\_softmax\_queue\_jump}\}$).
Left: cluster-weighted F1 for structure recovery. Middle: inferred number of clusters $G$, with the dashed line marking the true value $G_{\mathrm{true}}=3$. Right: runtime in hours. Assessor clustering yields equal or better structural recovery and typically recovers the true number of clusters more accurately, while list clustering is more prone to over-fragmentation and has less stable computational cost.}
\label{fig:blockb_assessor_vs_list}
\end{figure}

\subsection{Experiment C: CLOUD IAC Experiment for HPO}
\subsubsection{Cloud Iac 6 dataset}
\label{app:cloud_iac_6_dataset}
We evaluate our method on \textbf{Cloud-IaC-6}, a benchmark of cloud-provisioning tasks ranging from simple instance creation to complex high-availability clusters (see Tables~\ref{tab:cloud_iac_6} and~\ref{tab:aliyun_sceario}). The scenarios are derived from Aliyun cloud infrastructure products (see Table~\ref{tab:aliyun_glossary}). The dataset contains 60 successful real execution traces, comprising 54 LLM-generated traces from a diverse pool of agents, including Qwen-Plus, DeepSeek, and Kimi, together with 6 expert traces. The ground-truth graphs are specified by experts; see Figure~\ref{fig:aliyun-gt-covers}. We provide the full implementation and benchmark datasets in our public repository.\footnote{\url{https://anonymous.4open.science/r/Cloud-IaC-6-B970/README.md}}

\begin{table}[H]
\centering
\small
\renewcommand{\arraystretch}{1.3}
\begin{tabular}{|c|l|p{8.5cm}|}
\hline
\textbf{ID} & \textbf{Scenario Identifier} & \textbf{Description} \\ \hline
1 & SIMPLE\_ECS & Provisions a VPC, VSwitch, and Security Group, followed by a single ECS instance. \\ \hline
2 & SLB\_ECS\_RDS & A classic 3-tier web architecture integrating Server Load Balancer (SLB), ECS, and Relational Database Service (RDS). \\ \hline
3 & SLB\_ECS\_REDIS & A web architecture featuring a caching layer, utilizing SLB, ECS, and Redis. \\ \hline
4 & EIP\_SLB\_ECS & A public-facing application using an Elastic IP (EIP) bound to an SLB and an ECS backend. \\ \hline
5 & DUAL\_ZONE\_ECS\_SLB & Implements High Availability (HA) across multiple Availability Zones at the compute layer. \\ \hline
6 & DUAL\_ZONE\_ECS\_SLB\_RDS & A full-stack HA architecture featuring cross-zone ECS instances and a Primary/Secondary RDS deployment. \\ \hline
\end{tabular}
\caption{Cloud Infrastructure Benchmarking Scenarios}
\label{tab:aliyun_sceario}
\end{table}

\begin{table}[h]
    \centering
    \caption{Glossary of Aliyun Cloud Infrastructure Terms}
    \label{tab:aliyun_glossary}
    \renewcommand{\arraystretch}{1.2} 
    \begin{tabular}{l p{10cm}}
        \toprule
        \textbf{Term} & \textbf{Description} \\
        \midrule
        \textbf{ECS} (Elastic Compute Service) & A web service that provides resizable compute capacity in the cloud (virtual servers), allowing users to launch instances with a variety of operating systems and hardware configurations. \\
        
        \textbf{SLB} (Server Load Balancer) & A traffic distribution service that manages high traffic by distributing incoming network requests across multiple ECS instances to ensure high availability and reliability. \\
        
        \textbf{RDS} (Relational Database Service) & A managed database service that provides scalable and reliable relational databases (e.g., MySQL, PostgreSQL) without the need for manual hardware provisioning or maintenance. \\
        
        \textbf{VPC} (Virtual Private Cloud) & A private, isolated network environment within the cloud where users can configure IP address ranges, subnets, and routing tables to securely manage their resources. \\
        
        \textbf{VSwitch} (Virtual Switch) & A virtual networking component within a VPC that connects different cloud resources (like ECS instances) in a specific zone or subnet. \\
        
        \textbf{EIP} (Elastic IP) & A static, public IP address designed for dynamic cloud computing, allowing users to mask the failure of an instance or software by rapidly remapping the address to another instance. \\
        
        \textbf{Redis} & An in-memory data structure store used as a database, cache, and message broker, often utilized in web architectures to improve performance. \\
        
        \textbf{HA} (High Availability) & A system design approach that ensures a certain level of operational performance (uptime) for a higher-than-normal period, often achieved by deploying resources across multiple zones (e.g., Dual Zone). \\
        \bottomrule
    \end{tabular}
\end{table}

\noindent \textbf{Trace Data Structure.}
Each entry in the dataset is a serialized execution trace $\tau = (I, \mathcal{A}, \mathcal{B})$, where:
\begin{itemize}
    \item \textbf{Intent ($I$):} The natural language instruction (e.g., ``Create a 2-core ECS in Hangzhou Zone H'').
    \item \textbf{Action Sequence ($\mathcal{A}$):} The linear sequence of API calls executed by the agent (e.g., \texttt{CreateVpc} $\to$ \texttt{RunInstances}).
    \item \textbf{Blackboard State ($\mathcal{B}$):} The shared context containing resource IDs (e.g., \texttt{VpcId}, \texttt{SecurityGroupId}) produced by earlier actions and consumed by later ones.
\end{itemize}

Figure~\ref{lst:trace_example} illustrates a sample trace from Scenario S1 (\texttt{simple\_ecs}). Although the agent executes the actions sequentially (System 2 behavior), the underlying dependencies reveals latent concurrency: \texttt{CreateVSwitch} and \texttt{CreateSecurityGroup} both depend on \texttt{CreateVpc}, but are independent of each other. 

\begin{figure}[h!]
\begin{minipage}{\linewidth}
\small
\begin{verbatim}
{
  "trace_id": "T01_qwen-plus_20260104",
  "intent": "Create a 2-core 4G ECS instance in Hangzhou Zone H",
  "action_sequence": [
    { "step": 1, "action": "CreateVpc", 
      "output": {"VpcId": "vpc-9517..."} },
      
    { "step": 2, "action": "CreateVSwitch", 
      "params": {"VpcId": "vpc-9517...", "ZoneId": "cn-hangzhou-h"},
      "output": {"VSwitchId": "vsw-191b..."} },
      
    { "step": 3, "action": "CreateSecurityGroup", 
      "params": {"VpcId": "vpc-9517..."},
      "output": {"SecurityGroupId": "sg-0fae..."} },
      
    { "step": 4, "action": "RunInstances", 
      "params": {"VSwitchId": "vsw-191b...", "SecurityGroupId": "sg-0fae..."},
      "output": {"InstanceId": "i-007d..."} }
  ]
}
\end{verbatim}
\captionof{figure}{\textbf{Sample Execution Trace (S1: Simple ECS).} The log captures the linear execution of actions. Note the explicit data dependencies: Step 4 requires outputs from Steps 2 and 3, while Steps 2 and 3 only require Step 1.}
\label{lst:trace_example}
\end{minipage}
\end{figure}

\begin{table}[t]
\centering
\caption{\textbf{Cloud-IaC-6 benchmark statistics for the combined real-trace setting.}
The benchmark contains $60$ successful real traces collected across six scenarios, comprising $54$ LLM-generated traces and $6$ expert traces.
Here, $M_a$ and $|E|$ denote the number of nodes and edges in the ground-truth partial order, $n_{\mathrm{real}}$ denotes the number of successful real traces for the scenario, \textbf{IP-Cov} denotes bidirectional incomparable-pair coverage induced by the observed real traces before synthetic augmentation, and $n_{\mathrm{full}}$ denotes the total number of traces after augmentation required to reach $\mathrm{IP\text{-}Cov}=1.0$.}
\label{tab:cloud_iac_6}
\small
\setlength{\tabcolsep}{5pt}
\begin{tabular}{l l c c c r c}
\toprule
\textbf{ID} & \textbf{Scenario Name} & $\mathbf{M_a}$ & $\mathbf{|E|}$ & $\mathbf{n_{\mathrm{real}}}$ & \textbf{IP-Cov} & $\mathbf{n_{\mathrm{full}}}$ \\
\midrule
S1 & \texttt{simple\_ecs}               & 5  & 5  & 11 & 0.0\%  & 13 \\
S2 & \texttt{slb\_ecs\_rds}             & 12 & 14 & 10 & 25.0\% & 33 \\
S3 & \texttt{slb\_ecs\_redis}           & 9  & 10 & 11 & 53.3\% & 30 \\
S4 & \texttt{eip\_slb\_ecs}             & 9  & 10 & 11 & 56.2\% & 17 \\
S5 & \texttt{dual\_zone\_ecs\_slb}      & 7  & 8  & 9  & 20.0\% & 13 \\
S6 & \texttt{dual\_zone\_ecs\_slb\_rds} & 10 & 12 & 8  & 33.3\% & 24 \\
\bottomrule
\end{tabular}
\end{table}

\begin{figure*}[t]
\centering
\begin{subfigure}[t]{0.32\textwidth}
  \centering
  \includegraphics[width=\linewidth]{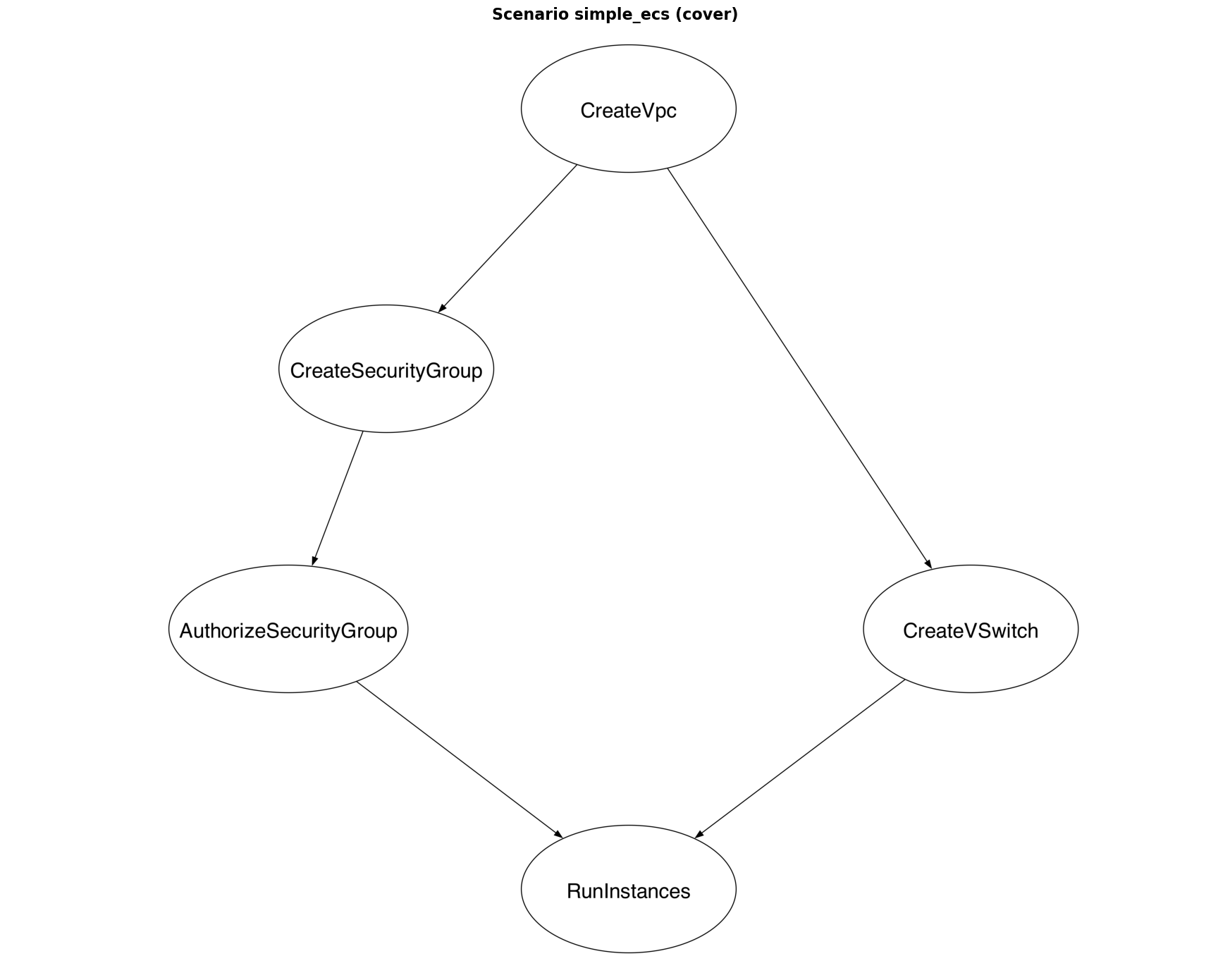}
  \caption{\texttt{simple\_ecs}}
\end{subfigure}\hfill
\begin{subfigure}[t]{0.32\textwidth}
  \centering
  \includegraphics[width=\linewidth]{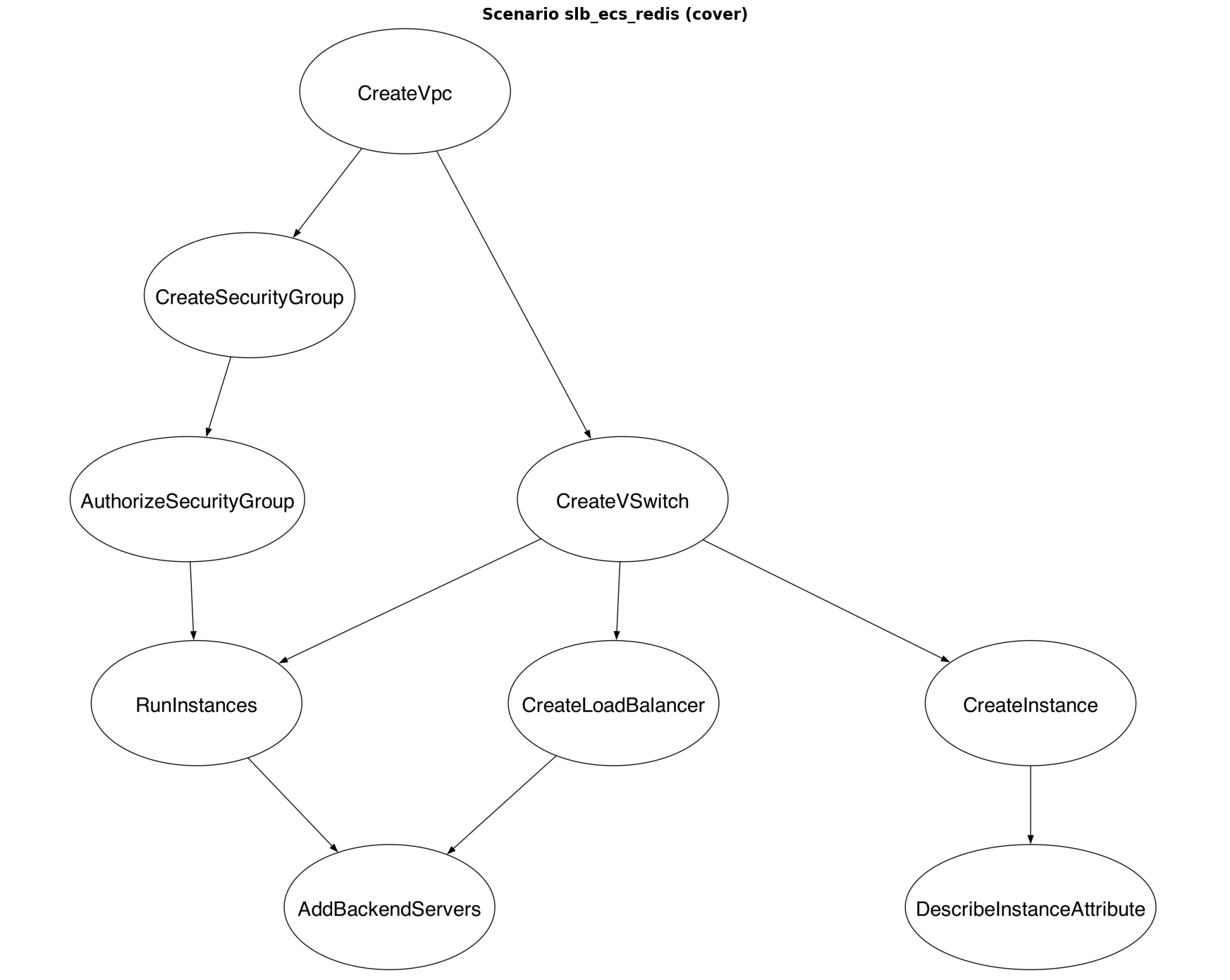}
  \caption{\texttt{slb\_ecs\_redis}}
\end{subfigure}\hfill
\begin{subfigure}[t]{0.32\textwidth}
  \centering
  \includegraphics[width=\linewidth]{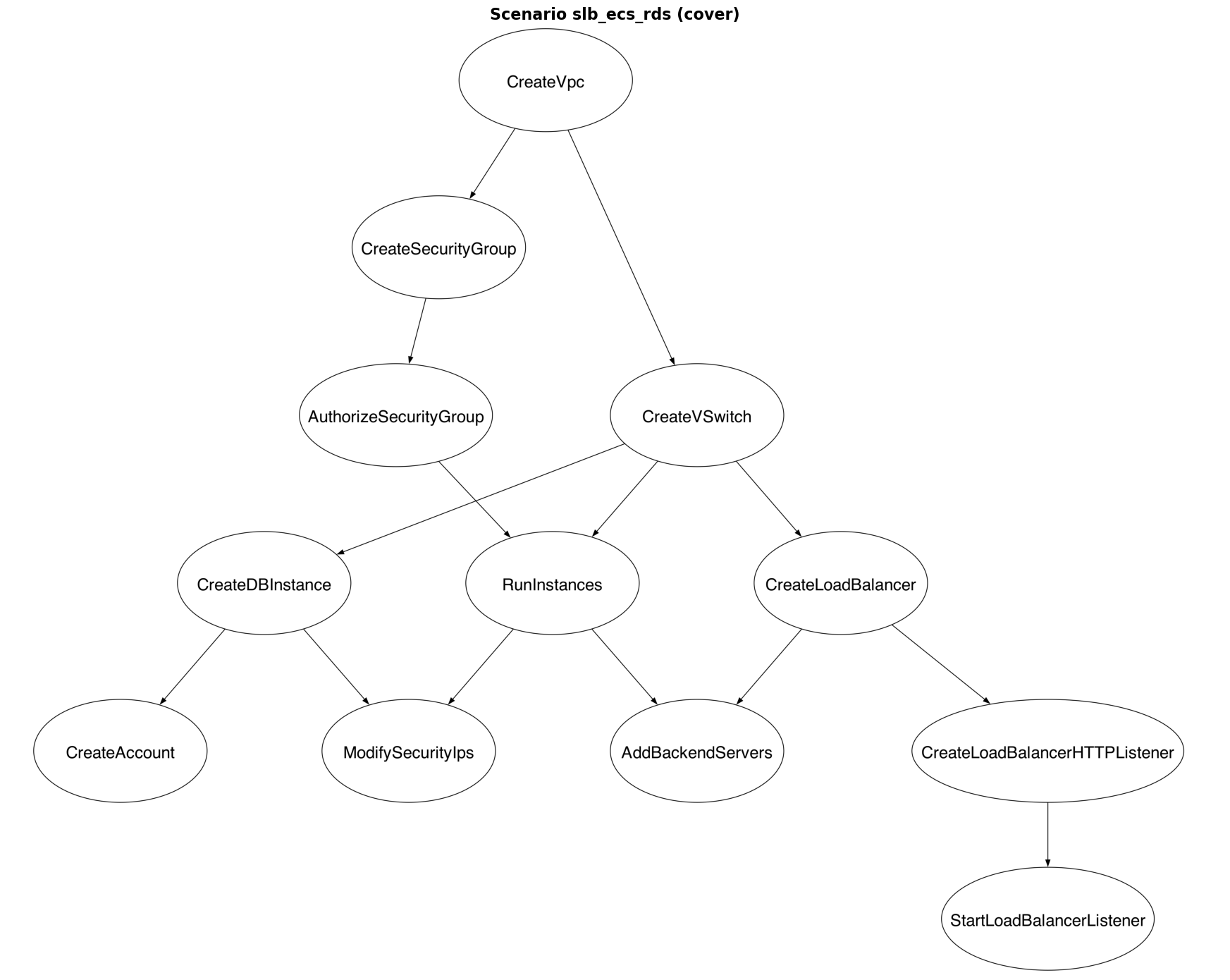}
  \caption{\texttt{slb\_ecs\_rds}}
\end{subfigure}

\vspace{2mm}

\begin{subfigure}[t]{0.32\textwidth}
  \centering
  \includegraphics[width=\linewidth]{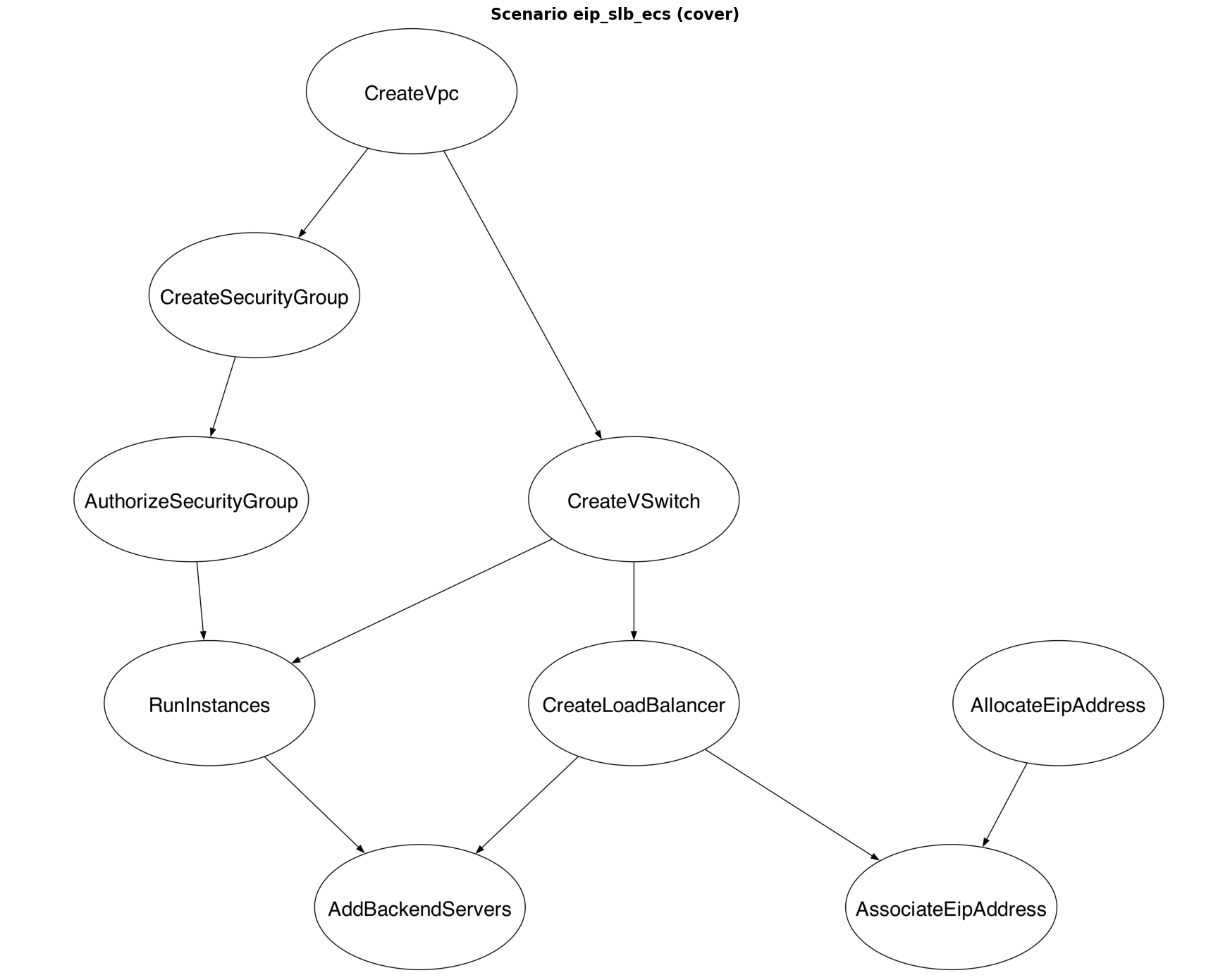}
  \caption{\texttt{eip\_slb\_ecs}}
\end{subfigure}\hfill
\begin{subfigure}[t]{0.32\textwidth}
  \centering
  \includegraphics[width=\linewidth]{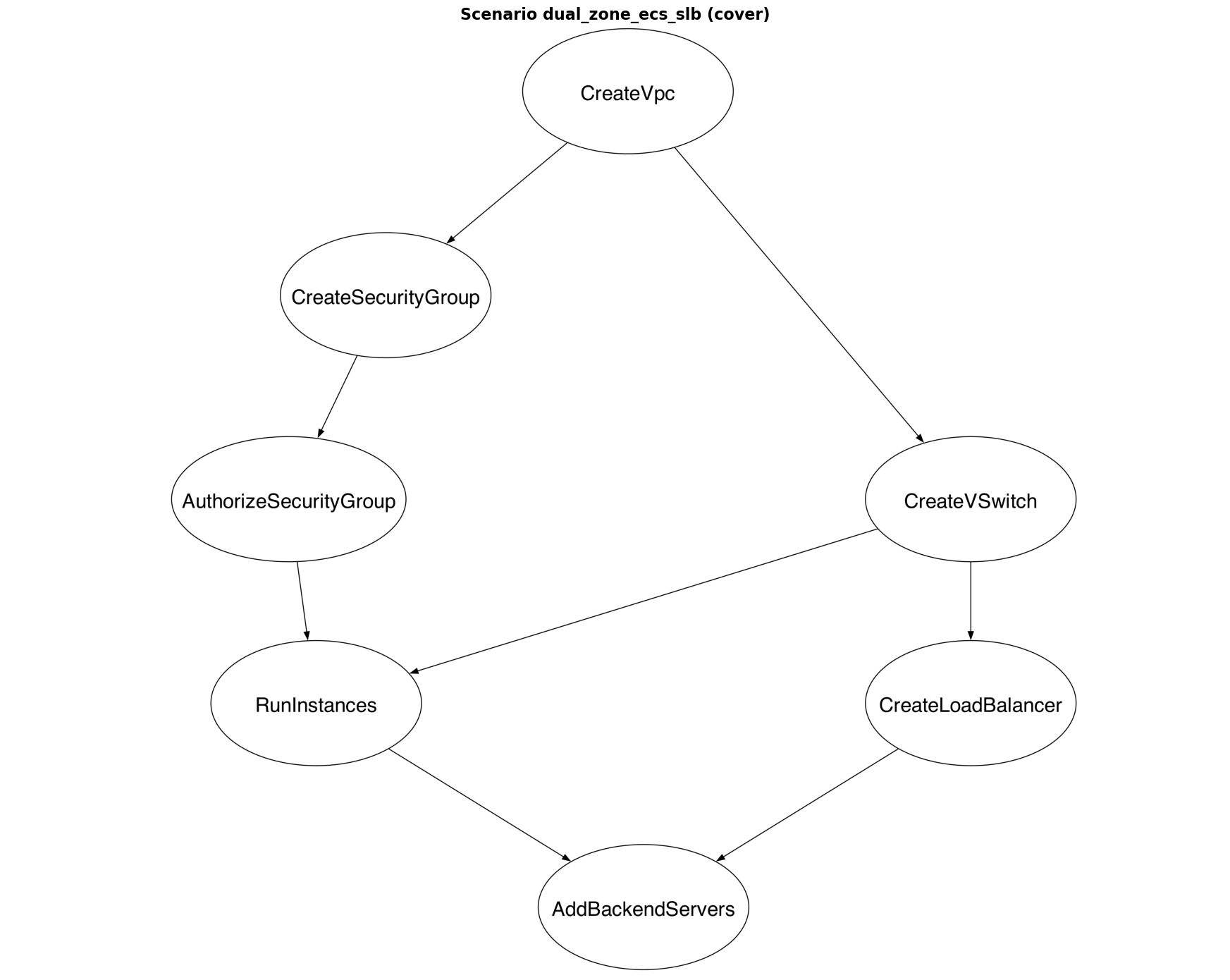}
  \caption{\texttt{dual\_zone\_ecs\_slb}}
\end{subfigure}\hfill
\begin{subfigure}[t]{0.32\textwidth}
  \centering
  \includegraphics[width=\linewidth]{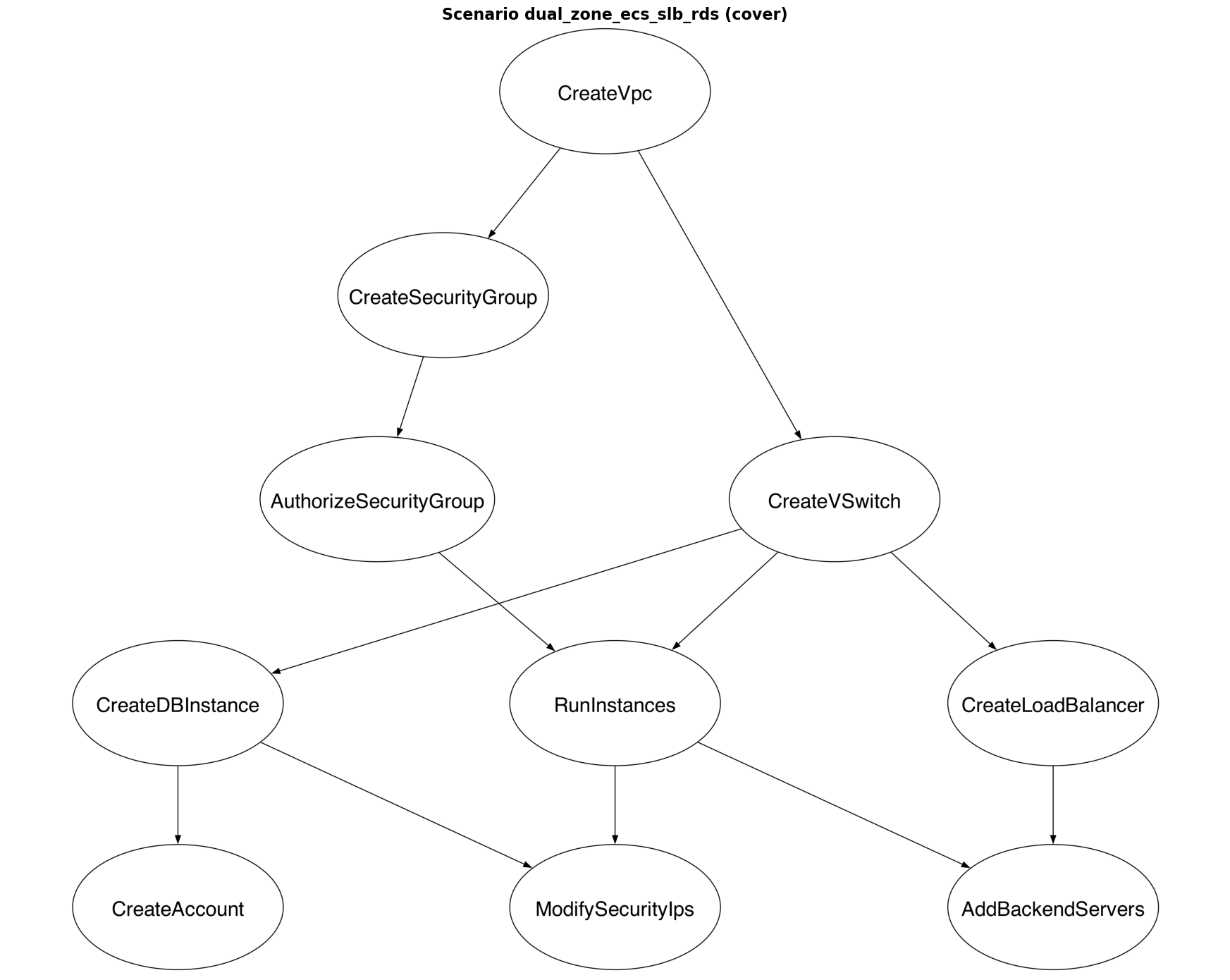}
  \caption{\texttt{dual\_zone\_ecs\_slb\_rds}}
\end{subfigure}
\caption{Ground-truth action-precedence graphs (covers) for the six Aliyun scenarios.
Nodes are cloud API actions; edges denote mandatory precedence constraints.}
\label{fig:aliyun-gt-covers}
\end{figure*}

\subsubsection{Evaluation metrics and incomparability-pair coverage}
\label{app:metrics_ipcov}

Let $h^\star$ denote the ground-truth partial order on action set $A$, and let $\widehat h$ denote an inferred partial order. We write
\[
C(h) = \{(a,b)\in A\times A : a \succ_h b,\ \nexists c \in A \text{ such that } a \succ_h c \succ_h b\}
\]
for the set of cover relations (equivalently, the edge set of the transitive reduction of $h$), and
\[
I(h) = \bigl\{\{a,b\}\subseteq A : a \not\succ_h b \text{ and } b \not\succ_h a\bigr\}
\]
for the set of incomparable task pairs. In the main text, we refer to elements of $I(h^\star)$ as \emph{critical pairs}.

For a generic ground-truth set $S^\star$ and estimate $\widehat S$, we define
\[
\mathrm{Prec}(S^\star,\widehat S)=\frac{|S^\star \cap \widehat S|}{|\widehat S|},\qquad
\mathrm{Rec}(S^\star,\widehat S)=\frac{|S^\star \cap \widehat S|}{|S^\star|},
\]
and
\[
\mathrm{F1}(S^\star,\widehat S)
=
\frac{2\,\mathrm{Prec}(S^\star,\widehat S)\,\mathrm{Rec}(S^\star,\widehat S)}
{\mathrm{Prec}(S^\star,\widehat S)+\mathrm{Rec}(S^\star,\widehat S)}.
\]

The \emph{cover-edge F1} is defined as
\[
\mathrm{F1}_{\mathrm{cover}}
=
\mathrm{F1}\!\bigl(C(h^\star),\, C(\widehat h)\bigr),
\]
which evaluates recovery of the direct ordering constraints in the transitive reduction.

The \emph{critical-pair F1} is defined as
\[
\mathrm{F1}_{\mathrm{crit}}
=
\mathrm{F1}\!\bigl(I(h^\star),\, I(\widehat h)\bigr),
\]
which evaluates how well the inferred structure preserves task pairs that are truly incomparable in the ground truth.

Now let $\mathcal D=\{\pi^{(1)},\dots,\pi^{(m)}\}$ denote a collection of observed traces, and let $A_\pi \subseteq A$ be the set of actions appearing in trace $\pi$. We say that a trace $\pi$ is \emph{feasible} under $\widehat h$, written $\pi \models \widehat h$, if
\[
a \succ_{\widehat h} b,\ a,b\in A_\pi
\quad \Longrightarrow \quad
\operatorname{pos}_\pi(a) > \operatorname{pos}_\pi(b),
\]
for all ordered pairs constrained by $\widehat h$. The \emph{trace feasibility} score is then
\[
\mathrm{TF}(\widehat h;\mathcal D)
=
\frac{1}{|\mathcal D|}
\sum_{\pi \in \mathcal D}
\mathbf{1}\{\pi \models \widehat h\}.
\]

Finally, to quantify how much evidence the trace set provides about true incomparabilities, we define the \emph{incomparability-pair coverage} (IP-Cov). For each $\{a,b\}\in I(h^\star)$, let
\[
\kappa_{\mathcal D}(a,b)
=
\mathbf{1}\Bigl\{
\exists\, \pi,\pi' \in \mathcal D
\text{ such that }
a,b \in A_\pi \cap A_{\pi'}
\text{ and }
\operatorname{pos}_\pi(a) > \operatorname{pos}_\pi(b),\
\operatorname{pos}_{\pi'}(a) < \operatorname{pos}_{\pi'}(b)
\Bigr\}.
\]
That is, $\kappa_{\mathcal D}(a,b)=1$ if the observed traces contain evidence for both relative orders of the incomparable pair $\{a,b\}$. We then define
\begin{equation}\label{eq:IP-Cov-define}
\mathrm{IP\text{-}Cov}(\mathcal D;h^\star)
=
\frac{1}{|I(h^\star)|}
\sum_{\{a,b\}\in I(h^\star)}
\kappa_{\mathcal D}(a,b).
\end{equation}
Thus, $\mathrm{IP\text{-}Cov}$ measures the fraction of truly incomparable pairs for which the dataset reveals both admissible orderings. 

\begin{table}[t]
\centering
\caption{\textbf{Cross-likelihood comparison on Cloud-IaC-6 at $10^6$ MCMC iterations.}
For each likelihood, we report the mean HPO cover-edge F1, critical-pair F1, and trace feasibility across the six scenarios, together with the runtime of the hierarchical HPO fit, the aggregate runtime of the six single-PO baselines, and the total elapsed time.}
\label{tab:aliyun_likelihood_comparison_appendix}
\small
\setlength{\tabcolsep}{5pt}
\begin{tabular}{l c c c c c c}
\toprule
\textbf{Likelihood} & \textbf{Cover F1} & \textbf{Crit F1} & \textbf{Feas.} & \textbf{HPO (h)} & \textbf{Single (h)} & \textbf{Total (h)} \\
\midrule
\texttt{weighted\_queue\_jump}         & 0.988 & 0.997 & 0.982 & 6.618 & 7.931 & 14.559 \\
\texttt{queue\_jump}                   & 0.988 & 0.997 & 0.982 & 6.850 & 7.695 & 14.554 \\
\texttt{frontier\_softmax\_likelihood} & 0.860 & 0.842 & 0.808 & 1.922 & 3.979 & 5.915 \\
\bottomrule
\end{tabular}
\end{table}

\subsection{Experiment D: The Sound Experiment for HCPO Assessor}

\subsubsection{The 3D sound dataset}
\label{subsec:data-sound}

Our first case study re-analyses listening data from
\citet{crispino19}, who investigate how three–dimensional (3-D) sound
spatialisation shapes perceived human agency in acousmatic music. The stimuli
are synthetic, yet physically constrained, renderings of a single cello
bow–stroke captured with 3-D motion tracking. The vertical component of the
bow trajectory is mapped to pitch, the bow velocity controls amplitude and spectral variation, and the resulting signals are spatialised using higher-order Ambisonics in a virtual concert hall.
\begin{table}[htbp]
  \centering
  \caption{Stimulus conditions for the 3D bow–stroke sounds.}
  \label{tab:bow-stimuli}
  \begin{tabular*}{\linewidth}{@{\extracolsep{\fill}} lcccl}
    \toprule
    ID  & Spatial cue          & Pitch cue & Vol-2 cue & Extra notes \\
    \midrule
    S1  & full                 & \cmark & \cmark & baseline (5 s) \\
    S2  & full (front scene)   & \cmark & \cmark & scene rendered in front \\
    S3  & none (mono)          & \cmark & \cmark & mono over 1 speaker \\
    S4  & partial (global)     & \cmark & \cmark & fewer spatial details \\
    S5  & minimal              & \cmark & \cmark & only 3 direction changes \\
    S6  & full                 & \cmark & \xmark & Vol-2 removed \\
    S7  & full                 & \xmark & \cmark & pitch variation removed \\
    S8  & full                 & \xmark & \xmark & pitch + Vol-2 removed \\
    S9  & partial              & \xmark & \xmark & S4 minus pitch + Vol-2 \\
    S10 & minimal              & \xmark & \xmark & S5 minus pitch + Vol-2 \\
    S11 & full                 & \cmark & \cmark & 30\% slower (6.5 s) \\
    S12 & full                 & \cmark & \cmark & 50\% slower (7.5 s) \\
    \bottomrule
  \end{tabular*}
\end{table}

\begin{table}[h]
\centering
\small
\caption{\textbf{Sound data: VI-optimal partition summaries.}
For each observation model, we report the cluster-size profile of the VI-optimal partition, the corresponding number of clusters $\hat G$, and the VI objective value. Lower VI objective values indicate better agreement with the posterior similarity structure.}
\label{tab:sound_partition_summary}
\setlength{\tabcolsep}{5pt}
\renewcommand{\arraystretch}{1.05}
\begin{tabular}{l p{0.46\linewidth} c c}
\toprule
\textbf{Observation model} & \textbf{VI-optimal cluster sizes} & $\hat G$ & \textbf{VI objective} \\
\midrule
frontier-softmax        & 16, 14, 5, 4, 4, 3 & 6 & $-62.62$ \\
weighted queue-jump & 18, 14, 7, 4, 3    & 5 & $-73.10$ \\
\bottomrule
\end{tabular}
\end{table}

\begin{table}[t]
\centering
\small
\caption{\textbf{Sound refit: multi-chain restart MCMC efficiency.}
Effective sample sizes (ESS) for key metrics from the final 4-chain restart run (seeds 42/143/244/345), along with the summed ESS across chains. All split-$\widehat{R}$ values for these metrics are below $1.05$ (range $1.001$--$1.039$).}
\label{tab:sound_refit_restart_ess}
\begin{tabular}{lrrrrr}
\toprule
\textbf{Metric} & \textbf{seed42} & \textbf{seed143} & \textbf{seed244} & \textbf{seed345} & \textbf{ESS sum} \\
\midrule
loglik        & 724.519  & 532.314  & 642.681  & 510.071  & 2409.584 \\
$\rho$        & 38.528   & 30.439   & 56.239   & 61.117   & 186.323 \\
$\tau$        & 22.343   & 31.760   & 82.558   & 73.404   & 210.065 \\
$K$           & 319.721  & 236.305  & 294.844  & 320.260  & 1171.130 \\
prob\_noise   & 182.479  & 217.678  & 312.085  & 336.489  & 1048.730 \\
$K_{\log(1+\rho)}$ & 212.531  & 100.926  & 188.486  & 201.086  & 703.029 \\
depth         & 1034.049 & 1399.369 & 1559.238 & 1078.648 & 5071.305 \\
\bottomrule
\end{tabular}
\end{table}

This appendix records the execution protocol used for all experiment blocks in the paper.
The goal is to make the computational pipeline reproducible (fixed configurations and seeds),
and to ensure that reported posterior summaries are supported by basic MCMC diagnostics.

\begin{figure}[h]
\centering
\includegraphics[width=0.95\linewidth]{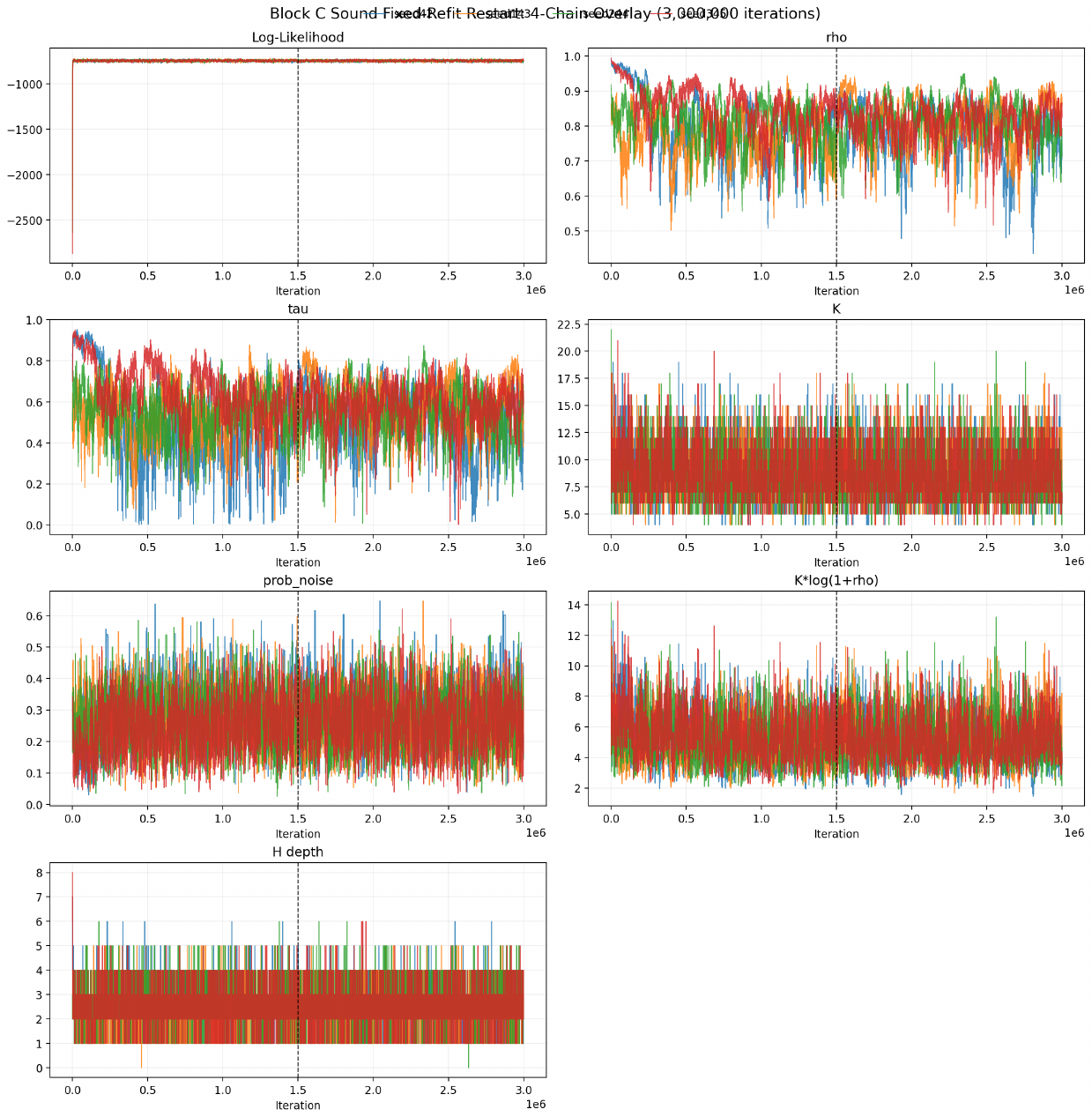}
\caption{\textbf{Sound refit: multi-chain overlay.}
Overlay of four refit chains (with restarts) conditional on the fixed cluster assignments, used to assess
cross-chain agreement and mixing.}
\label{fig:multiseed_profile_new}
\end{figure}

\begin{figure}[h]
\centering
\includegraphics[width=0.95\linewidth]{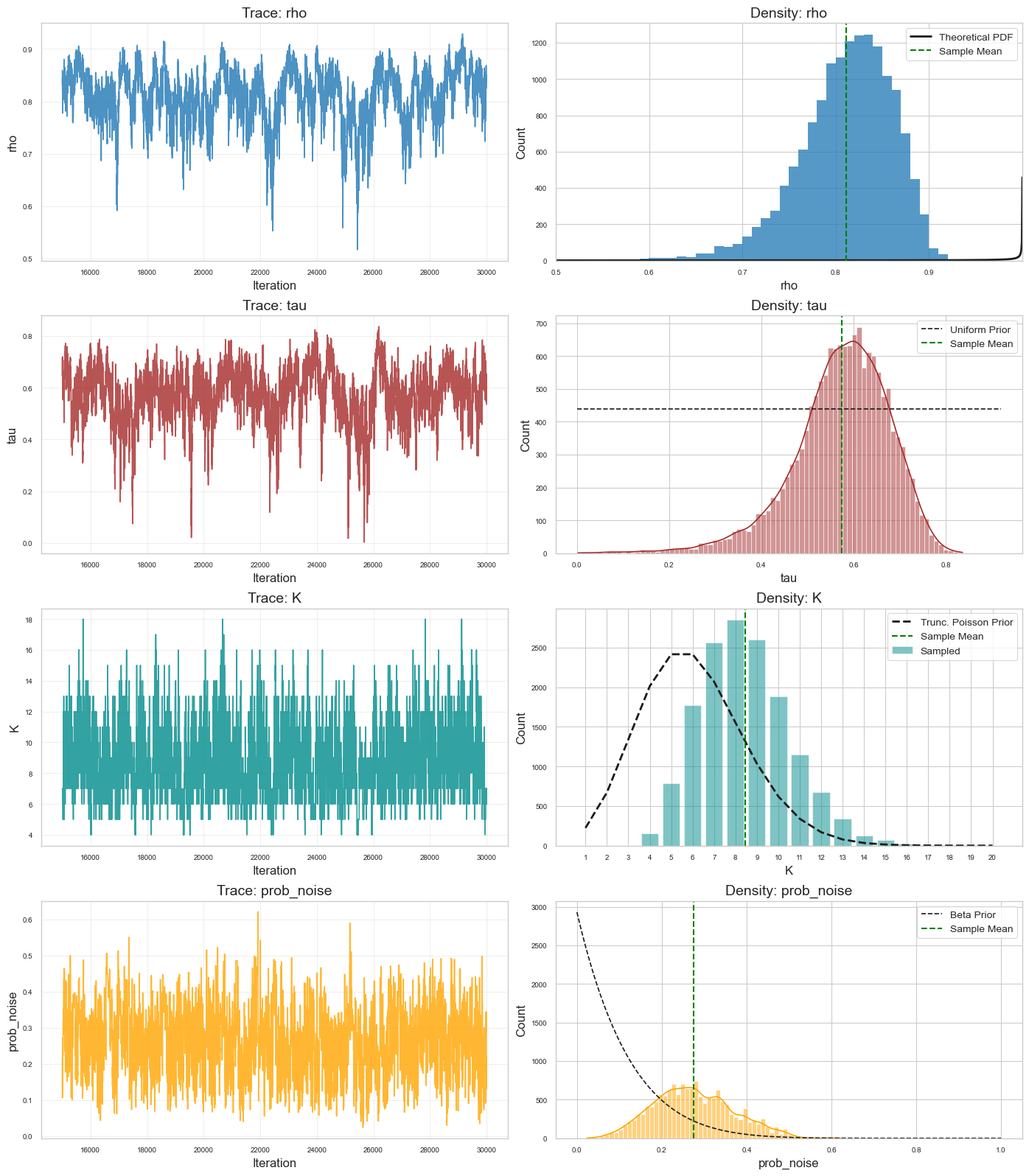}
\caption{\textbf{Sound refit: posterior diagnostics.}
Posterior trace summaries from the fixed-partition refit on the sound dataset (four chains; seeds 42/143/244/345),
used to assess convergence and mixing of key scalar quantities.}
\label{fig:sound_refit_posterior_diagnostics}
\end{figure}

We treat each participant as an assessor providing a sparse, noisy partial
order over the $12$ stimuli, induced by their $30$ pairwise choices. This
setting is particularly well-suited to hierarchical partial-order models,
which can pool information across listeners while allowing for clustered
variation in how spatial and dynamical cues translate into perceived human
agency.

\begin{table}[h]
\centering
\small
\caption{\textbf{Sound data: MCMC summaries for two observation models.}
Posterior mean (SD) and effective sample size (ESS) computed from $n_{\text{post}}=3000$ post burn-in draws.}
\label{tab:sound_obsmodel_diag}
\setlength{\tabcolsep}{6pt}
\renewcommand{\arraystretch}{1.0}
\begin{tabular}{lccccc}
\toprule
& \multicolumn{2}{c}{\textbf{frontier-softmax}} && \multicolumn{2}{c}{\textbf{Weighted queue-jump}} \\
\cmidrule{2-3}\cmidrule{5-6}
\textbf{Parameter} & \textbf{Mean (SD)} & \textbf{ESS} && \textbf{Mean (SD)} & \textbf{ESS} \\
\midrule
$K$            & 11.38 (2.37) & 176.7 && 11.01 (2.42) & 115.5 \\
Depth          & 2.37 (0.81)  & 21.1  && 1.91 (0.72)  & 76.5 \\
Log-likelihood & $-715.30$ (16.72) & 81.1 && $-714.80$ (17.25) & 57.4 \\
Noise prob.    & 0.129 (0.071) & 4.19 && 0.181 (0.095) & 44.2 \\
$\rho$         & 0.832 (0.041) & 12.6 && 0.838 (0.036) & 35.2 \\
$\tau$         & 0.458 (0.176) & 6.30 && 0.611 (0.120) & 28.5 \\
\bottomrule
\end{tabular}
\end{table}

\begin{figure}[t]
  \centering
  \includegraphics[width=0.95\linewidth]{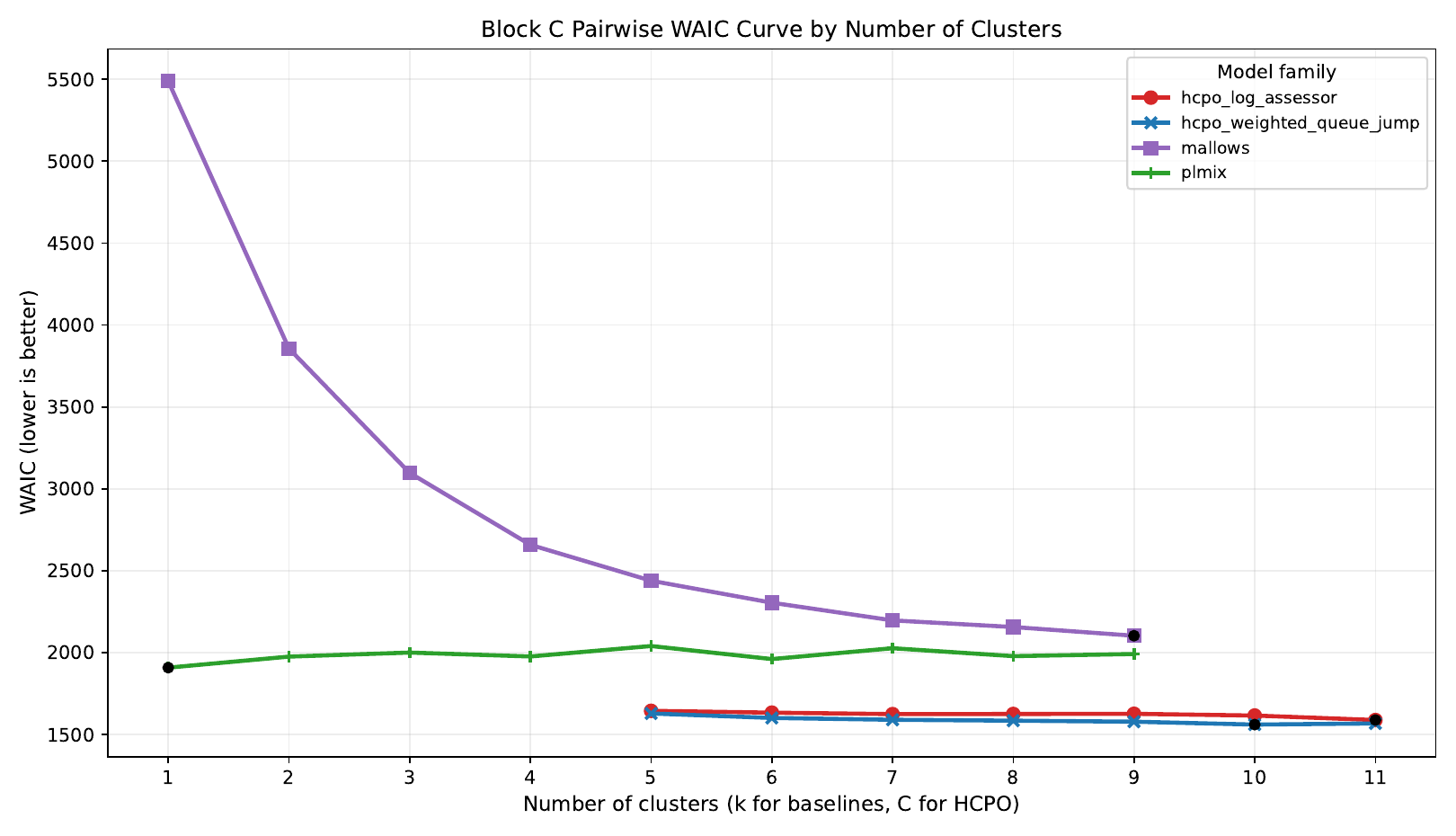}
  \caption{\textbf{Sound data (pairwise unit): WAIC versus mixture size / number of clusters.}
  WAIC is computed on the same $n_{\mathrm{obs}}=1380$ pairwise comparisons for all models.
  For the PL and Mallows baselines, the horizontal axis is the number of mixture components.
  For HCPO, it is the number of clusters $G$ appearing in the posterior sample (two observation models shown).
  Lower values are better.}
  \label{fig:sound_pairwise_waic}
\end{figure}

\subsubsection{Experiment D: further details for model comparisons}\label{sec:sound-model-comp-detail}

We compare HCPO with two total-order mixture baselines, each using three mixture components. For the Mallows mixture, we fit a grid of candidate models over \(G\) to explore performance against $G$, and also fix \(G=3\) for further focused comments. Because the resulting MCMC output is subject to label switching, we relabel each retained post-burn-in draw and define the assessor-level cluster assignment by the modal relabeled label,
\[
\hat c_a = \arg\max_g \sum_t \mathbf{1}\{c_a^{(t)} = g\}.
\]
For the Plackett--Luce (PL) mixture, we use a multistart MAP fit with \(G=3\). This yields soft membership probabilities \(\hat{\bm z}_a \in [0,1]^3\), satisfying \(\sum_{g=1}^3 \hat z_{ag}=1\), which we convert to hard assignments via
\[
\hat c_a = \arg\max_g \hat z_{ag}.
\]
For behavioral interpretation, we align the three Mallows clusters to the three PL clusters by optimal matching.
After alignment, the assessor match rate between the PL and Mallows partitions is $0.8043$, and the aligned
cluster-level stimulus ranks are strongly concordant, with Spearman correlations
$(0.979,\ 0.825,\ 0.846)$; see Figure~\ref{fig:blockc_rank_comparison}.

To compare predictive fit on a common scale, we compute WAIC on the \emph{pairwise} observation unit, treating each
A/B judgment as one data point ($n_{\mathrm{obs}}=1380$). Figure~\ref{fig:sound_pairwise_waic} plots WAIC as a
function of the number of mixture components for the total-order baselines (Plackett--Luce and Mallows), and as a
function of the number of inferred clusters for HCPO under two observation models (weighted queue-jump and
log-assessor/Frontier-Softmax). Lower WAIC indicates better expected out-of-sample performance.

Two qualitative patterns emerge. First, the Mallows baseline improves substantially as the number of mixture components increases. When \(G\) is small, a total-order mixture must explain many mutually inconsistent pairwise outcomes, including cyclic patterns, using only a few archetypal permutations, which results in poor predictive fit. As \(G\) increases, the fit improves because heterogeneous listeners can be separated into more homogeneous subgroups. Second, the PL curve is comparatively flat across \(G\), suggesting that, under this fitting regime, additional mixture components provide only limited predictive benefit. Over the range of cluster counts shown, HCPO consistently achieves lower WAIC than the total-order baselines, and the weighted queue-jump variant uniformly outperforms the log-assessor/Frontier-Softmax variant; see Figure~\ref{fig:sound_pairwise_waic}.

\begin{figure}[t]
\centering
\begin{minipage}{0.49\linewidth}
  \centering
  \includegraphics[page=1,width=\linewidth]{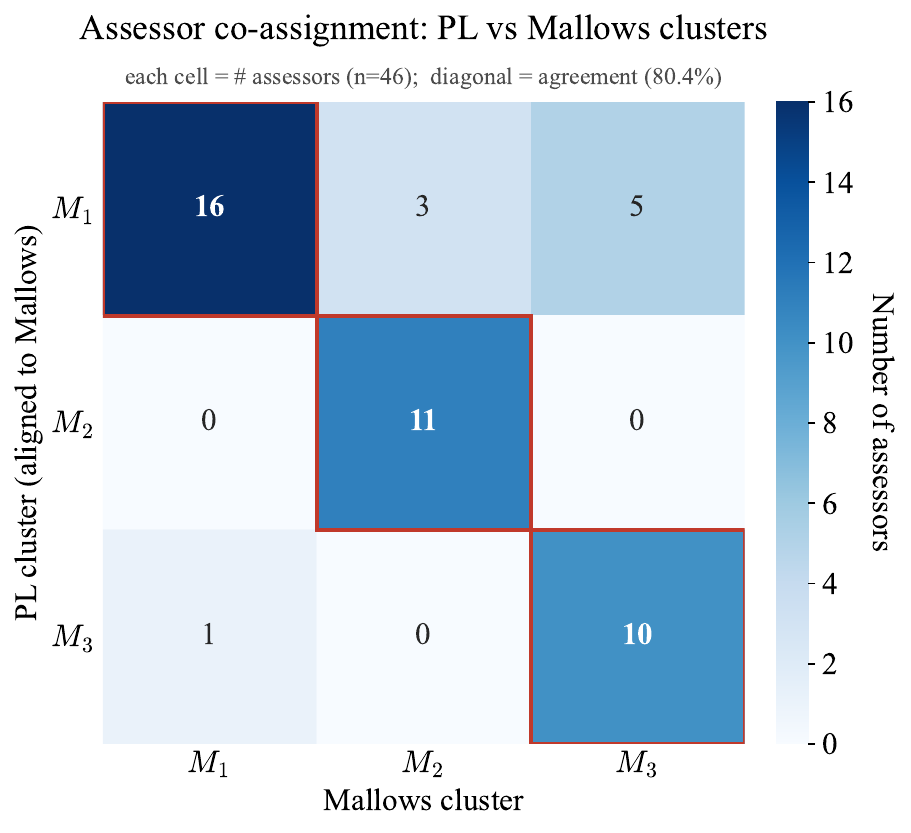}
\end{minipage}\hfill
\begin{minipage}{0.49\linewidth}
  \centering
  \includegraphics[page=2,width=\linewidth]{figures/main_sound_baselines_rank_comparison.pdf}
\end{minipage}
\caption{\textbf{Direct benchmark comparison of aligned clusterwise rankings (PL vs.\ relabeled Mallows).}}
\label{fig:blockc_rank_comparison}
\end{figure}

For HCPO, the weighted queue-jump variant achieves consistently lower WAIC than the log-assessor/Frontier-Softmax variant across the
range of inferred cluster counts shown in Figure~\ref{fig:sound_pairwise_waic}. In light of this systematic dominance
(and the stronger MCMC efficiency reported in Appendix Table~\ref{tab:sound_obsmodel_diag}), we adopt the
\emph{weighted queue-jump} observation model for the main sound-data results. Following the baseline analysis of \citet{crispino19}, we fix the number of mixture components to three for
both total-order baselines. The PL mixture baseline ($G=3$) attains $\mathrm{WAIC}=2046.72$ and
$\mathrm{LOOIC}=2047.04$. The Mallows mixture baseline ($K=3$) yields substantially larger values
($\mathrm{WAIC}=3182.48$, $\mathrm{LOOIC}=2969.93$) and exhibits very weak PSIS diagnostics
(775/1380 observations with Pareto-$k>0.7$; 751/1380 with Pareto-$k>1.0$), so LOO-based summaries for Mallows
should be interpreted cautiously. For HCPO, the weighted queue-jump likelihood achieves $\mathrm{WAIC}=1593.52$,
improving on the log-assessor/Frontier-Softmax variant ($\mathrm{WAIC}=1628.44$). PSIS--LOO is unstable for the HCPO fits in this
experiment, and we therefore use WAIC (rather than LOOIC) as the primary predictive criterion for HCPO model
comparison.

Due to sparsity and unstable cluster occupancy across draws, we refit the hierarchical model conditional on inferred cluster assignments using four chains
(seeds 42/143/244/345). See figure \ref{fig:multiseed_profile_new}. Convergence and mixing for the refit are summarized in Appendix Table~\ref{tab:sound_refit_restart_ess}
(split-$\widehat{R}<1.05$, range $1.001$--$1.039$). For cluster-level Hasse summaries we use the seed 345
refit run, which exhibits strong ESS for key structural and noise-related metrics.

To compare HCPO with the PL and Mallows mixtures, we align HCPO clusters to baseline components by best match
(maximal overlap / closest clusterwise ranking), since HCPO infers $G_{\mathrm{VI}}=5$ clusters whereas the
baselines use three components. Under this alignment (using the seed 345 refit), HCPO exhibits systematic
\emph{directional} shifts in item placement: relative to the PL mixture, items 1/8/3 tend to move upward while
items 12/11/7 tend to move downward; relative to the Mallows mixture, items 8/3/5/9 tend to move upward while
items 11/12/7/4 tend to move downward. Beyond these rank shifts, HCPO highlights partial comparability: several
clusters contain broad top layers (near-indifference sets) together with a small number of clearly inferior stimuli,
a pattern that is explicit in the inferred Hasse diagrams in Figure~\ref{fig:sound_final_inference} but is
necessarily collapsed into strict within-component permutations by the total-order baselines.

\subsection{Experiment E: Ghana sweet--potato consumer--preference dataset}
\label{app:potato_dataset_desc}

We evaluate our model using the Ghana sweet--potato consumer--preference dataset,
a widely studied example of heterogeneous ranking data arising from on--farm varietal evaluation.
The data originate from a large citizen--science initiative applying the Triadic Comparison of Technologies (tricot)
methodology, and are distributed together with the \texttt{SFPL} package.

The Ghana study involved rural consumers in multiple regions, who evaluated boiled
sweet--potato samples from a panel of thirteen breeding lines and released varieties.
Participants tasted three varieties at a time, prepared using standardised protocols, and reported
preference rankings (Win $\succ$ Mid $\succ$ Lose). In addition to the preference orders,
each variety is accompanied by five physicochemical and agronomic covariates: dry matter content,
sweetness, flesh colour, root size, and fibre content. These are used as item--level predictors in downstream
ranking models.

\begin{figure}[t]
\centering
\includegraphics[width=\linewidth]{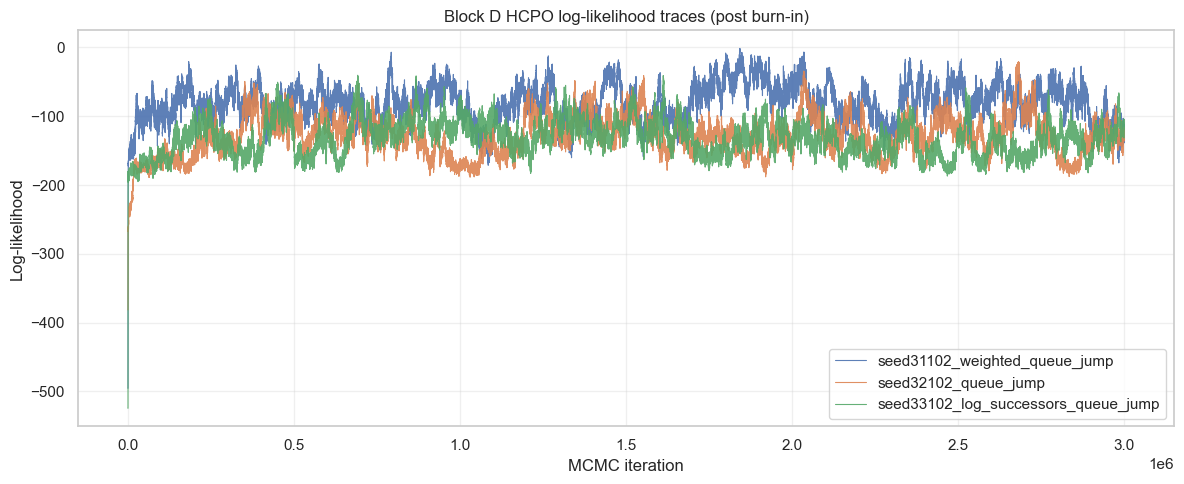}
\caption{\textbf{Ghana (Block D): post burn-in log-likelihood traces.}
Log-likelihood trajectories for three HCPO observation models: weighted queue-jump, queue-jump, and
frontier-softmax.}
\label{fig:ghana_blockd_loglik}
\end{figure}

\begin{figure}[t]
\centering
\includegraphics[width=\linewidth]{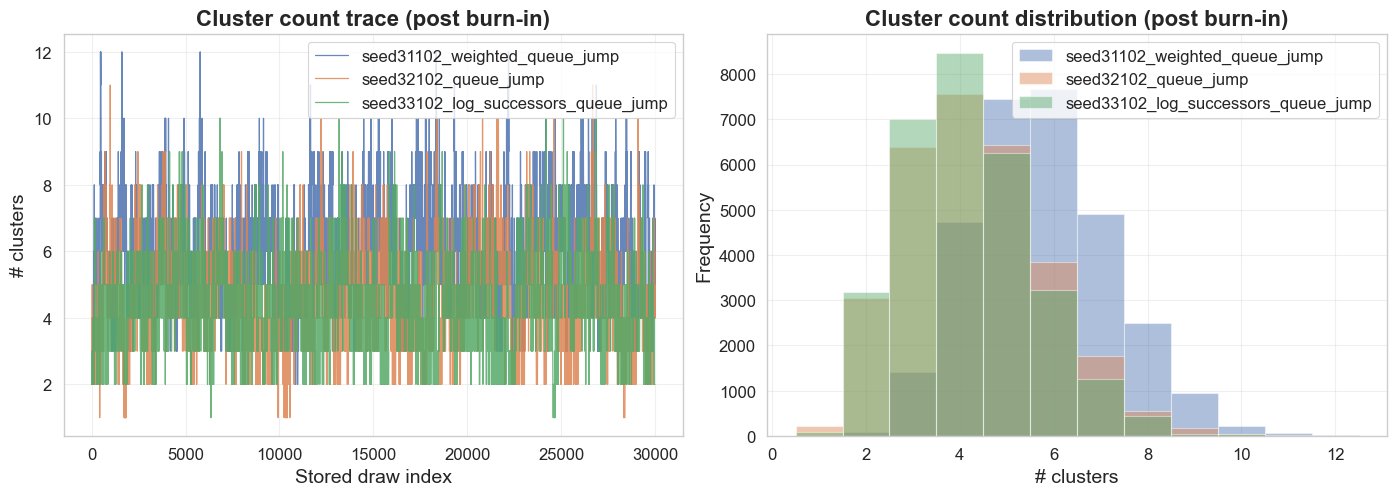}
\caption{\textbf{Ghana (Block D): posterior cluster-count diagnostics.}
(\emph{Left}) trace of the number of clusters across stored post burn-in draws. (\emph{Right}) histogram of the
posterior distribution of the number of clusters.}
\label{fig:ghana_blockd_cluster_count}
\end{figure}

\begin{table}[t]
\centering
\small
\caption{\textbf{Ghana: posterior summaries and MCMC efficiency (post burn-in).}
Posterior mean (SD) and effective sample size (ESS) for key scalar quantities under three observation models.
For log-likelihood, ESS is computed on a thinned subtrace (stride shown in the original diagnostics output).}
\label{tab:ghana_blockd_mcmc_diag}
\setlength{\tabcolsep}{6pt}
\renewcommand{\arraystretch}{1.05}
\begin{tabular}{llrrrr}
\toprule
\textbf{Likelihood} & \textbf{Param} & \textbf{Mean} & \textbf{SD} & \textbf{ESS} & \textbf{ESS rel.} \\
\midrule
weighted queue-jump & $K$             & 3.214 & 1.759 & 195.17 & 0.00651 \\
weighted queue-jump & $K_{\log(1+\rho)}$ & 2.163 & 1.217 & 208.08 & 0.00694 \\
weighted queue-jump & depth           & 8.481 & 2.467 & 81.73  & 0.00272 \\
weighted queue-jump & loglik          & -80.50 & 27.50 & 55.43  & 0.00111 \\
weighted queue-jump & prob\_noise     & 0.249 & 0.141 & 71.19  & 0.00237 \\
weighted queue-jump & $\rho$          & 0.944 & 0.117 & 237.95 & 0.00793 \\
weighted queue-jump & $\tau$          & 0.316 & 0.172 & 152.03 & 0.00507 \\
\midrule
queue-jump          & $K$             & 2.838 & 1.689 & 177.96 & 0.00593 \\
queue-jump          & $K_{\log(1+\rho)}$ & 1.843 & 1.152 & 228.63 & 0.00762 \\
queue-jump          & depth           & 8.670 & 2.745 & 52.91  & 0.00176 \\
queue-jump          & loglik          & -128.15 & 28.67 & 33.37 & 0.00067 \\
queue-jump          & prob\_noise     & 0.300 & 0.115 & 40.10  & 0.00134 \\
queue-jump          & $\rho$          & 0.893 & 0.169 & 214.06 & 0.00714 \\
queue-jump          & $\tau$          & 0.431 & 0.230 & 55.21  & 0.00184 \\
\midrule
Frontier-softmax q-j  & $K$             & 2.868 & 1.687 & 102.47 & 0.00342 \\
Frontier-softmax q-j  & $K_{\log(1+\rho)}$ & 1.814 & 1.155 & 184.13 & 0.00614 \\
Frontier-softmax q-j  & depth           & 8.492 & 2.877 & 32.11  & 0.00107 \\
Frontier-softmax q-j  & loglik          & -131.22 & 24.83 & 50.31 & 0.00101 \\
Frontier-softmax q-j  & prob\_noise     & 0.304 & 0.104 & 70.02  & 0.00233 \\
Frontier-softmax q-j  & $\rho$          & 0.878 & 0.205 & 50.56  & 0.00169 \\
Frontier-softmax q-j  & $\tau$          & 0.461 & 0.236 & 37.57  & 0.00125 \\
\bottomrule
\end{tabular}
\end{table}

\begin{figure}[h]
\centering
\includegraphics[width=0.8\linewidth]{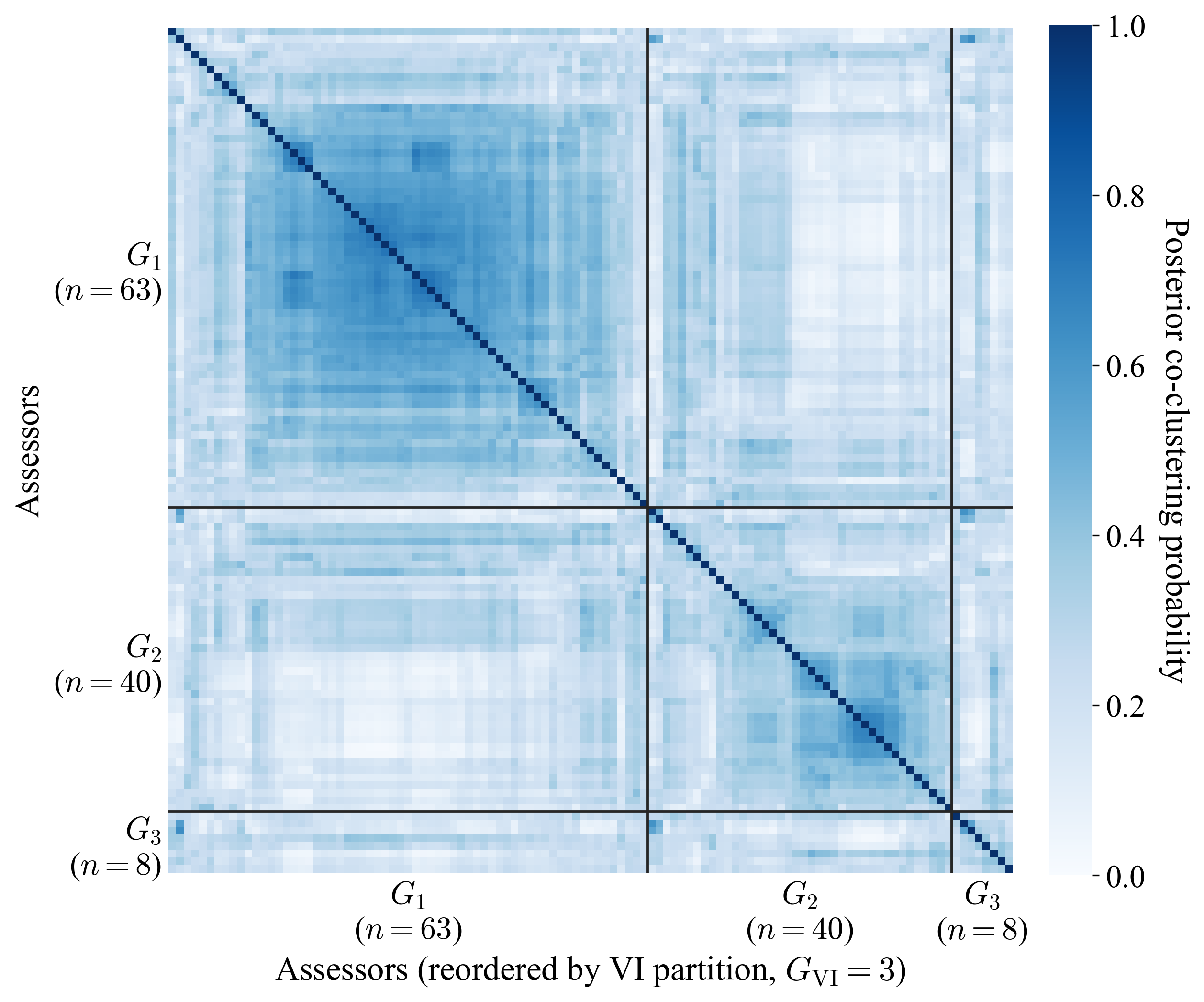}
\caption{\textbf{Ghana tricot: posterior similarity matrix (PSM) under weighted queue-jump.}
Assessors are reordered according to the VI-optimal partition with $G_{\mathrm{VI}}=3$ and cluster sizes
$(63,40,8)$ ; darker cells indicate higher posterior co-clustering
probability.}
\label{fig:app_ghana_psm_weighted_vi3}
\end{figure}

\begin{table}[t]
\centering
\small
\caption{\textbf{Ghana: MCMC efficiency for fixed-partition HPO refits.}
Posterior mean, standard deviation, and effective sample size (ESS) for key scalar parameters, computed from post-burn-in draws. The refits use the VI-selected cluster count, namely \(G=1\) for the frontier-softmax and queue-jump models, and \(G=3\) for the weighted queue-jump model.}
\label{tab:ghana_refit_ess}
\setlength{\tabcolsep}{6pt}
\renewcommand{\arraystretch}{1.05}
\begin{tabular}{l l l r r r r}
\toprule
\textbf{Likelihood} & \textbf{Refit model} & \textbf{Param} & $n_{\text{post}}$ & \textbf{Mean} & \textbf{SD} & \textbf{ESS} \\
\midrule
frontier-softmax & HPO refit (\(G=1\)) & loglik      & 2,985,000 & $-176.08$ & 3.31  & 2,399.26 \\
frontier-softmax & HPO refit (\(G=1\)) & prob\_noise & 15,000    & 0.520     & 0.064 & 58.64 \\
frontier-softmax & HPO refit (\(G=1\)) & $\rho$      & 15,000    & 0.906     & 0.128 & 216.80 \\
frontier-softmax & HPO refit (\(G=1\)) & $\tau$      & 15,000    & 0.505     & 0.279 & 836.76 \\
\midrule
queue-jump                  & HPO refit (\(G=1\)) & loglik      & 2,985,000 & $-176.55$ & 3.32  & 1,558.95 \\
queue-jump                  & HPO refit (\(G=1\)) & prob\_noise & 15,000    & 0.502     & 0.079 & 85.85 \\
queue-jump                  & HPO refit (\(G=1\)) & $\rho$      & 15,000    & 0.893     & 0.150 & 157.51 \\
queue-jump                  & HPO refit (\(G=1\)) & $\tau$      & 15,000    & 0.536     & 0.298 & 345.22 \\
\midrule
weighted queue-jump         & HPO refit (\(G=3\)) & loglik      & 2,985,000 & $-99.47$  & 5.31  & 650.79 \\
weighted queue-jump         & HPO refit (\(G=3\)) & prob\_noise & 15,000    & 0.311     & 0.067 & 388.64 \\
weighted queue-jump         & HPO refit (\(G=3\)) & $\rho$      & 15,000    & 0.790     & 0.252 & 926.15 \\
weighted queue-jump         & HPO refit (\(G=3\)) & $\tau$      & 15,000    & 0.961     & 0.017 & 118.00 \\
\bottomrule
\end{tabular}
\end{table}

\begin{figure}[t]
\centering
\includegraphics[width=0.8\linewidth]{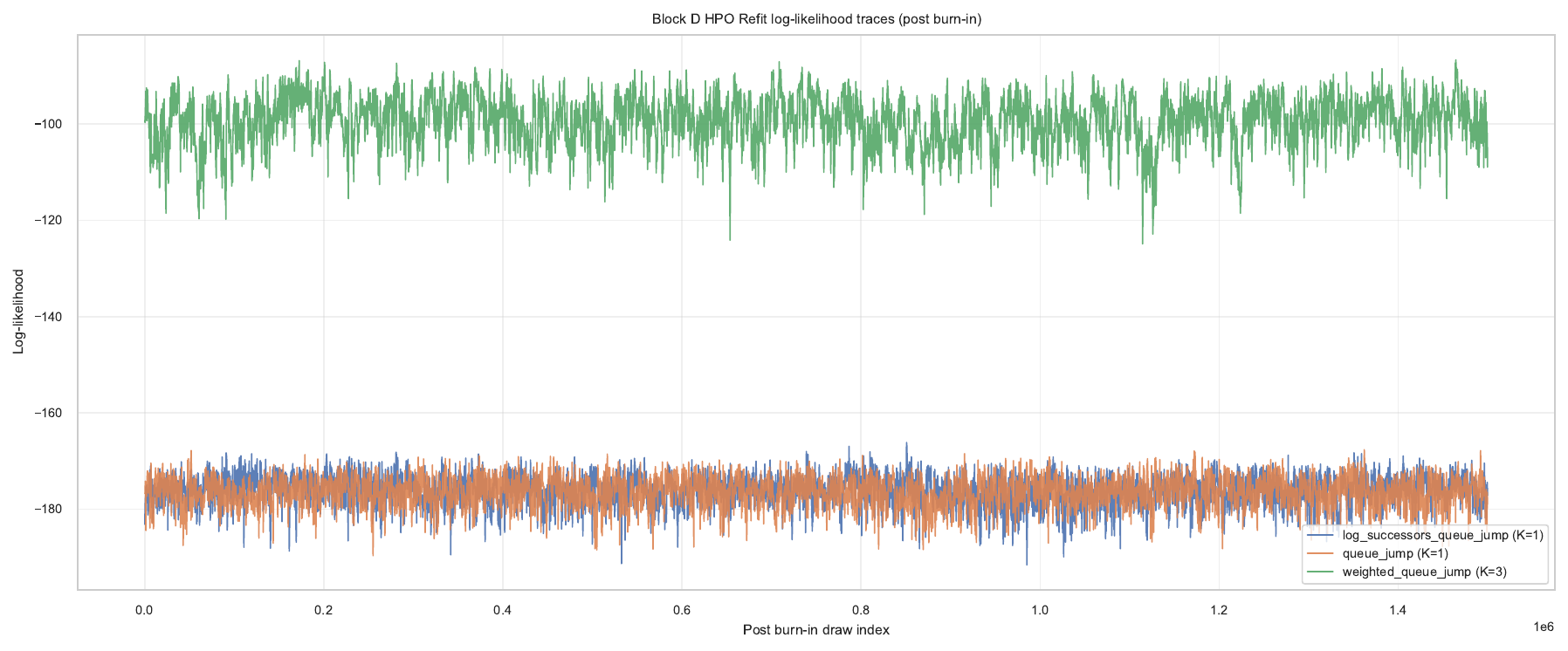}
\caption{\textbf{Ghana refit: post burn-in log-likelihood trace comparison.}
Post burn-in log-likelihood traces for the fixed-$K$ HPO refits under three observation models.
The weighted queue-jump refit exhibits consistently higher log-likelihood values and a stable stationary regime.}
\label{fig:app_ghana_refit_loglik_compare}
\end{figure}


\begin{figure}[h]
\centering
\includegraphics[width=\linewidth]{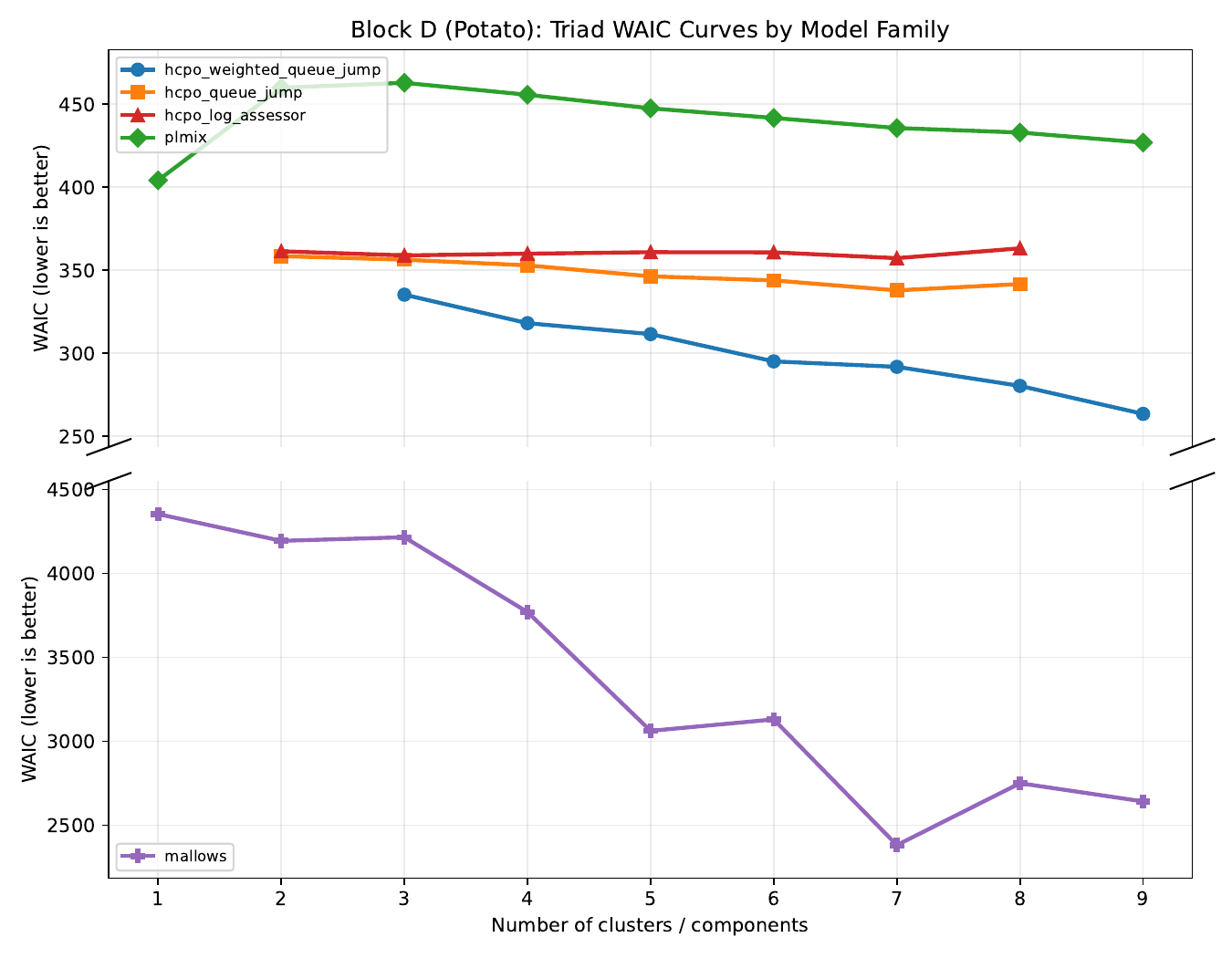}
\caption{\textbf{Ghana (Block D): triad WAIC curves by model family.}
WAIC (lower is better) is computed on the native triad observation unit ($n_{\mathrm{obs}}=111$) and plotted against model complexity. For the PL mixture and Mallows baselines, the horizontal axis is the number of mixture components ($1$--$9$); for HCPO, it is the posterior cluster count $G$ for grouped-by-$G$ WAIC estimates with sufficient draws. The Mallows curve is shown on a separate scale because exact Mallows WAIC, computed from saved posterior draws using the augmented-ranking conditional likelihood rather than the earlier proxy, is much larger than the HCPO and PL mixture values. Across the evaluated range, HCPO achieves substantially lower triad WAIC than the total-order baselines, with the weighted queue-jump variant best overall.}
\label{fig:ghana_triad_waic_curves}
\end{figure}

\begin{table}[t]
\centering
\small
\caption{\textbf{Ghana: fixed-partition HPO refits on the original triad observation unit.}
For each likelihood, we report the VI cluster count used in the refit ($G_{\mathrm{VI}}$) and predictive
criteria computed on the original triad observation unit ($n_{\mathrm{obs}}=111$). All scores use
$n_{\mathrm{draws}}=800$ draws. Lower is better.}
\label{tab:ghana_refit_waic_original}
\setlength{\tabcolsep}{5pt}
\renewcommand{\arraystretch}{1.05}
\begin{tabular}{l c rr}
\toprule
\textbf{Likelihood} & $G_{\mathrm{VI}}$ & \textbf{WAIC} & \textbf{LOOIC} \\
\midrule
frontier-softmax & 1 & 381.01 & 380.75 \\
queue-jump                  & 1 & 371.76 & 371.73 \\
weighted queue-jump         & 3 & 236.49 & 237.19 \\
\bottomrule
\end{tabular}

\vspace{0.25em}
\footnotesize
All rows use the original triad observation unit ($n_{\mathrm{obs}}=111$) and $n_{\mathrm{draws}}=800$.
\end{table}

\section{Fixed-Leaf HPO for the Arcagni Example}
\label{app:arcagni-fixed-leaf}

We consider a fixed-leaf specialization of HPO for the regional well-being example of \citet{arcagni2022complexity}. Unlike the main model, where lower-level data are noisy observation lists or task-based traces, the data here are themselves partial orders. Specifically, we treat the three domain-specific posets as exact observed leaves on a common ground set of $n=20$ Italian regions, with no list likelihood, no task structure, no covariates, and no additional noise parameter. The latent dimension is fixed at $K=n/2=10$, so the stochastic unknowns are
\[
(U^0,U^1,U^2,U^3,\tau,\rho),
\]
where $U^0$ is the population-level latent matrix, $U^1,U^2,U^3$ are the leaf-level latent matrices, and the inferential target is the root poset
\[
h^0 = h(U^0).
\]

Let $P^1,P^2,P^3 \in \{0,1\}^{n\times n}$ denote the observed leaf posets, represented by their transitive closures. $h(U)$ is the intersection of the $K$ total orders induced by the columns of $U$. In the present appendix the leaf posets are imposed exactly,
\[
h(U^a)=P^a,\qquad a=1,2,3,
\]
and the posterior takes the form
\[
\pi(U^0,U^1,U^2,U^3,\tau,\rho \mid P^1,P^2,P^3)
\propto
p(\rho)\,p(\tau)\,p(U^0\mid \rho)
\prod_{a=1}^3 p(U^a\mid U^0,\tau,\rho)\,
\mathbf{1}\{h(U^a)=P^a\}.
\]
Thus the leaf data enter only through hard constraints, and there is no separate observation-order likelihood. 
The sampler alternates constrained leaf updates, Gaussian updates of $U^0$ given the current leaves, and Metropolis updates of $\rho$ and $\tau$. Each leaf update combines an order-changing linear-extension proposal with a constrained Gibbs sweep, and any proposal that changes the induced leaf poset is rejected immediately. Since $K$ is fixed throughout, there are no reversible-jump, birth, or death moves.

Figure~\ref{fig:arcagni-fixed-leaf-data} shows the three observed leaf posets used in this experiment: personal economic situation, health and family relations, and leisure friendship. Their closures contain $177$, $116$, and $114$ strict comparabilities, with $13$, $74$, and $76$ incomparable pairs, respectively. As external references, \citet{arcagni2022complexity} report an optimal approximating poset with $120$ strict comparabilities and $70$ incomparabilities, a unique top region TRE, and minimal regions APU, BAS, CAL, CAM, MOL, SAR, and SIC. We also reconstruct the bucket-order solution shown in their Figure~16 and the final linear order reported in their Table~5.

\begin{figure}[t]
    \centering
    \includegraphics[width=\textwidth]{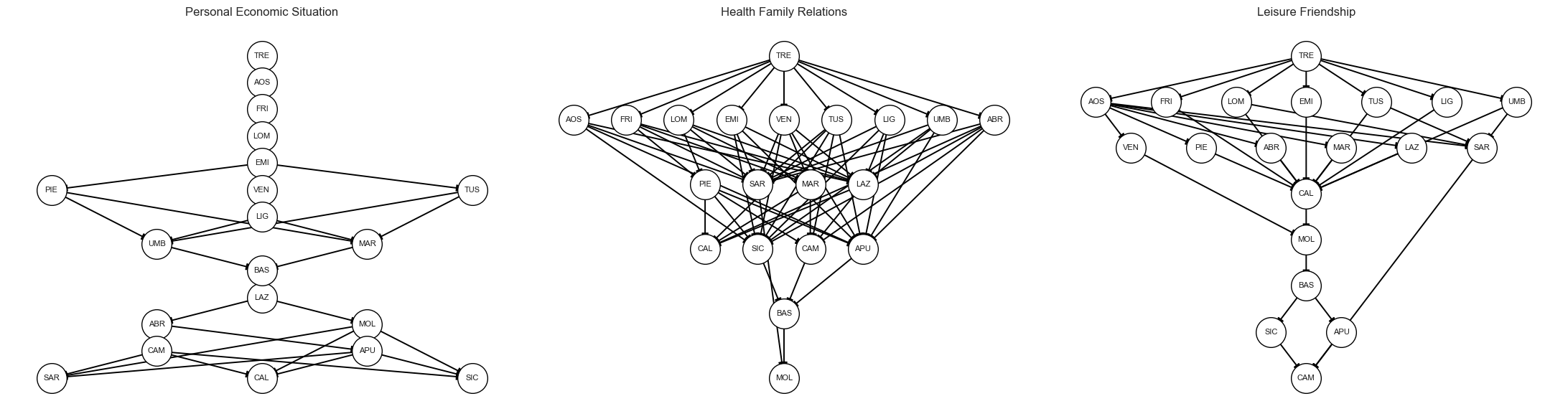}
    \caption{Observed leaf posets in the fixed-leaf Arcagni analysis.}
    \label{fig:arcagni-fixed-leaf-data}
\end{figure}

Figure~\ref{fig:arcagni-fixed-leaf-result} summarizes the corresponding fixed-leaf HPO output. The posterior means were $\rho=0.974$ and $\tau=0.766$. The posterior-mode root poset had $152$ strict comparabilities and $49$ cover relations. Relative to the reconstructed bucket-order closure, it shared $113$ strict comparabilities, corresponding to a comparability Jaccard overlap of $0.585$. Relative to the reconstructed final linear-order closure, it shared $141$ strict comparabilities, with Jaccard overlap $0.701$. Thus the inferred root is more ordered than the optimal approximating poset reported by \citet{arcagni2022complexity}; it recovers the same unique top region, TRE, but does not reproduce the reported minimal-region set.

\begin{figure}[t]
    \centering
    \includegraphics[width=\textwidth]{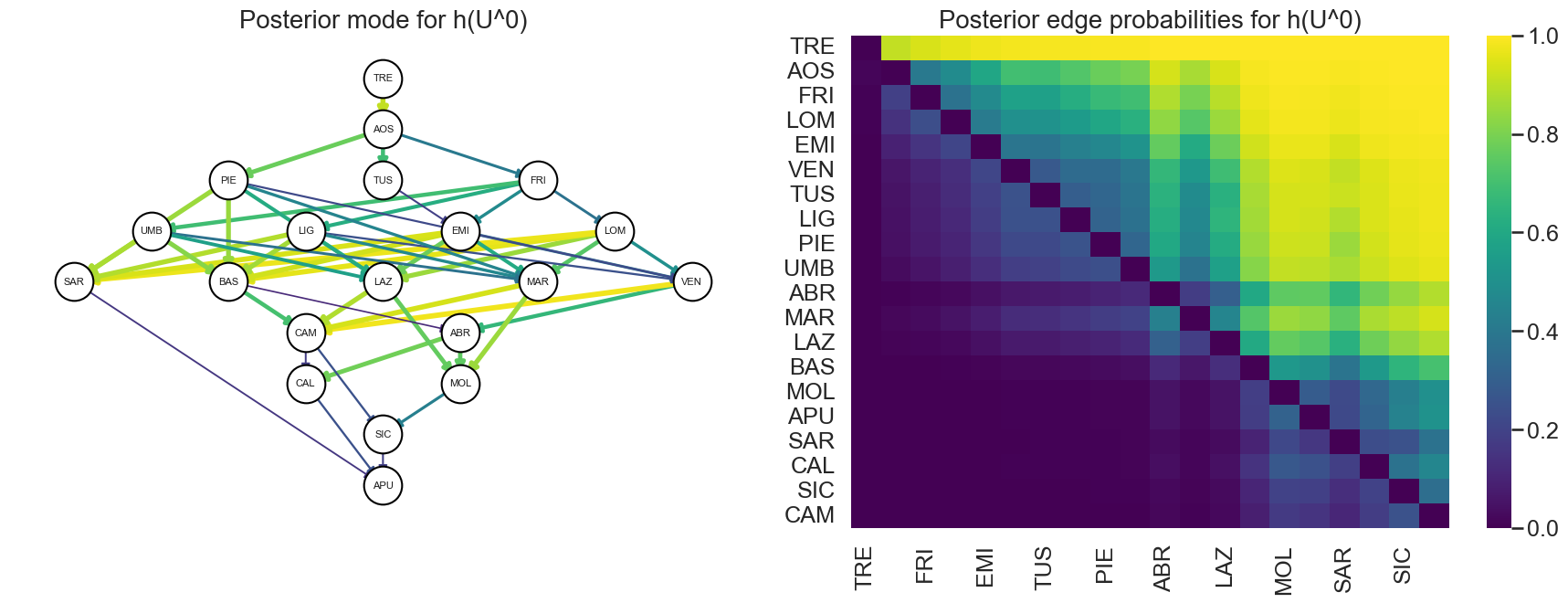}
    \caption{Posterior mode of $h(U^0)$ (left) and edgewise posterior probabilities for $h(U^0)$ (right).}
    \label{fig:arcagni-fixed-leaf-result}
\end{figure}

The posterior over root structures remains diffuse: the empirical posterior mode appeared only once among the $19{,}950$ stored root draws. This example should therefore be interpreted as a Bayesian hierarchical analogue of the Arcagni problem rather than as a literal reproduction of its greedy mutual-ranking-probability objective.

\section{Introduction-table}\label{sec:intro-table-dataset-sizes}
Table~\ref{tab:intro-data-summary} summarizes the sizes and structures of the representative ranking-data settings discussed below. Taken together, these examples illustrate that HPO is intended for a broad range of grouped rank data, from assessor-based studies and peer grading to historical witness lists, biological rankings, and trace-derived rankings.

\begin{table}[t]
\centering
\footnotesize
\setlength{\tabcolsep}{4pt}
\renewcommand{\arraystretch}{1.05}
\caption{Representative ordered-data settings discussed in the introduction.
Here $M$ denotes the number of items (or nodes), $A$ denotes the natural grouping
structure when present, and \#lists denotes the number of observed rankings/lists. The \citet{Johnson02MonkeyJasa} data has ties which we do not consider. The \citet{nicholls25AOAS} data is a timeseries so it takes a hierarchical model of a sort not considered here. We analyse datasets from the papers in bold.}
\label{tab:intro-data-summary}
\begin{tabular}{
    >{\raggedright\arraybackslash}p{3.5cm}
    c
    c
    >{\raggedright\arraybackslash}p{2.7cm}
    >{\raggedright\arraybackslash}p{1.6cm}
    >{\raggedright\arraybackslash}p{1.4cm}}
\toprule
Reference & Hier. & $M$ & Grouping & \#lists & Length \\
\midrule
\citet{Johnson02MonkeyJasa} & Yes & 24 & 9 paradigms / 30 procedures & 30 & $\le 12$ \\
\citet{GormleyMurphy2009} & No & 5 & -- & 2498 & 1--5 \\
\citet{nagy13pigeonPO} & Yes & 10/30 & 3 & unclear & 2/unclear\\ 
\citet{raman14crowdGrade} & No & 42 & -- & 996 & 6.73 avg. \\
\bf{\citet{crispino19}} & Yes & 12 & 46 assessors & 1380 & 2 \\
\bf{\citet{arcagni2022complexity}}  & Yes & 20 & 3 local posets & 0 & n/a \\
\citet{nicholls25AOAS} & Yes & 67 & 76 years & 371 & 2--15 \\
\bf{\citet{li26delinearizing}} & Yes & 15 & 6 & 60 & 5-12\\
\bottomrule
\end{tabular}

\vspace{2pt}
\begin{minipage}{0.94\linewidth}
\footnotesize
\textit{Note.} ``Hier.'' indicates whether the data have an explicit grouped or multi-level
structure relevant for HPO. The Arcagni row refers to the fixed-leaf appendix example,
where the observations are posets rather than ranked lists.
\end{minipage}
\end{table}
\end{document}